\documentclass[12pt,letterpaper,leqno,doublespacing]{article}
\pdfoutput=1

\usepackage{setspace}
\usepackage{footmisc}

\usepackage[labelfont={bf,small},textfont={bf,small},format=hang]{caption}
\makeatletter
\renewcommand\@seccntformat[1]{\csname the#1\endcsname.\quad}
\makeatother

\renewcommand{\thesubsection}{\Alph{subsection}}
\renewcommand{\thesubsubsection}{\Alph{subsection}.\arabic{subsubsection}}

\usepackage{graphicx,epstopdf,amsmath,amsfonts,amssymb,amsthm,versionPO}
\usepackage{marginnote,datetime,enumitem, rotating,fancyvrb}

\usepackage{float}

\usepackage[dvipsnames]{xcolor}

\usepackage[hidelinks]{hyperref}

\usepackage[longnamesfirst]{natbib}
\usdate

\usepackage{booktabs}
\usepackage{pdflscape}
\usepackage{chngcntr}
\usepackage{longtable}
\usepackage{subfig}
\usepackage{adjustbox}
\usepackage{multirow}
\usepackage{bm}

\usepackage{tabularx}
\usepackage{threeparttable}
\usepackage{tabu}
\usepackage{makecell}
\usepackage{siunitx}
\renewcommand{\arraystretch}{0.8}

\excludeversion{notes}
\includeversion{links}

\iflinks{}{\hypersetup{draft=true}}

\ifnotes{%
	\usepackage[margin=1in, top=1.5in, bottom=1.5in]{geometry}%
}{%
	\usepackage[margin=1in]{geometry}%
	\usepackage[disable]{todonotes}%
}

\makeatletter\let\chapter\@undefined\makeatother

\newtheorem{proposition}{Proposition}
\newtheorem{theorem}{Theorem}
\newtheorem{lemma}[theorem]{Lemma}
\newtheorem{example}[theorem]{Example}
\newtheorem{remark}[theorem]{Remark}

\begin{document}

\setlist{noitemsep}
	\onehalfspacing

	\title{{\large\bf
    When and Why Na\"{\i}ve Diversification Works: \\ A Simple Diagnostic Strategy}\thanks{Email addresses: Feng, fenghan@amss.ac.cn; Huang, difang@amss.ac.cn; Wang, wjue@amss.ac.cn; Zhang, zjz@stat.wisc.edu. The authors thank the comments from seminar participants at the 45th International Symposium on Forecasting and the6th International Symposium on Interval Data Modelling: Theory and Applications. Financial support from the National Natural Science Foundation of China is gratefully acknowledged by Difang Huang (72503232, 72574227, T2293771), Jue Wang (72271229, 71988101), and Zhengjun Zhang (71991471, 72442027). Jue Wang also acknowledges support from the National Social Science Foundation of China (22VRC056).}} 

	\author{{\bf Han Feng} \\
        Academy of Mathematics and Systems Science, Chinese Academy of Sciences \\
        School of Economics and Management, University of Chinese Academy of Sciences \\
        AMSS Center for Forecasting Science, Chinese Academy of Sciences \\
        \vspace{0.05in} \\
		{\bf Difang Huang} \\
        Academy of Mathematics and Systems Science, Chinese Academy of Sciences \\
        School of Economics and Management, University of Chinese Academy of Sciences \\
        AMSS Center for Forecasting Science, Chinese Academy of Sciences \\
        \vspace{0.05in} \\
        {\bf Jue Wang} \\
        Academy of Mathematics and Systems Science, Chinese Academy of Sciences \\
        School of Economics and Management, University of Chinese Academy of Sciences \\
        AMSS Center for Forecasting Science, Chinese Academy of Sciences \\
        \vspace{0.05in} \\
        {\bf Zhengjun Zhang} \\
        School of Data Science and AI, Beijing University of Chinese Medicine \\
        AMSS Center for Forecasting Science, Chinese Academy of Sciences\\
        Department of Statistics, University of Wisconsin
        \vspace{0.05in}}

    \date{}
	\maketitle
	\thispagestyle{empty}

    \clearpage
    \thispagestyle{empty}

    \begin{center}
    {\large\bf When and Why Na\"{\i}ve  Diversification Works: \\
    A Simple Diagnostic Strategy}
    \end{center}
    \vspace{0.25in}

    \begin{abstract} 
    We explain the long-standing puzzle of na\"{\i}ve diversification with a simple, testable condition: equal weighting is minimum-variance optimal when the forecast-error covariance matrix has a uniform eigenstructure. This ``Golden Criterion'' drives a two-stage adaptive strategy that dynamically blends naive and optimized weights based on the empirical distance from this condition. Applied to U.S. equity premium forecasting, the method delivers consistent out-of-sample gains in forecast accuracy, utility, and Sharpe ratios. Diversity-driven shrinkage dominates at short horizons, while optimized weights regain their edge at longer horizons, offering clear horizon-dependent guidance for portfolio construction.

    \end{abstract}

	\medskip 



	\clearpage

    \clearpage
	\pagenumbering{arabic}
    \pagestyle{plain}

	\doublespacing

\clearpage

\section{Introduction}

A central challenge in empirical finance is how to combine forecasts from competing predictive models. When multiple models generate return predictions, the forecaster must choose combination weights, and this choice involves a fundamental tension. Estimated optimal weights exploit cross-model differences in accuracy and correlation but are notoriously unstable: small perturbations in the estimated error covariance matrix can produce large swings in the implied allocation \citep{smith2009simple, claeskens2016forecast}. Equal (``na\"{\i}ve'') weights avoid this instability entirely but ignore all information in the error covariance structure. The empirical evidence that equal weighting often outperforms optimized alternatives, documented by \citet{demiguel2009optimal} for portfolio allocation and by \citet{stock2004combination} and \citet{genre2013combining} for macroeconomic forecasting, has motivated a large literature seeking to understand when and why simplicity performs well \citep{smith2009simple,claeskens2016forecast,yuan2024naive}. Yet the answer to the prior question, what structural property of the forecast errors determines whether equal weighting or optimization is the better strategy, has remained an open question.

Two lines of research bear on this question. The first derives conditions under which equal weighting is exactly or approximately optimal.\footnote{{\citet{yuan2024naive} proves this in a 1-factor model with diversifiable risk, as dimensionality increases.}} The second line develops strategies that blend optimized and na\"{\i}ve weights.\footnote{{\citet{tu2011markowitz} demonstrate that linearly combining the Markowitz mean--variance portfolio with the $1/N$ rule yields superior out-of-sample performance; \citet{blanc2020bias} derive a shrinkage procedure for forecast combinations; \citet{roccazzella2022optimal} propose constrained optimization with shrinkage of covariance matrix; \citet{yuan2024naive} explore conditions where 1/N rule can be beaten by combining it with estimated rules in terms of Sharpe ratio; and \citet{li2025dominate} introduce a method that dominates the historical average.}} These combination strategies deliver practical gains, but they do not explain what feature of the covariance environment governs the optimal degree of shrinkage toward equal weights. Can this feature be identified, tested empirically, and exploited to construct an adaptive combination rule?

This paper shows that the answer lies in the eigenstructure of the forecast error covariance matrix. We offer an explanation for a long-standing puzzle in portfolio management: the persistent empirical success of na\"{\i}ve diversification. We call this explanation the \textit{``Golden Criterion''}: the minimum-variance optimal combination coincides exactly with equal weighting when the error covariance matrix admits the all-ones vector as an eigenvector. While the criterion is algebraically straightforward, it is missed in the literature, and as such, formally defining it becomes necessary, even though we may consider it theoretically self-evident. Its contribution lies in the three roles it plays within our analysis. First, it unifies the disparate optimality conditions in the prior literature under a single, model-free, empirically testable covariance property. Second, it extends to a two-stage setting that is new to the literature: when the first-stage (train period) covariance matrix admits an all-ones vector as an eigenvector, the first-stage optimal weights equal equal weights, and the adaptive shrinkage path collapses to equal weights for any second-stage (validation period) covariance matrix. The reverse also theoretically holds, and both properties are formally established as necessary and sufficient empirical conditions. Third, and most importantly for practice, the criterion serves as the diagnostic core of our two-stage adaptive shrinkage method, governing when and how aggressively to tilt toward equal weights. Empirically, we identify a broad set of real-data covariance environments in which the Golden Criterion is approximately satisfied, concentrated at short forecast horizons where estimation noise dominates persistent signals, confirming that the condition has genuine empirical content and is not merely a theoretical possibility.

Building on this structural characterization, we propose two-stage adaptive shrinkage weights (2S--ASWs, hereafter ``adaptive shrinkage'' for conciseness) that combine estimated optimal weights and equal weights through a data-driven shrinkage parameter ($\lambda_{opt}$). Unlike the fixed combination of \citet{tu2011markowitz}, $\lambda_{opt}$ is determined endogenously by minimizing the combined forecast's mean squared error through the accuracy--diversity decomposition \citep{krogh1994neural, brown2005diversity,thomson2019combining}. The resulting objective is quadratic in the shrinkage parameter and admits a closed-form solution. When the Golden Criterion is approximately satisfied, $\lambda_{opt}$ approaches unity and the method prescribes near-complete shrinkage toward equal weights. As the covariance structure departs from this condition at longer horizons, $\lambda_{opt}$ declines and the method recovers more of the optimized allocation. The method thus nests pure equal weighting and pure optimal weighting as boundary cases and generalizes existing combination approaches \citep{aiolfi2006persistence, tu2011markowitz,blanc2020bias} while directly addressing the instability that undermines many optimized portfolios \citep{bekaert2002dating, karavias2023structural}.

We apply this approach to U.S. equity premium forecasting using 13 macroeconomic and financial predictors from the \citet{goyal2024comprehensive} dataset over 1996--2019. We benchmark 2S--ASWs against a comprehensive set of alternatives: pure optimal weights, equal weights, Lasso, Ridge, and Elastic Net, PCA regressions, and recent combination methods including \citet{rapach2010out} (median, trimmed mean, DMSPE), \citet{blanc2020bias}, and \citet{li2025dominate}. We document three main results. First, 2S--ASWs generate positive out-of-sample $R^2$ at the one-month horizon (monthly horizon index $J=1$) across all sample-splitting specifications, reaching as high as $5.41\%$, even when both pure optimal weights and equal weights individually produce negative $R^2$ at $J=1$. That the combination outperforms both constituents illustrates the economic content of the accuracy--diversity trade-off. At $J=3$ and $J=6$, several benchmarks are competitive, and no single method dominates, whereas at $J=12$ and $J=24$, bias-corrected and recency-weighted optimal weights (\citet{blanc2020bias}, OWs($\bm{\Sigma}_2$)) regain the lead. Second, a mean--variance investor using 2S--ASW forecasts realizes annualized utility gains exceeding 200 basis points relative to the historical mean benchmark, with Sharpe ratios up to~1.27. These improvements are statistically significant under the \citet{clark2007approximately} MSFE-adjusted test and the \citet{harvey1998tests} encompassing test. Third, the optimal shrinkage parameter $\lambda_{opt}$ clusters near unity at short horizons and declines as the horizon lengthens, mirrored by the convergence and subsequent divergence of the largest eigenvalue and the normalized covariance sum implied by the Golden Criterion. Furthermore, $\lambda_{opt}$, which equals or closely approaches 1, is frequently observed in real-world datasets, confirming that the Golden Criterion is not a theoretical idealization but a condition that real covariance structures can approximate. Updating the predictors to December~2024, 2S--ASWs continue to outperform at $J=1$, though overall out-of-sample gains attenuate, consistent with \citet{denk2024predicting}. In a reversed-sample exercise, 2S--ASWs exhibit consistent improvement, especially at $J=1$, while Ridge regression shows pronounced improvement at $J=24$, indicating that the two-stage adaptive approach enhances even existing regularization techniques at long horizons.

This paper connects and contributes to several areas of literature. First, we contribute to the forecast combination literature. Since the seminal work of \citet{bates1969combination,clemen1989combining}, researchers have wrestled with the ``forecast combination puzzle'' that a simple equal-weighted average of forecasts frequently outperforms sophisticated, estimated weighting schemes \citep{stock2004combination}. While recent advances have sought to address this by introducing shrinkage techniques to forecast combinations \citep{blanc2020bias, roccazzella2022optimal}, they often rely on numerical cross-validation that obscures the underlying economic mechanism. We build on this literature by demonstrating that the eigenstructure of the error covariance matrix contains directly exploitable, structural information about the optimal combination strategy. Rather than treating shrinkage as a purely statistical regularization tool, we leverage the accuracy--diversity decomposition to derive a principled, closed-form shrinkage parameter. This approach explicitly provides a way to balance individual predictive precision and forecast diversity, allowing the optimal weighting scheme to emerge naturally from the geometry of the forecast errors.

Second, we extend the literature on portfolio theory and asset allocation. A central paradox in modern portfolio choice is that while mean-variance optimization is theoretically sound \citep{markowitz1952modern}, the na\"{\i}ve $1/N$ heuristic often delivers superior out-of-sample performance due to pervasive estimation noise \citep{demiguel2009optimal}. This led to a search for robust ``plug-in'' estimators, such as the Bayesian shrinkage of \citet{jorion1986bayes} or the norm-constrained portfolios of \citet{brodie2009sparse}. To reconcile these paradigms, \citet{tu2011markowitz} proposed linearly combining the na\"{\i}ve and sophisticated portfolios. We extend this combination approach by moving beyond static, ex-ante fixed allocations. We provide a data-adaptive rule whose shrinkage intensity responds endogenously to the prevailing covariance environment. By linking the shrinkage mechanism to the empirical distance from our established Golden Criterion, the framework allows investors to dynamically transition between the robust na\"{\i}ve prior and the optimized allocation as market conditions and signal-to-noise ratios evolve across investment horizons.

Finally, our work provides a formal structural grounding for the behavioral finance literature concerning decision heuristics. The foundational literature on bounded rationality and heuristics \citep{simon1955behavioral, tversky1974judgment} highlights that individuals frequently resort to simplifying rules to manage heavy cognitive loads. Under uncertainty, individuals and practitioners exhibit a well-documented preference for equal-weight allocations, driven by the equiprobability bias and the cognitive appeal of simplicity and fairness \citep{messick1993, benartzi2001, kahneman2011, Cappelen2016}. While traditional finance often regards this $1/N$ heuristic as a behavioral bias that sacrifices theoretical efficiency, we provide a decision-theoretic interpretation for its persistence. Our Golden Criterion identifies the exact covariance environments in which this intuitive behavioral impulse is, in fact, statistically optimal in a minimum-variance sense. Furthermore, our adaptive shrinkage methodology allows us to systematically quantify the economic cost of departing from this heuristic as a function of the forecast horizon, linking behavioral intuition with portfolio theory.

The remainder of the paper proceeds as follows. Section~\ref{sec: method} develops the methodology, derives the optimal shrinkage parameter, and presents the economic interpretation linking the largest eigenvalue, market mode, and the Golden Criterion. Section~\ref{sec: data} describes the data and presents summary statistics. Section~\ref{sec: empirical} presents in-sample and out-of-sample results, including asset allocation and economic evaluation. Section~\ref{sec: Disc} discusses connections to high-dimensional portfolio optimization, behavioral foundations, and horizon-dependent strategy selection. Section~\ref{sec: conclusion} concludes.

\section{Methodology}\label{sec: method}

This study employs the Goyal-Welch (GW) predictors \citep{goyal2024comprehensive} across multiple distinct modeling approaches: predictive regressions using individual GW predictors, forecast combinations using optimal weights (OWs), equal weights (EWs), and adaptive shrinkage weights (2S--ASWs), classical shrinkage methods in the literature such as Lasso, Ridge, and Elastic Net, recently proposed combination benchmark models such as those in \cite{rapach2010out} (Median, Trimmed Mean, DMSPE), \cite{blanc2020bias} and ASWs--\cite{li2025dominate}, and principal component regressions utilizing the first (PCA(1)) and first three (PCA(3)) principal components extracted from the GW predictor set.

\subsection{Forecast Combination Framework}

Following \cite{lin2018forecasting}, we begin with the standard predictive regression framework for equity premium forecasting. For each predictor $i$, we estimate the univariate regression:
\begin{equation}\label{eq:individual regression}
r_{t+1,t+J} = \alpha_{i,J} + \beta_{i,J} x_{it} + \epsilon_{i|t+1,t+J},
\end{equation}
where $r_{t+1,t+J}=(r_{t+1}+\cdots+r_{t+J})/J$ denotes the annualized $J$-month-ahead excess return on the S\&P 500 index relative to the 30-day Treasury bill rate in months $t+1$ through $t+J$, $x_{it}$ represents the $i$th predictor variable at time $t$ ($i=1,2,...,N$; $t=1,2,...,T$), and $\epsilon_{i|t+1,t+J}$ is an error term with zero mean. The individual forecasts are generated as:
\begin{equation}\label{eq:individual regression forecast}
\hat{r}_{t+1,t+J|t,i} = \hat{\alpha}_{i,J} + \hat{\beta}_{i,J} x_{it},
\end{equation}
where $\hat{\alpha}_{i,J}$ and $\hat{\beta}_{i,J}$ denote the OLS coefficients estimated from equation (\ref{eq:individual regression}).

The forecast combination approach aggregates these individual predictions to exploit the potential benefits of model diversification. The combined forecast is constructed as a weighted average of individual forecasts:
\begin{equation}\label{eq:r_hat comb}
\hat{r}_{comb,t+1,t+J|t} = \sum_{i=1}^{N} w_i \hat{r}_{t+1,t+J|t,i},
\end{equation}
where the weights $w_i$ sum to unity. This linear combination framework provides a flexible mechanism for incorporating information from multiple predictors while maintaining computational tractability.

We implement two benchmark weighting schemes that represent polar cases in the forecast combination literature. Equal weights (EWs) assign $w_i = 1/N$ to each predictor, reflecting a model-agnostic approach that treats all predictors symmetrically. This strategy proves particularly robust when individual model performance is highly uncertain or when models exhibit similar predictive accuracy. Optimal weights (OWs) are derived by minimizing the mean squared forecast error (MSFE):
\begin{equation}\label{eq:OWs}
\begin{aligned}
\min_{w} \quad & \frac{1}{T}\sum_{t=1}^{T}(\hat{r}_{comb,t+1,t+J|t}-r_{t+1,t+J})^2 \\
\text{s.t.} \quad & \mathbf{w}^T \mathbf{1} = 1
\end{aligned}
\end{equation}
where we allow negative weights to accommodate potential hedging strategies against systematically biased individual forecasts. The optimal weighting scheme theoretically minimizes forecast error by exploiting cross-model correlations and relative accuracy differences, though it may suffer from estimation uncertainty in finite samples.

\subsection{Two-stage Adaptive Shrinkage Combination Methodology: Reconciling Heuristics and Optimality}\label{sec:shrinkage method}

The adaptive shrinkage combination methodology represents a principled approach to tackling the bias-variance trade-off inherent in forecast combination, which addresses the well-documented instability of optimal weights by linearly interpolating between OWs and EWs using a data-driven shrinkage parameter $\lambda \in [0,1]$:
\begin{equation}\label{eq:SWs}
\hat{\mathbf{w}}^\lambda = \frac{\lambda}{N}\mathbf{1} + (1-\lambda)\hat{\mathbf{w}}^{ow},
\end{equation}
where $\hat{\mathbf{w}}^{ow}$ represents the optimal weights from equation (\ref{eq:OWs}), $\mathbf{1}$ is an $N \times 1$ vector of ones, and $\frac{1}{N}\mathbf{1}$ corresponds to equal weights. This convex combination ensures that adaptive shrinkage weights inherit the desirable properties of both benchmark schemes while mitigating their respective limitations.

The optimal shrinkage intensity is determined through the accuracy-diversity decomposition, which provides analytical insights into the sources of combination forecast performance. For individual forecast $f_{ih}$ (the $i$th forecast at observation $h$) and true value $y_h$, the combined forecast MSFE decomposes as:
\begin{equation}\label{eq:decompose MSFEcomb}
\begin{aligned}
\text{MSFE}_{comb} &= \frac{1}{H}\sum_{h=1}^H\left(\sum_{i=1}^{N} w_i f_{ih} - y_h\right)^2 \\
&= \sum_{i=1}^{N} w_i \text{MSFE}_i - \sum_{i=1}^{N-1} \sum_{j>i}^N w_i w_j \text{Div}_{ij},
\end{aligned}
\end{equation}
where $\text{MSFE}_i$ denotes the mean squared forecast error of the $i$th individual forecast, and $\text{Div}_{ij}$ represents the pairwise diversity between forecasts $i$ and $j$, measured as the mean squared difference between the two forecast series. This decomposition reveals that combination forecast accuracy depends on both the weighted average of individual accuracies and the diversity-weighted interactions between models.

For computational efficiency and analytical tractability, we express this decomposition in matrix notation as:
\begin{equation}\label{eq:MSFEcomb}
\text{MSFE}_{comb} = \mathbf{w}^T\mathbf{S} - \frac{1}{2}\mathbf{w}^T\mathbf{D}\mathbf{w},
\end{equation}
where $\mathbf{S}$ is an $N \times 1$ vector of individual MSFEs and $\mathbf{D}$ is an $N \times N$ diversity matrix with elements $D_{ij} = \text{Div}_{ij}$ for $i \neq j$ and $D_{ii} = 0$.

Substituting equation (\ref{eq:SWs}) into (\ref{eq:MSFEcomb}) yields a quadratic function in $\lambda$:
\begin{equation}\label{eq:plug MSEFcomb}
\begin{aligned}
\text{MSFE}_{comb}(\lambda) &= \left(\frac{\lambda}{N}\mathbf{1} + (1-\lambda)\hat{\mathbf{w}}^{ow}\right)^T\mathbf{S} \\
&\quad - \frac{1}{2}\left(\frac{\lambda}{N}\mathbf{1} + (1-\lambda)\hat{\mathbf{w}}^{ow}\right)^T\mathbf{D}\left(\frac{\lambda}{N}\mathbf{1} + (1-\lambda)\hat{\mathbf{w}}^{ow}\right).
\end{aligned}
\end{equation}

Let $a=\frac{1}{N}\mathbf{1}-\hat{\mathbf{w}}^{ow}$. Since $\mathbf{D}$ is a squared-distance diversity matrix and $a^T\mathbf{1}=0$, we have $a^T\mathbf{D}a\leq 0$. When $a^T\mathbf{D}a<0$, the objective is strictly convex along the shrinkage path, and the first-order condition yields a unique stationary point, the unconstrained optimal shrinkage parameter:
\begin{equation}\label{eq:min lambda}
\lambda^* = \frac{\left(\frac{1}{N}\mathbf{1}^T - (\hat{\mathbf{w}}^{ow})^T\right)\left(\mathbf{S} - \mathbf{D}\hat{\mathbf{w}}^{ow}\right)}{\left(\frac{1}{N}\mathbf{1}^T - (\hat{\mathbf{w}}^{ow})^T\right)\mathbf{D}\left(\frac{1}{N}\mathbf{1} - \hat{\mathbf{w}}^{ow}\right)}.
\end{equation}

To ensure economic interpretability and numerical stability, we constrain the shrinkage parameter to the unit interval:
\begin{equation}\label{eq:restrict opt lambda}
\lambda_{opt} = \max\left\{0, \min\left\{1, \lambda^*\right\}\right\}.
\end{equation}

To elucidate the theoretical foundations of the adaptive shrinkage combination methodology, we examine its optimization structure through a stylized two-stage framework. This analysis provides deeper insights into the conditions under which adaptive shrinkage weights preserve optimality properties. Consider $N$ time series with error covariance matrices $\mathbf{\Sigma}_1$ and $\mathbf{\Sigma}_2$ for two consecutive periods, both assumed symmetric and positive definite.

The first-stage optimization problem determines optimal weights under the initial covariance structure:
\begin{align}
& \underset{\mathbf{w}}{\text{min}} \quad \mathbf{w}^T \mathbf{\Sigma}_1 \mathbf{w} \notag \\
& \text{s.t.} \quad \mathbf{w}^T \mathbf{1} = 1,
\label{eq: 1st optimization}
\end{align}
yielding the solution $\hat{\mathbf{w}}^{ow} = \frac{\mathbf{\Sigma_1^{-1}\mathbf{1}}}{\mathbf{1}^{T}\mathbf{\Sigma_1}^{-1} \mathbf{1}}$.

The second-stage optimization problem addresses the adaptive shrinkage combination under the updated covariance structure:
\begin{align}
& \underset{\lambda \in [0,1]}{\text{min}} \quad (\hat{\mathbf{w}}^{\lambda})^T \mathbf{\Sigma}_2 \hat{\mathbf{w}}^{\lambda} \notag \\
& \text{s.t.} \quad (\hat{\mathbf{w}}^{\lambda})^T \mathbf{1} = 1.
\label{eq: 2nd optimization}
\end{align}
The induced set of feasible combination weights is
\begin{equation*}
\mathcal{S}_{2S-ASW}=\left\{\hat{\mathbf w}^{\lambda}=\lambda\mathbf w^{ew}+(1-\lambda)\hat{\mathbf w}^{ow}\mid \lambda\in [0,1] \right\}.
\end{equation*}

Intuitively, when the risk structure of assets is perfectly aligned, so that a uniform exposure to all assets does not increase portfolio variance, the equal-weight strategy ($\lambda=1$ in equation (\ref{eq:SWs})) coincides with the optimal solution. This alignment mathematically corresponds to the error covariance matrix having an all-ones eigenvector, which we term the Golden Criterion.

We establish three key propositions that characterize the relationship between optimal and adaptive shrinkage weights under different assumptions:

\begin{proposition}
$\hat{\mathbf{w}}^{\lambda}=\hat{\mathbf{w}}^{ow} = \mathbf{w}^{ew}$ for any $\lambda\in [0,1]$ if and only if $\mathbf{1}$ is an eigenvector of $\mathbf{\Sigma}_1$ (or equivalently of $\mathbf{\Sigma}_1^{-1}$).
\label{proposition 1}
\end{proposition}

Proposition \ref{proposition 1} provides a theoretically necessary and sufficient condition for the structure of matrix $\mathbf{\Sigma}_1^{-1}(\mathbf{\Sigma}_1)$ under a special case where the optimal weights estimated from problem~(\ref{eq: 1st optimization}) are exactly equal to equal weights. Under this scenario, the adaptive shrinkage weights will always be $\hat{\mathbf{w}}^{ow}$ or $\mathbf{w}^{ew}$ for all $\lambda\in [0,1]$.

Proposition~\ref{proposition 1} gives a necessary and sufficient condition under which the first-stage optimal weights coincide with equal weights. In this case, the shrinkage path collapses to $\mathbf{w}^{ew}$ for all $\lambda\in[0,1]$.

The following two propositions are stated under the premise that $\hat{\mathbf{w}}^{ow}\neq \mathbf{w}^{ew}$.

\begin{proposition}
Under covariance stationarity ($\mathbf{\Sigma}_1 = \mathbf{\Sigma}_2 = \mathbf{\Sigma}$), the second-stage optimum is attained at $\lambda_{opt}=0$, and therefore adaptive shrinkage weights coincide with optimal weights: $\hat{\mathbf{w}}^{\lambda} = \hat{\mathbf{w}}^{ow}$.
\label{proposition 3}
\end{proposition}

Proposition \ref{proposition 3} provides theoretical validation for the adaptive shrinkage methodology under stationary conditions. When the underlying covariance structure remains constant across periods, the optimal shrinkage parameter equals zero, indicating that no shrinkage toward equal weights is necessary. This result aligns with the intuition that optimal weights should remain optimal when the data-generating process is stable, thereby confirming the methodology's theoretical coherence.

\begin{proposition}
Under covariance non-stationarity ($\mathbf{\Sigma}_1 \neq \mathbf{\Sigma}_2$), the unconstrained optimal shrinkage parameter satisfies $\lambda^*=0$ if and only if
\begin{equation}
\mathbf{1}^T\mathbf{\Sigma}_1^{-1}\mathbf{\Sigma}_2
\left(\mathbf{I}-\frac{N}{m}\mathbf{\Sigma}_1^{-1}\right)\mathbf{1}=0,
\end{equation}
where $m=\mathbf{1}^T\mathbf{\Sigma}_1^{-1}\mathbf{1}$.
\label{proposition 4}
\end{proposition}

Proposition \ref{proposition 4} reveals the economically significant insight of our theoretical analysis. Even when the forecasting environment undergoes structural changes, reflected in different covariance matrices, the optimal combination weights can remain unchanged under specific algebraic conditions. This finding has profound implications for forecast combination practice, suggesting that certain types of regime changes may not necessitate weight adjustments. The condition essentially requires that the transformation between pre- and post-break covariance structures preserves the relative ranking and relationships embedded in the optimal weight vector. This result provides a theoretical foundation for understanding when forecast combinations exhibit robustness to structural breaks. The detailed proofs of these propositions, along with illustrative examples demonstrating non-trivial cases where Propositions \ref{proposition 3} and \ref{proposition 4} hold, are provided in the Appendix.

\subsection{Economic Interpretation}

The adaptive shrinkage combination methodology addresses the fundamental accuracy-diversity trade-off in forecast combination. The parameter $\lambda$ governs the relative emphasis on equal weights versus optimal weights, reflecting the competing forces of model accuracy and forecast diversity. When $\lambda > 0.5$, the combination favors equal weights, prioritizing diversity over individual model accuracy; conversely, when $\lambda < 0.5$, optimal weights dominate, emphasizing accuracy at the potential cost of reduced diversity.

This framework provides a systematic solution to well-documented limitations of pure optimal weighting schemes, which suffer from substantial estimation uncertainty and overfitting in finite samples \citep{stock2004combination}. While optimal weights minimize in-sample MSFE by construction, they often exhibit poor out-of-sample performance due to high estimation variance. Equal weights, though potentially biased, provide robust out-of-sample performance with zero estimation variance. The shrinkage approach balances these competing considerations through the data-driven selection of $\lambda$.

The methodology is also consistent with models of information processing and expectation formation from behavioral finance. The two-stage estimation procedure--first computing optimal weights, then determining the shrinkage parameter--parallels the gradual information incorporation documented in behavioral finance literature and sticky information models \citep{mankiw2002sticky,coibion2015information}. This structure accommodates the slow adjustment of beliefs and delayed price responses that characterize actual market dynamics, providing a behavioral foundation for the observed benefits of shrinkage in forecast combination. The resulting weights can be interpreted as reflecting the optimal degree of conservatism when updating forecasting strategies in response to new information.

\subsection{Largest Eigenvalue, Market Mode, and the Golden Criterion}

When the covariance matrix $\mathbf{\Sigma}$ admits the all-ones vector as an eigenvector with eigenvalue $\delta$, the minimum-variance portfolio coincides with equal weighting regardless of the magnitude of $\delta$. In this exact case, $\delta=\frac{\mathbf{1}^{T}\mathbf{\Sigma}\mathbf{1}}{\mathbf{1}^{T}\mathbf{1}}$. More generally, $\frac{\mathbf{1}^{T}\mathbf{\Sigma}\mathbf{1}}{\mathbf{1}^{T}\mathbf{1}}$ is the Rayleigh quotient of $\mathbf{\Sigma}$ in the all-ones direction. Let $\mathbf{u}=\mathbf{1}/\sqrt{\mathbf{1}^{T}\mathbf{1}}$, then $\frac{\mathbf{1}^T\mathbf{\Sigma}\mathbf{1}}{\mathbf{1}^{T}\mathbf{1}}=\mathbf{u}^T\mathbf{\Sigma} \mathbf{u}$, which measures the variance explained along the equal-weight direction.

This distinction clarifies the relation between the Golden Criterion and the first principal component. For the positive definite covariance matrices considered, the largest eigenvalue, denoted by $\rho_{\max}(\mathbf{\Sigma})$, is the maximum value of $\mathbf{x}^{T}\mathbf{\Sigma}\mathbf{x}$ over all unit vectors $\mathbf{x}$. Hence, by the Rayleigh--Ritz theorem, $\mathbf{u}^{T}\mathbf{\Sigma}\mathbf{u} \leq \rho_{\max}(\mathbf{\Sigma})$, with equality if and only if the all-ones direction belongs to the eigenspace associated with $\rho_{\max}(\mathbf{\Sigma})$. Thus, the all-ones directional variance equals the dominant eigenvalue only when the first principal component is aligned with the all-ones direction.

This connection gives the Golden Criterion a market-mode interpretation. When the all-ones direction is close to the first principal component, forecast errors share a common component with similar loadings across predictors, and equal weighting becomes closely related to the dominant covariance structure. If $\mathbf{\Sigma}$ is entrywise non-negative and the all-ones vector is an eigenvector, Perron--Frobenius theory implies that this positive eigenvector corresponds to the dominant eigenvalue. In that case, the Golden Criterion and the first principal component coincide: the eigenvalue associated with the all-ones vector is the dominant eigenvalue, i.e., $\delta=\rho_{\max}(\mathbf{\Sigma})$, and hence $\mathbf{\Sigma}\mathbf{1}=\rho_{\max}(\mathbf{\Sigma})\mathbf{1}=\delta_1 \mathbf{1}$. When some covariance entries are negative, however, an all-ones eigenvector need not correspond to the dominant eigenvalue, so the market-mode interpretation becomes weaker even though the Golden Criterion itself remains an eigenstructure condition.

In Section~\ref{sec: empirical}, we examine whether real-world forecast error covariance structures approximate this condition. We compare the dominant eigenvalue $\rho_{\max}(\mathbf{\Sigma}_1)$ with the all-ones directional variance $\tilde m/N$, where $\tilde m=\mathbf{1}^{T}\mathbf{\Sigma}_1\mathbf{1}$. Closeness between these two quantities suggests that the first principal component is approximately aligned with the all-ones direction; divergence between them indicates a departure from this alignment.

\section{Data}\label{sec: data}

We measure stock excess returns as the difference between continuously compounded returns on the S\&P 500 index (including dividends) and the 30-day Treasury bill rate. Following \citet{welch2008comprehensive}, we employ 13 widely used predictors that capture key financial and macroeconomic variables established in the return predictability literature:
\begin{enumerate}
	\item \textit{Log dividend-price ratio} (DP): log of a 12-month moving sum of S\&P 500 dividends minus log of the S\&P 500 index level.
	\item \textit{Log earnings-price ratio} (EP): difference between the log of the 12-month moving sum earnings on the S\&P 500 and the log of the index level.
	\item \textit{Log dividend-earnings ratio} (DE): log of a 12-month moving sum of dividends minus log of earnings.
	\item \textit{Stock return volatility} (RVOL): 12-month moving standard deviation of S\&P 500 returns.
	\item \textit{Book-to-market ratio} (BM): book-to-market value ratio for the Dow Jones Industrial Average.
	\item \textit{Net equity expansion} (NTIS): ratio of 12-month moving sum of net equity issues by NYSE-listed firms to total end-of-year NYSE market capitalization.
	\item \textit{Treasury bill rate} (TBL): three-month Treasury bill rate (secondary market).
	\item \textit{Long-term government bond yield} (LTY): yield on long-term government bonds.
	\item \textit{Long-term government bond return} (LTR): return on long-term government bonds.
	\item \textit{Term spread} (TMS): difference between long-term government bond yield and Treasury bill rate.
	\item \textit{Default yield spread} (DFY): difference between Moody's BAA- and AAA-rated corporate bond yields.
	\item \textit{Default return spread} (DFR): difference between long-term corporate and government bond returns.
	\item \textit{Inflation} (INFL): Consumer Price Index for All Urban Consumers from the Bureau of Labor Statistics.
\end{enumerate}

Table~\ref{tab:summary_and_correlation} presents descriptive statistics and correlations for the sample period from January 1996 to August 2019. Panel A shows that valuation ratios DP, EP, and DE exhibit negative means, with EP displaying the most negative mean (-3.15) and large variation (standard deviation = 0.37). Interest rate and bond return variables demonstrate considerable volatility, particularly LTR, which ranges from -6.42 to 8.40 with a standard deviation of 3.02. Panel B reveals strong correlations among valuation-based predictors, notably a -0.88 correlation between EP and DE, and a 0.72 correlation between DE and DFY. The negative associations between interest rate variables (TBL, LTY) and valuation ratios align with standard economic theory.

\bigskip
\centerline{\bf [Place Table~\ref{tab:summary_and_correlation} about here]}
\bigskip

\section{Empirical Results}\label{sec: empirical}
We examine stock excess return predictability using combination strategies across two frameworks. We first investigate in-sample predictive performance before conducting out-of-sample evaluation, where forecasts are assessed against data excluded from parameter estimation, a more stringent test of return predictability.

\subsection{In-sample Analysis}\label{sec: in-sample}
\subsubsection{Individual Predictor Performance}
We generate time-series forecasts for each GW predictor following equations~(\ref{eq:individual regression}) and~(\ref{eq:individual regression forecast}). Each predictor variable $x_j$ is standardized to a unit standard deviation to ensure comparable economic interpretation across variables.

For individual regressions, let $T$ denote the total sample months, with the first $T_1$ months allocated for in-sample estimation and the remainder reserved for out-of-sample forecast generation. We employ $T_1=54$ (January 1996 to June 2000) and $T_1=42$ (January 1996 to June 1999), respectively, to assess methodological robustness. Table~\ref{tab:ins indvidual} demonstrates that all GW predictors achieve superior in-sample $R^2$ values relative to the historical mean benchmark, with performance improving as forecast horizon $J$ increases. However, since in-sample estimates are susceptible to overfitting, out-of-sample results provide more compelling evidence of genuine predictive ability.

\bigskip
\centerline{\bf [Place Table~\ref{tab:ins indvidual} about here]}
\bigskip

\subsubsection{Combination Forecast Performance}
Combination forecasts integrate the individual forecast series generated from each GW predictor above. The individual forecast series are sequentially partitioned into training ($T_{c1}$), validation ($T_{c2}$), and test sets ($T_{c3}$) for combination analysis. The adaptive shrinkage combination method requires a two-step estimation: first, determining optimal weights from the training period $T_{c1}$, then selecting the optimal shrinkage parameter $\lambda_{opt}$ through the validation $T_{c2}$. Table~\ref{tab:ins combination} reports combination results in the validation set under optimal shrinkage parameters, as an in-sample analysis.

We compare optimal weights (OWs), equal weights (EWs), principal component analysis (PCA), and established shrinkage methods, including Lasso, Ridge, and Elastic Net (ENet). We also examine the recent combination benchmark models such as those in \cite{rapach2010out} (Median, Trimmed Mean, Discount Mean Square Prediction Error (DMSPE)), \cite{blanc2020bias}\footnote{see Appendix~\ref{sec: appendix detail}.}, and ASWs--\cite{li2025dominate}\footnote{see Appendix~\ref{sec: appendix detail}.}. Since shrinkage methods require validation sets for hyperparameter tuning, results for adaptive shrinkage weights (2S--ASWs), Lasso, Ridge, and ENet in Table~\ref{tab:ins combination} reflect performance on identical validation samples $T_{c2}$ (July 2000 to December 2006 for $T_1 = 54$; July 1999 to December 2006 for $T_1 = 42$). Conversely, OWs, EWs, Median, Trimmed Mean, DMSPE, \cite{blanc2020bias}, ASWs--\cite{li2025dominate}, and PCA methods require no separate validation phase; therefore, their results utilize combined training and validation samples (January 1996 to December 2006).

We categorize methods into two groups: shrinkage methods (2S--ASWs, Lasso, Ridge, and ENet) and one-step methods (OWs, \cite{blanc2020bias}, EWs, Median, Trimmed Mean, DMSPE, PCA, and ASWs--\cite{li2025dominate}) (see Table \ref{tab:methods_glossary} for details). Within the shrinkage group, ENet consistently delivers superior in-sample performance across both $T_1$ specifications. 2S--ASWs demonstrate advantages at short horizons (e.g., $J = 1$), but neither 2S--ASWs nor alternative shrinkage methods dominantly outperform the others as horizons lengthen. Among one-step methods, OWs achieve the highest in-sample $R^2$ values, followed by \cite{blanc2020bias}. EWs, Median, Trimmed Mean, and ASWs--\cite{li2025dominate} exhibit comparable performance, with ASWs--\cite{li2025dominate} ($A=10$) yielding the strongest results among them. At $J = 1$, PCA(1) and PCA(3) outperform EWs, while EWs' performance surpasses PCA methods at longer forecast horizons.

\bigskip
\centerline{\bf [Place Table~\ref{tab:ins combination} about here]}
\bigskip

\subsection{Out-of-Sample Analysis} \label{sec: Out-of-Sample Forecast}
\subsubsection{Individual Predictor Performance}
Promising in-sample performance does not guarantee effective out-of-sample predictive ability \citep{welch2008comprehensive}. Therefore, examining the out-of-sample forecasting power of individual and combination methods is essential. Following \citet{campbell2008predicting, welch2008comprehensive, rapach2010out}, we employ the out-of-sample $R^2$ statistic to measure forecasting performance:
\begin{equation*}
    R_{OS}^{2} = 1- \frac{\frac{1}{L}{\sum_{t=1}^{L}(r_{t}-\hat{r}_{t})^2}}{\frac{1}{L}{\sum_{t=1}^{L}(r_t-\overline{r}_t)^2}},
\end{equation*}
where $\hat{r}_t$ represents out-of-sample forecasts from individual regressions based on equation~(\ref{eq:individual regression forecast}), $\overline{r}_t$ denotes historical mean benchmark forecasts, and $L$ is the number of out-of-sample observations. Both forecasts employ expanding estimation windows. Positive (negative) $R^2_{OS}$ values indicate superior (inferior) forecasting accuracy relative to the historical mean benchmark, with positive values signifying lower mean squared forecast error (MSFE) and meaningful predictive power. We implement the MSFE-adjusted test of \citet{clark2007approximately} to assess whether forecasts exhibit significantly lower MSFE than the historical mean benchmark, testing $H_0: R^2_{OS} \leq 0$ against $H_A: R^2_{OS} > 0$. We also calculate HLN statistics of \cite{harvey1998tests} to conduct a forecast encompassing test between 2S--ASWs and other benchmarks pairwise, comparing the information content of 2S--ASWs with that of the other methods.

Table~\ref{tab:oos indvidual} presents out-of-sample results for individual predictor regressions. Among all predictors, only the dividend-price ratio (DP) predictor exhibits robust out-of-sample predictive power. When $T_1 = 54$, DP generates $R^2_{OS}$ values ranging from 0.77\% at $J=1$ to 2.44\% at $J=12$. When $T_1 = 42$, performance improves substantially, with $R^2_{OS}$ increasing from 0.87\% at $J=1$ to 6.54\% at $J=12$. All remaining GW predictors yield negative $R^2_{OS}$ values, indicating inferior performance relative to the historical mean benchmark. At the extended horizon of $J=24$, however, the term spread (TMS) shows $R^2_{OS}$ values over 7\% under both $T_1$ settings, demonstrating superior forecasting ability. These findings underscore the limited predictive power of individual GW predictors and provide strong motivation for the subsequent combination approaches examined.

\bigskip
\centerline{\bf [Place Table~\ref{tab:oos indvidual} about here]}
\bigskip

\subsubsection{Combination Forecast Performance}
The adaptive shrinkage combination procedure using 2S--ASWs is illustrated in Figure~\ref{fig:oos forecast process}, with alternative shrinkage methods (Lasso, Ridge, ENet) following analogous procedures. We first derive OWs based on forecasts from the initial $T_{c1}$ months, then designate the subsequent $T_{c2}$ months as the validation set to determine the optimal shrinkage parameter $\lambda_{opt}$ and compute 2S--ASWs. The construction of 2S--ASWs incorporates OWs from the training period, EWs, the accuracy vector $\bm{S}$, and the diversity matrix $\bm{D}$, all of which are calculated using individual forecasts within the validation set. After obtaining $\lambda_{opt}$ by solving equation~(\ref{eq:restrict opt lambda}), we derive 2S--ASWs accordingly. Finally, the estimated 2S--ASWs, OWs, and EWs generate combination forecasts for future excess returns over the test period $T_{c3}$ (January 2007 to August 2019).

\bigskip
\centerline{\bf [Place Figure~\ref{fig:oos forecast process} about here]}
\bigskip

To evaluate the performance of the adaptive shrinkage combination method, we examine various combinations of ($T_{c1}$, $T_{c2}$). For consistency in weight derivation, we recompute OWs using combined training and validation samples. This framework provides a rigorous evaluation when comparing the out-of-sample predictive performance of OWs and \citet{blanc2020bias} with others. Meanwhile, we also evaluate the performance of OWs derived from periods $T_{c2}$ (i.e., $\bm{\Sigma}_2$) only, given that $\bm{\Sigma}_2$ captures the most recent information closest to the test period. Lasso, Ridge, and ENet follow similar implementations, with the training set ($T_{c1}$) for regression coefficient estimation and the validation set ($T_{c2}$) for hyperparameter selection. We also assess the out-of-sample forecasting ability of PCA(1) and PCA(3).

The regularization parameter $\alpha$, controlling penalty strength in Lasso, Ridge, and Elastic Net regressions, is selected from a logarithmically spaced grid over $[10^{-6}, 10^{6}]$ with 500 points. The mixing parameter $\beta$ in Elastic Net, determining the balance between $\ell_1$ and $\ell_2$ penalties, is chosen from values ranging from 0.0001 to 1.0. Table~\ref{tab:sample division} in the Appendix details sample division and parameter selection procedures.

Table~\ref{tab:oos combination} reports out-of-sample $R^2$ statistics of combination methods ($R^2_{COS}$) and associated $p$-values for adaptive shrinkage combination (2S--ASWs) and alternative methods under $T_1 = 54$ and $T_1 = 42$ specifications for individual forecast generation. 2S--ASWs consistently exhibit positive $R^2_{COS}$ values and substantially outperform alternative methods, most decisively at the one-month horizon ($J=1$). At $J=3$ and $J=6$, in contrast, multiple methods, including DMSPE and ASWs--\cite{li2025dominate}, compete closely, with no single approach dominating. OWs ($\bm{\Sigma}_2$) and OWs demonstrate the weakest out-of-sample predictive performance, particularly at shorter horizons ($J=1$ and $J=3$), but exhibit the strongest forecasting power as the horizon lengthens to $J=12$ to $J=24$, especially the OWs ($\bm{\Sigma}_2$). The method proposed by \cite{blanc2020bias} displays a similar pattern, underperforming at short-term horizons while achieving substantially stronger results at $J=24$. This similarity may be attributed to the fact that \cite{blanc2020bias} employs a one-step shrinkage procedure directly based on the one-step OWs, without a second validation period to determine $\lambda_{opt}$. The bifurcation of horizon-dependent performance suggests that short-horizon forecasting benefits most from protection against estimation noise (where adaptive shrinkage toward equal weights excels), while long-horizon forecasting benefits from exploiting persistent covariance structure (where precise weight estimation is rewarded). 

Other shrinkage methods (Lasso, Ridge, ENet) and PCA methods yield negative $R^2_{COS}$ values, but they also perform better at longer forecast horizons. EWs, Median, Trimmed Mean, and ASWs--\cite{li2025dominate} exhibit limited predictive ability at $J = 1$, though they can all achieve positive $R^2_{COS}$ from $J=3$ to $J=24$, making them robust choices for medium- to long-term forecasting. Notably, DMSPE with $\theta=0.9$ stands out among the one-step methods: despite yielding negative $R^2_{COS}$ at $J=1$ ($-0.27\%$), it achieves the highest $R^2_{COS}$ among all methods at $J=3$ ($5.00\%$) and $J=6$ ($13.00\%$) across nearly all sample-splitting configurations, and delivers strong performance at longer horizons ($32.91\%$ at $J=12$ and $56.57\%$ at $J=24$). The dominance of DMSPE at intermediate horizons reflects the advantage of dynamically discounting older forecast errors, which effectively captures time-varying predictor relevance without requiring explicit covariance estimation. An intriguing finding is that the $R^2_{COS}$ values of 2S--ASWs typically exhibit a ``check-mark''-shaped pattern across forecast horizons $J$, with $J=6$ often corresponding to the most negative $R^2_{COS}$. This pattern is particularly evident when $T_{c1}=12$ and $T_{c2}=120$, contrasting with typical term structure findings \citep{campbell1988dividend, fama1988dividend, bansal2021term}. Of particular note, while both OWs and EWs individually yield negative $R^2_{COS}$ at $J=1$, the shrinkage-based combination in 2S--ASWs nevertheless produces positive out-of-sample $R^2_{COS}$, a result that is perhaps surprising given the poor predictive performance of its constituent forecasts. In the Appendix~\ref{sec:appendix 2024}, we further examine the performance of all these methods when the sample period is extended to December 2024. 2S--ASWs continue to demonstrate their superiority over competing benchmarks, particularly at the short horizon of $J=1$; however, the overall out-of-sample forecasting performance deteriorates noticeably relative to that of 2019, a finding that is consistent with the evidence documented in \citet{denk2024predicting}. We refer the reader to Table~\ref{tab:extend 2024 oos combination} for a detailed exposition of these results.

\bigskip
\centerline{\bf [Place Table~\ref{tab:oos combination} about here]}
\bigskip

\subsubsection{Shrinkage Dynamics}

To understand 2S--ASWs' out-of-sample forecasting preferences, Figure~\ref{fig:lambda_opt} displays $\lambda_{opt}$ values and weights assigned to each GW prediction series for the optimal configuration ($T_{c1}=6$, $T_{c2}=126$) when $T_1=42$. Nearly all $\lambda_{opt}$ values exceed 0.5. Specifically, at $J=1$ and $J=3$, the proportions of $\lambda_{opt}$ exactly equaling 1 are 19.74\% and 57.33\%, respectively, while 66.45\% and 84.67\% of the values exceed 0.9. These results indicate a strong preference for EWs, particularly during the initial out-of-sample periods.

A notable feature is the decline in $\lambda_{opt}$ after early 2015, suggesting increased reliance on OWs in the adaptive shrinkage combination thereafter. This period coincides with several major market developments, including the Federal Reserve's first interest rate increase after 2006 \citep{anderson2020happened}, the crash of the Chinese stock market \citep{huang2019saving}, and the collapse in crude oil prices \citep{berk2020shift}. We interpret this pattern as suggestive evidence that 2S--ASWs adjust toward OWs when the covariance structure of forecast errors departs from the conditions favoring equal weighting. The preference of 2S--ASWs, which starts to shift from EWs towards OWs, especially during 2017 and early 2018, before rapidly reverting to EWs, aligns closely with the finding of \citet{wang2024non}, which demonstrates that portfolio strategies tend to favor na\"{\i}ve simple averaging during periods of market reversals. Specifically in Fig.~\ref{fig:excess returns}, the market exhibits a relatively stable positive trend from 2016 to early 2018, during which adaptive shrinkage weights increasingly deviate from EWs. However, as the market volatility resurges after the beginning of 2018, 2S--ASWs quickly revert toward EWs.

\bigskip
\centerline{\bf [Place Figure~\ref{fig:lambda_opt} about here]}
\bigskip

\bigskip
\centerline{\bf [Place Figure~\ref{fig:excess returns} about here]}
\bigskip

In Appendix~\ref{sec:appendix A.2}, we provide a theoretical analysis demonstrating that the same two-stage optimization framework can be applied after reversing the roles of the error covariance matrices $\mathbf{\Sigma}_1$ and $\mathbf{\Sigma}_2$. Motivated by this observation, we empirically examine a reversed specification in which the temporal roles of the training ($\mathbf{\Sigma}_1$) and validation samples ($\mathbf{\Sigma}_2$) are interchanged: $T_{c2}$ is used to estimate OWs while $T_{c1}$ computes $\lambda_{opt}$ and adaptive shrinkage weights (R2S--ASWs). Although this reversed specification yields weaker out-of-sample performance in general, R2S--ASWs demonstrate notable improvements at long-term horizons ($J=12/24$). Most strikingly, established shrinkage methods, Ridge, Lasso, and ENet, compared to the original scenario, exhibit a surprising improvement under the reversed specification, with Ridge emerging as the strongest performer at $J=24$. Crucially, the plots of $\lambda_{opt}$ suggest a consistent and symmetrical pattern: $\lambda_{opt}$ values still predominantly exceed 0.5 and frequently approach 1 (EWs) across five horizons, underscoring the persistent dominance and robustness of EWs in equity premium forecasting.

Figure~\ref{fig:SWs} exhibits the weight elements of 2S--ASWs assigned to each GW predictor's forecast. Individual forecast series demonstrate overall symmetric patterns. At $J=3$, weights assigned to each predictor are minimal before early 2015, while at $J=1$ and $J=12$, weights exhibit greater temporal variation. Among factors, the DP predictor consistently receives the highest weight (except at $J=1$), particularly in later out-of-sample periods where its weight substantially exceeds others. This suggests an important role for DP in the adaptive shrinkage combination, consistent with its superior individual out-of-sample performance.

Additionally, RVOL receives significantly negative weights at $J=3$, 6, and 12, reflecting poor combination performance. DE and LTR also receive substantial negative weights at $J=1$ and $J=24$, respectively, aligning with their individual out-of-sample performance. The magnitude and sign of most predictors' weights vary temporally rather than remaining constant, highlighting both the challenge of individual predictor use and the flexibility afforded by forecast combination methods.

\bigskip
\centerline{\bf [Place Figure~\ref{fig:SWs} about here]}
\bigskip

\subsubsection{Eigenvalue Evidence}

Proposition~\ref{proposition 1} has practical implications under the Golden Criterion. If the first-stage error covariance matrix $\mathbf{\Sigma}_1$ has the all-ones vector as an eigenvector, then the first-stage optimal weights are exactly equal weights. Consequently, the shrinkage path collapses to equal weights, and the naive strategy is optimal, i.e., $\hat{\mathbf{w}}^\lambda=\hat{\mathbf{w}}^{ow}=\mathbf{w}^{ew}$ for all $\lambda\in[0,1]$, regardless of the second-stage covariance matrix $\mathbf{\Sigma}_2$. In this sense, an algebraic relationship between the eigenvalue $\delta_1$ of $\mathbf{1}$ associated with $\mathbf{\Sigma}_1$, the sum of all elements in $\mathbf{\Sigma}_1^{-1}$, and the number of predictors $N$ is satisfied, i.e., $m=\mathbf{1}^T\mathbf{\Sigma}_1^{-1}\mathbf{1}=N/\delta_1$, or equivalently $\delta_1=N/m$. The complete mathematical derivation of this conclusion is provided in Appendix~\ref{sec:appendix A.2}. We examine the empirical counterpart of this implication by comparing the dominant eigenvalue $\rho_{\max}(\mathbf{\Sigma}_1)$ with the all-ones Rayleigh quotient $\frac{\tilde m}{N}$, where $\tilde m=\mathbf{1}^T\mathbf{\Sigma}_1\mathbf{1}$.
Under the exact Golden Criterion, $\tilde m/N$ equals the eigenvalue $\delta_1$; it equals the dominant eigenvalue only when this all-ones eigenvector is also the first principal component. More generally, closeness between $\rho_{\max}(\mathbf{\Sigma}_1)$ and $\tilde m/N$ indicates that the first principal component is closely aligned with the all-ones direction.

Still taking the configuration $T_{c1}=6$, $T_{c2}=126$, and $T_1=42$ as an example, we compute in each rolling window the dominant eigenvalue $\rho_{\max}(\mathbf{\Sigma}_1)$ and $\tilde m/N$. The empirical results are presented in Figure~\ref{fig:delta_1 and m/N}. At short horizons $J=1,3,6$, all entries of $\mathbf{\Sigma}_1$ are positive across rolling windows. In this case, Perron--Frobenius theory provides a useful interpretation: when the all-ones direction is close to an eigenvector, it is naturally associated with the dominant positive principal component. The very close movement between $\rho_{\max}(\mathbf{\Sigma}_1)$ and $\tilde m/N$ at these horizons is therefore consistent with the first principal component being approximately aligned with the all-ones direction, supporting the empirical relevance of the Golden Criterion.

As the forecast horizon $J$ lengthens, $\lambda_{opt}$ deviates more visibly from unity, and the gap between $\rho_{\max}(\mathbf{\Sigma}_1)$ and $\tilde m/N$ becomes more pronounced. This pattern coincides with a change in the sign structure of the covariance matrix: for $J=12$ and $J=24$, negative entries begin to appear in $\mathbf{\Sigma}_1$ in some rolling windows. Once the covariance matrix is no longer entrywise positive, the all-ones direction need not correspond to the dominant eigenvector, even if it remains economically meaningful as a diversification direction. The observed divergence between $\rho_{\max}(\mathbf{\Sigma}_1)$ and $\tilde m/N$ is therefore consistent with the forecast-error covariance structure moving farther away from the all-ones alignment at longer horizons.

Because $\lambda_{opt}$ is constrained to lie in $[0,1]$ and many estimates remain close to the upper boundary, exact equality between the two quantities is not expected in finite samples. Nevertheless, during periods in which 2S--ASWs deviate more clearly from EWs, such as 2015--2018, the divergence between $\rho_{\max}(\mathbf{\Sigma}_1)$ and $\tilde m/N$ becomes more evident, especially as $J$ increases. Even at short horizons, $J=1$ and $J=3$, small but persistent discrepancies can still be observed. Overall, these empirical patterns suggest that real-world covariance structures often approximate the Golden Criterion at short horizons, while longer-horizon covariance matrices display weaker all-ones alignment, helping to explain why na\"{\i}ve diversification is more effective in short-term forecasting environments than in others.

\bigskip
\centerline{\bf [Place Figure~\ref{fig:delta_1 and m/N} about here]}
\bigskip

\subsubsection{Asset Allocation}
While out-of-sample $R^2$ statistics measure mean squared forecast errors relative to historical mean benchmarks, economic performance metrics of optimal portfolios provide more relevant assessments for practitioners. Economic evaluation tests whether investors can identify and exploit predictability for market timing. Following \citet{campbell2008predicting}, we employ utility-based metrics to evaluate the economic performance of adaptive shrinkage combination methods.

Consider a mean-variance investor who optimally allocates wealth between the equity market portfolio and risk-free bonds at each time $t$. The optimal weight allocated to the risky market portfolio is determined by:
\begin{equation*}
    w_t = \frac{1}{\gamma}\frac{\hat{r}_{comb,t+1,t+J}}{{\hat{\sigma}^2}_{comb,t+1,t+J}}, w_t \in [0, 1.5]
\end{equation*}
where $\gamma$ represents the risk-aversion parameter, $\hat{r}_{comb,t+1,t+J}$ denotes the predicted monthly excess return on the U.S. aggregate stock market based on combination methods (2S--ASWs, OWs, EWs, Lasso, Ridge, ENet) or PCA approaches up to period $t$, and ${\hat{\sigma}^2}_{comb,t+1,t+J}$ represents the predicted volatility of monthly market excess returns estimated using a 10-year rolling window. Following \citet{rapach2016short}, we set $\gamma = 3$ and constrain $w_t \in [0, 1.5]$ to prevent short sales while allowing up to 50\% leverage.

The investor realizes average utility throughout the out-of-sample evaluation period:
\begin{equation*}
    U(\mu,\sigma)=\hat{\mu}_p - 0.5\gamma{\hat{\sigma}^2}_p,
\end{equation*}
where $\hat{\mu}_p$ and ${\hat{\sigma}^2}_p$ represent the sample mean and variance of the investment portfolio over the complete out-of-sample period. We calculate utility gains as the difference in realized average utility between investors using adaptive shrinkage combination or alternative methods versus those employing historical mean forecasts. We annualize utility gains to express them as expected annualized returns, interpretable as the transaction costs or portfolio management fees that mean-variance investors would pay to access predictive regressions rather than historical mean forecasts. Following \citet{rapach2016short}, we assume investors rebalance at the same frequency as the forecast horizon to analyze the economic value of long-horizon predictability.

Table~\ref{tab:oos combination utility} reports average utility gains (in annualized percentage returns) for mean-variance investors with relative risk aversion of three who allocate between equities and risk-free assets using predictive forecasts versus historical averages. Consistent with $R^2_{COS}$ statistics, adaptive shrinkage combination forecasts demonstrate substantial economic value. 2S--ASWs maintain positive average utility gains across all ($T_1$, $T_{c1}$, $T_{c2}$, $J$) configurations at monthly horizons, with the maximum reaching 10.67\% when ($T_1=42, T_{c1}=48$, $T_{c2}=84$). OWs, OWs ($\bm{\Sigma}_2$), the method of \citet{blanc2020bias}, and alternative shrinkage methods occasionally outperform 2S--ASWs in specific cases such as ($T_{c1}=48$, $T_{c2}=84$) for both $T_1=54$ and $T_1=42$ especially at long horizon ($J=24$).

Following \citet{rapach2016short}, we include utility gains from passive buy-and-hold strategies in the final row of Table~\ref{tab:oos combination utility}. Buy-and-hold portfolios, where investors passively maintain market exposure without rebalancing, generate positive utility gains relative to historical mean forecasts except at $J=24$. However, these gains remain substantially below those achieved using the adaptive shrinkage combination method.

\bigskip
\centerline{\bf [Place Table~\ref{tab:oos combination utility} about here]}
\bigskip

We calculate Sharpe ratios (SR) for portfolios based on adaptive shrinkage combination and alternative forecasting methods relative to historical mean benchmarks. The Sharpe ratio equals mean excess portfolio return divided by portfolio return standard deviation: $SR=\mu_p/\sigma_p$. Following \citet{rapach2016short}, we incorporate 50 basis points of transaction costs to approximate realistic trading conditions. Table~\ref{tab:oos combination Sharp} presents annualized Sharpe ratios for mean-variance investors allocating between risky equities and risk-free assets based on various forecasting approaches, with historical mean Sharpe ratios serving as benchmarks. The adaptive shrinkage combination method consistently delivers superior Sharpe ratios relative to historical mean benchmarks over the January 2007 to August 2019 evaluation period, with most ratios exceeding 0.5, and particularly strong performance when $T_1=54$. OWs, EWs, and alternative shrinkage methods also generate relatively high Sharpe ratios, surpassing both historical mean and buy-and-hold benchmarks. Notably, OWs ($\bm{\Sigma}_2$) consistently achieve the highest Sharpe ratio at $J=24$, which highlights the informational value of recency. PCA methods exhibit modest performance, with the lowest Sharpe ratio reaching only 0.26 at $J=12$.

Overall, the above results demonstrate that stock return prediction using the adaptive shrinkage combination method provides substantial economic value for risk-averse investors through both utility gains and Sharpe ratio improvements over conventional benchmarks.

\bigskip
\centerline{\bf [Place Table~\ref{tab:oos combination Sharp} about here]}
\bigskip

\section{Discussion}\label{sec: Disc}
{\bf  High-dimensional portfolio}. Recent advances in high-dimensional portfolio optimization have focused on sparsity-inducing estimators, such as LASSO-based selection, thresholding methods, or other penalized mean-variance frameworks (e.g., \citealp{Fan2012, brodie2009sparse}). These approaches attempt to overcome estimation risk by selecting a subset of assets or predictors to form a concentrated portfolio. While they are mathematically appealing in large-dimensional settings, such methods face two well-documented challenges in practice. First, sparse portfolios are often unstable and highly sensitive to small changes in data, leading to frequent and costly rebalancing. Second, by excluding many available assets, they can compromise market investability and diversification, potentially increasing systematic risk and reducing liquidity.

Our results do not conflict with these high-dimensional techniques but instead complement them with a theoretically grounded criterion for combining weights. Even when a sparse selection procedure is applied as a preliminary step, the resulting portfolio weights can be further improved under our framework. Specifically, if the error covariance structure of the selected assets approaches the Golden Criterion, where the covariance matrix of residuals has an all-ones eigenvector, then an equal-weight adjustment or an adaptive shrinkage combination achieves optimality in the minimum-variance sense. This ensures that diversification is not lost even when statistical procedures select a narrow set of predictors, providing a safeguard against estimation error and instability. Thus, while high-dimensional selection methods are valuable for reducing dimensionality, our approach addresses a different but equally important aspect of portfolio design: achieving robust and theoretically optimal allocations given the realized covariance structure. By explicitly combining selection-driven precision with covariance-driven diversification, the adaptive strategy we propose remains applicable and beneficial across both low- and high-dimensional portfolio settings.

{\bf Behavioral foundation.} A large body of research shows that under uncertainty, individuals gravitate toward even or equal allocations. The equiprobability bias leads people to treat random processes as if all outcomes were equally likely \citep{Cappelen2016}, while the equality heuristic highlights that people often divide resources evenly as a simplifying rule that reduces cognitive effort and perceived regret \citep{messick1993,gigerenzer2011}. Portfolio evidence echoes this tendency: \citet{benartzi2001} documents the widespread use of the ``$1/N$ heuristic'' in retirement plans, and \citet{kahneman2011} emphasizes that such rules reflect intuitive reasoning favoring simplicity and fairness. Similarly, findings on probability matching \citep{vulkan2000,koehler2010} show that decision makers frequently align choices with perceived equal odds rather than maximize expected payoffs. Although our empirical implementation uses Goyal--Welch predictors as forecasting signals rather than tradable assets, the adaptive shrinkage logic is not specific to these proxies: it offers a general framework for combining noisy signals, forecasts, or allocation rules in practical investment settings. This broader interpretation is consistent with \citep{wang2024non}, who show that na\"{\i}ve strategies can often attain, or closely approximate, optimal performance in realistic portfolio environments. Together, these insights suggest that equal-weighting resonates with a fundamental psychological intuition of ``even chances'', explaining its persistence in practice.

{\bf Risk attitudes.} Investor preferences further shape the appeal of equal-weighting. A risk-neutral investor views equal allocation as a fair baseline; a risk-seeking investor may tilt toward optimized, concentrated bets that exploit predictive signals; and a risk-averse investor naturally favors diversification and stability, often gravitating to equal-weighting or shrinkage toward it. Our Golden Criterion unifies these perspectives by showing that under identifiable covariance structures, equal-weighting is not only psychologically appealing to risk-averse investors but can be mathematically optimal. This bridges behavioral heuristics with formal portfolio theory, demonstrating when intuitive ``even chance'' reasoning aligns with rigorous decision-making principles.

\textbf{Horizon-dependent strategy selection.} The behavioral and decision-theoretic considerations discussed above provide an interpretation of the horizon-dependent empirical patterns. At short horizons ($J=1$), equity premium forecasts are especially noisy, and estimation uncertainty makes purely optimized weights costly. In this environment, shrinkage toward EWs can reduce variance while preserving diversification, which is consistent with the strong short-horizon performance of 2S--ASWs and the frequent concentration of $\lambda_{opt}$ near unity. As the horizon lengthens, persistent macroeconomic predictors may contain more stable information, and the covariance structure of forecast errors appears to move farther away from the all-ones alignment, as reflected by the growing divergence between $\rho_{\max}(\mathbf{\Sigma}_1)$ and $\tilde m/N$ in Figure~\ref{fig:delta_1 and m/N}. Under these conditions, covariance-adjusted methods such as OWs($\bm{\Sigma}_2$) and bias-corrected approaches such as \citet{blanc2020bias} become more competitive. These findings suggest that the relative value of equal weighting and optimization varies with the forecast horizon: equal weighting is valuable when the forecast-error covariance structure is close to the Golden Criterion, whereas optimized weighting becomes attractive as the covariance structure moves away from the all-ones alignment.

\section{Conclusion}\label{sec: conclusion}

Investors and portfolio managers face a familiar dilemma: should they trust complex optimization models to determine asset weights, or follow simple, equal-weight rules that often perform better in real markets? While optimization promises precision, it often leads to concentrated portfolios that are unstable, costly to rebalance, and less investable. By contrast, equal weighting ignores detailed forecasts but delivers consistent, diversified results.

This paper revisits one of the most enduring puzzles in behavioral economics and portfolio management: the persistent outperformance of na\"{\i}ve equal-weight allocations relative to statistically optimized weights in out-of-sample settings.
We address this gap by providing a necessary and sufficient condition under which equal weighting coincides with the minimum-variance solution. This condition offers a covariance-based explanation for when na\"{\i}ve diversification can be optimal. Specifically, when the error covariance matrix of either the initial estimation stage or the adaptive combination stage admits an eigenvector consisting entirely of ones, the minimum-variance optimal solution coincides with equal weighting. This condition provides the rigorous theoretical foundation for the empirical success of na\"{\i}ve diversification, transforming it from a statistical curiosity into a principled outcome of portfolio theory. 

Building on this insight, we propose a two-stage adaptive shrinkage methodology that dynamically combines optimal and na\"{\i}ve weights based on evolving covariance structures. This framework balances the trade-off between exploiting predictive accuracy and preserving forecast diversity, avoiding the instability and concentration risks that plague purely optimized allocations. Empirically, our analysis of S\&P 500 excess returns demonstrates substantial economic value across prediction horizons. The proposed 2S--ASWs method consistently outperforms benchmark models at the one-month horizon and enhances the long-horizon performance. For investors, this translates to annualized utility gains exceeding 200 basis points, significant improvements in Sharpe ratios, and an optimal shrinkage parameter often above 0.5. These findings highlight the practical importance of diversity in forecasts and portfolios, suggesting that broad participation often outweighs narrow precision in real-world markets.

More broadly, our results have three implications for behavior and financial economics. First, they suggest why na\"{\i}ve strategies can be difficult to beat out-of-sample: many real-data environments appear close to the Golden Criterion, making equal weighting nearly optimal in practice. Second, they provide a covariance-based rationale for diversification rules that maintain stability and investability. Third, they document a horizon-dependent pattern in return predictability: at short horizons, where estimation noise is severe, shrinkage toward equal weights is especially valuable; at longer horizons, persistent covariance patterns can make optimized weights more competitive. For practitioners, the message is clear: simple rules can be optimal in the right conditions, and blending them with optimization can improve results without sacrificing stability, practicality, or predictive power across varying horizons.

Future research can extend our framework in several directions. One is to explore whether the Golden Criterion emerges endogenously from equilibrium models of asset returns, particularly under heterogeneous beliefs or market frictions. Another is to apply the methodology to multi-asset portfolios, international markets, or alternative investments, where error covariance structures may deviate more significantly from the Golden Criterion. Finally, investigating why the dominance of diversity-driven shrinkage at short horizons reverses at longer horizons, and whether this reversal reflects the term structure of estimation risk, would deepen our understanding of horizon-dependent return predictability.

In sum, this paper establishes the Golden Criterion as a useful diagnostic tool for portfolio theory and the adaptive shrinkage combination as a practical bridge between empirical robustness and theoretical optimality. By providing both a precise condition for na\"{\i}ve dominance and a flexible framework for combining forecasts, we reconcile simplicity and sophistication in portfolio design, offering actionable insights for both academics and practitioners.

	\clearpage

	{
        \doublespacing
		\bibliographystyle{apa}
		\bibliography{ref}

\begin{thebibliography}{}

\bibitem[Aiolfi and Timmermann, 2006]{aiolfi2006persistence}
Aiolfi, M. and Timmermann, A. (2006).
\newblock Persistence in forecasting performance and conditional combination
  strategies.
\newblock {\em Journal of Econometrics}, 135(1-2):31--53.

\bibitem[Anderson et~al., 2020]{anderson2020happened}
Anderson, A.~G., Ihrig, J.~E., Meade, E.~E., and Weinbach, G.~C. (2020).
\newblock {\em What Happened in Money Markets after the Fed's December Rate
  Increase?}
\newblock SSRN.

\bibitem[Bansal et~al., 2021]{bansal2021term}
Bansal, R., Miller, S., Song, D., and Yaron, A. (2021).
\newblock The term structure of equity risk premia.
\newblock {\em Journal of Financial Economics}, 142(3):1209--1228.

\bibitem[Bates and Granger, 1969]{bates1969combination}
Bates, J.~M. and Granger, C.~W. (1969).
\newblock The combination of forecasts.
\newblock {\em Journal of the Operational Research Society}, 20(4):451--468.

\bibitem[Bekaert et~al., 2002]{bekaert2002dating}
Bekaert, G., Harvey, C.~R., and Lumsdaine, R.~L. (2002).
\newblock Dating the integration of world equity markets.
\newblock {\em Journal of Financial Economics}, 65(2):203--247.

\bibitem[Benartzi and Thaler, 2001]{benartzi2001}
Benartzi, S. and Thaler, R.~H. (2001).
\newblock Naive diversification strategies in defined contribution saving
  plans.
\newblock {\em American Economic Review}, 91(1):79--98.

\bibitem[Berk and {\c{C}}am, 2020]{berk2020shift}
Berk, I. and {\c{C}}am, E. (2020).
\newblock The shift in global crude oil market structure: A model-based
  analysis of the period 2013--2017.
\newblock {\em Energy Policy}, 142:111497.

\bibitem[Blanc and Setzer, 2020]{blanc2020bias}
Blanc, S.~M. and Setzer, T. (2020).
\newblock Bias--variance trade-off and shrinkage of weights in forecast
  combination.
\newblock {\em Management Science}, 66(12):5720--5737.

\bibitem[Brodie et~al., 2009]{brodie2009sparse}
Brodie, J., Daubechies, I., De~Mol, C., Giannone, D., and Loris, I. (2009).
\newblock Sparse and stable markowitz portfolios.
\newblock {\em Proceedings of the National Academy of Sciences},
  106(30):12267--12272.

\bibitem[Brown et~al., 2005]{brown2005diversity}
Brown, G., Wyatt, J., Harris, R., and Yao, X. (2005).
\newblock Diversity creation methods: a survey and categorisation.
\newblock {\em Information Fusion}, 6(1):5--20.

\bibitem[Campbell and Shiller, 1988]{campbell1988dividend}
Campbell, J.~Y. and Shiller, R.~J. (1988).
\newblock The dividend-price ratio and expectations of future dividends and
  discount factors.
\newblock {\em The Review of Financial Studies}, 1(3):195--228.

\bibitem[Campbell and Thompson, 2008]{campbell2008predicting}
Campbell, J.~Y. and Thompson, S.~B. (2008).
\newblock Predicting excess stock returns out of sample: Can anything beat the
  historical average?
\newblock {\em The Review of Financial Studies}, 21(4):1509--1531.

\bibitem[Cappelen et~al., 2016]{Cappelen2016}
Cappelen, A.~W., Nielsen, U.~H., Tungodden, B., Tyran, J.-R., and Wengström,
  E. (2016).
\newblock Fairness is intuitive.
\newblock {\em Experimental Economics}, 19(4):727–740.

\bibitem[Claeskens et~al., 2016]{claeskens2016forecast}
Claeskens, G., Magnus, J.~R., Vasnev, A.~L., and Wang, W. (2016).
\newblock The forecast combination puzzle: A simple theoretical explanation.
\newblock {\em International Journal of Forecasting}, 32(3):754--762.

\bibitem[Clark and West, 2007]{clark2007approximately}
Clark, T.~E. and West, K.~D. (2007).
\newblock Approximately normal tests for equal predictive accuracy in nested
  models.
\newblock {\em Journal of Econometrics}, 138(1):291--311.

\bibitem[Clemen, 1989]{clemen1989combining}
Clemen, R.~T. (1989).
\newblock Combining forecasts: A review and annotated bibliography.
\newblock {\em International Journal of Forecasting}, 5(4):559--583.

\bibitem[Coibion and Gorodnichenko, 2015]{coibion2015information}
Coibion, O. and Gorodnichenko, Y. (2015).
\newblock Information rigidity and the expectations formation process: A simple
  framework and new facts.
\newblock {\em American Economic Review}, 105(8):2644--2678.

\bibitem[DeMiguel et~al., 2009]{demiguel2009optimal}
DeMiguel, V., Garlappi, L., and Uppal, R. (2009).
\newblock Optimal versus naive diversification: How inefficient is the 1/n
  portfolio strategy?
\newblock {\em The Review of Financial Studies}, 22(5):1915--1953.

\bibitem[Denk and L{\"o}ffler, 2024]{denk2024predicting}
Denk, S. and L{\"o}ffler, G. (2024).
\newblock Predicting the equity premium with combination forecasts: A
  reappraisal.
\newblock {\em The Review of Asset Pricing Studies}, 14(4):545--577.

\bibitem[Fama and French, 1988]{fama1988dividend}
Fama, E.~F. and French, K.~R. (1988).
\newblock Dividend yields and expected stock returns.
\newblock {\em Journal of Financial Economics}, 22(1):3--25.

\bibitem[Fan et~al., 2012]{Fan2012}
Fan, J., Guo, S.~G., and Hao, N. (2012).
\newblock Variance estimation using refitted cross-validation in ultrahigh
  dimensional regression.
\newblock {\em Journal of the Royal Statistical Society Series B: Statistical
  Methodology}, 74(1):37–65.

\bibitem[Genre et~al., 2013]{genre2013combining}
Genre, V., Kenny, G., Meyler, A., and Timmermann, A. (2013).
\newblock Combining expert forecasts: Can anything beat the simple average?
\newblock {\em International Journal of Forecasting}, 29(1):108--121.

\bibitem[Gigerenzer and Gaissmaier, 2011]{gigerenzer2011}
Gigerenzer, G. and Gaissmaier, W. (2011).
\newblock Heuristic decision making.
\newblock {\em Annual Review of Psychology}, 62:451--482.

\bibitem[Goyal et~al., 2024]{goyal2024comprehensive}
Goyal, A., Welch, I., and Zafirov, A. (2024).
\newblock A comprehensive 2022 look at the empirical performance of equity
  premium prediction.
\newblock {\em The Review of Financial Studies}, 37(11):3490--3557.

\bibitem[Harvey et~al., 1998]{harvey1998tests}
Harvey, D.~I., Leybourne, S.~J., and Newbold, P. (1998).
\newblock Tests for forecast encompassing.
\newblock {\em Journal of Business \& Economic Statistics}, 16(2):254--259.

\bibitem[Huang et~al., 2019]{huang2019saving}
Huang, Y., Miao, J., and Wang, P. (2019).
\newblock Saving china’s stock market?
\newblock {\em IMF Economic Review}, 67:349--394.

\bibitem[Jorion, 1986]{jorion1986bayes}
Jorion, P. (1986).
\newblock Bayes-stein estimation for portfolio analysis.
\newblock {\em Journal of Financial and Quantitative Analysis}, 21(3):279--292.

\bibitem[Kahneman, 2011]{kahneman2011}
Kahneman, D. (2011).
\newblock {\em Thinking, Fast and Slow}.
\newblock Farrar, Straus and Giroux.

\bibitem[Karavias et~al., 2023]{karavias2023structural}
Karavias, Y., Narayan, P.~K., and Westerlund, J. (2023).
\newblock Structural breaks in interactive effects panels and the stock market
  reaction to covid-19.
\newblock {\em Journal of Business \& Economic Statistics}, 41(3):653--666.

\bibitem[Koehler and James, 2010]{koehler2010}
Koehler, D.~J. and James, G. (2010).
\newblock Probability matching and strategy availability.
\newblock {\em Memory \& Cognition}, 38(6):667--676.

\bibitem[Krogh and Vedelsby, 1994]{krogh1994neural}
Krogh, A. and Vedelsby, J. (1994).
\newblock Neural network ensembles, cross validation, and active learning.
\newblock {\em Advances in Neural Information Processing Systems}, 7.

\bibitem[Li et~al., 2025]{li2025dominate}
Li, K., Li, Y., Lyu, C., and Yu, J. (2025).
\newblock How to dominate the historical average.
\newblock {\em The Review of Financial Studies}, 38(10):3086--3116.

\bibitem[Lin et~al., 2018]{lin2018forecasting}
Lin, H., Wu, C., and Zhou, G. (2018).
\newblock Forecasting corporate bond returns with a large set of predictors: An
  iterated combination approach.
\newblock {\em Management Science}, 64(9):4218--4238.

\bibitem[Mankiw and Reis, 2002]{mankiw2002sticky}
Mankiw, N.~G. and Reis, R. (2002).
\newblock Sticky information versus sticky prices: a proposal to replace the
  new keynesian phillips curve.
\newblock {\em The Quarterly Journal of Economics}, 117(4):1295--1328.

\bibitem[Markowitz, 1952]{markowitz1952modern}
Markowitz, H. (1952).
\newblock Modern portfolio theory.
\newblock {\em Journal of Finance}, 7(11):77--91.

\bibitem[Messick, 1993]{messick1993}
Messick, D.~M. (1993).
\newblock Equality as a decision heuristic.
\newblock In Mellers, B.~A. and Baron, J., editors, {\em Psychological
  Perspectives on Justice}, pages 11--31. Cambridge University Press.

\bibitem[Rapach et~al., 2016]{rapach2016short}
Rapach, D.~E., Ringgenberg, M.~C., and Zhou, G. (2016).
\newblock Short interest and aggregate stock returns.
\newblock {\em Journal of Financial Economics}, 121(1):46--65.

\bibitem[Rapach et~al., 2010]{rapach2010out}
Rapach, D.~E., Strauss, J.~K., and Zhou, G. (2010).
\newblock Out-of-sample equity premium prediction: Combination forecasts and
  links to the real economy.
\newblock {\em The Review of Financial Studies}, 23(2):821--862.

\bibitem[Roccazzella et~al., 2022]{roccazzella2022optimal}
Roccazzella, F., Gambetti, P., and Vrins, F. (2022).
\newblock Optimal and robust combination of forecasts via constrained
  optimization and shrinkage.
\newblock {\em International Journal of Forecasting}, 38(1):97--116.

\bibitem[Simon, 1955]{simon1955behavioral}
Simon, H.~A. (1955).
\newblock A behavioral model of rational choice.
\newblock {\em The Quarterly Journal of Economics}, pages 99--118.

\bibitem[Smith and Wallis, 2009]{smith2009simple}
Smith, J. and Wallis, K.~F. (2009).
\newblock A simple explanation of the forecast combination puzzle.
\newblock {\em Oxford Bulletin of Economics and Statistics}, 71(3):331--355.

\bibitem[Stock and Watson, 2004]{stock2004combination}
Stock, J.~H. and Watson, M.~W. (2004).
\newblock Combination forecasts of output growth in a seven-country data set.
\newblock {\em Journal of Forecasting}, 23(6):405--430.

\bibitem[Thomson et~al., 2019]{thomson2019combining}
Thomson, M.~E., Pollock, A.~C., {\"O}nkal, D., and G{\"o}n{\"u}l, M.~S. (2019).
\newblock Combining forecasts: Performance and coherence.
\newblock {\em International Journal of Forecasting}, 35(2):474--484.

\bibitem[Tu and Zhou, 2011]{tu2011markowitz}
Tu, J. and Zhou, G. (2011).
\newblock Markowitz meets {T}almud: {A} combination of sophisticated and naive
  diversification strategies.
\newblock {\em Journal of Financial Economics}, 99(1):204--215.

\bibitem[Tversky and Kahneman, 1974]{tversky1974judgment}
Tversky, A. and Kahneman, D. (1974).
\newblock Judgment under uncertainty: Heuristics and biases: Biases in
  judgments reveal some heuristics of thinking under uncertainty.
\newblock {\em Science}, 185(4157):1124--1131.

\bibitem[Vulkan, 2000]{vulkan2000}
Vulkan, N. (2000).
\newblock An economist’s perspective on probability matching.
\newblock {\em Journal of Economic Surveys}, 14(1):101--118.

\bibitem[Wang and Zhang, 2024]{wang2024non}
Wang, Z. and Zhang, Z. (2024).
\newblock Non-reversal normal markets meet nonexcessive regulations: {MMVaR}
  outperforms naive strategy, {VaR}, and {CVaR}.
\newblock {\em {SSRN} (September 28, 2024)}.

\bibitem[Welch and Goyal, 2008]{welch2008comprehensive}
Welch, I. and Goyal, A. (2008).
\newblock A comprehensive look at the empirical performance of equity premium
  prediction.
\newblock {\em The Review of Financial Studies}, 21(4):1455--1508.

\bibitem[Yuan and Zhou, 2024]{yuan2024naive}
Yuan, M. and Zhou, G. (2024).
\newblock Why naive diversification is not so naive, and how to beat it?
\newblock {\em Journal of Financial and Quantitative Analysis},
  59(8):3601--3632.

\end{thebibliography}
	}

    \clearpage
    \begin{landscape}
    	\begin{figure}
         \centerline{\includegraphics[width=8in]{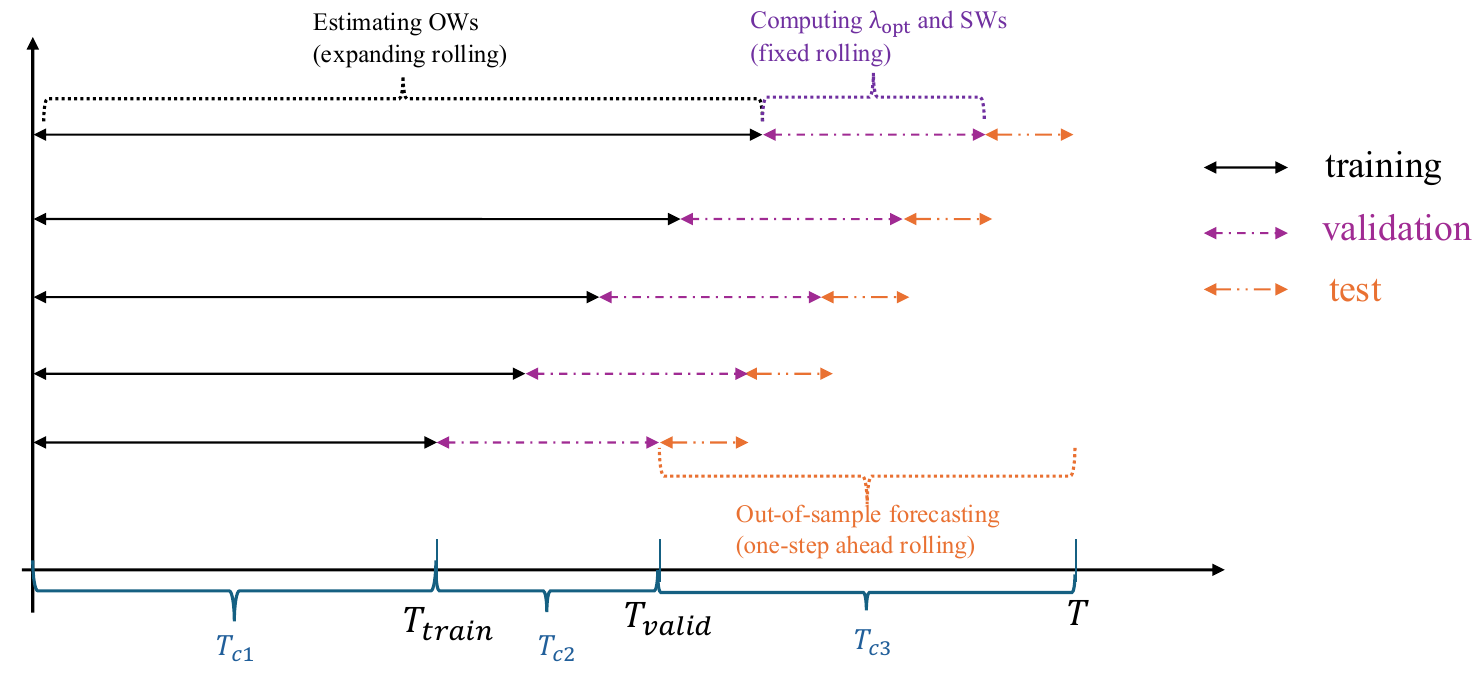}}
        \caption{{\bf Out-of-sample combination forecasting process.} {This figure shows the process of combination forecasting. The whole sample period is divided into three parts. The first $T_{c1}$ months of forecasts are used to estimate OWs, with one-step expanding rolling window; the next $T_{c2}$ months of forecasts are used to compute the optimal shrinkage factor $\lambda_{opt}$ and 2S--ASWs, with one-step fixed window; the final $T_{c3}=T-T_{c2}-T_{c1}$ months (January 2007 to August 2019) are used for out-of-sample evaluation with one-step-ahead combination forecast.}}
        \label{fig:oos forecast process}
    \end{figure}
    \end{landscape}

    \clearpage
    \begin{landscape}
    \begin{figure}
        \centering
        \subfloat[$J=1$]{
            \includegraphics[width=0.4\textwidth]{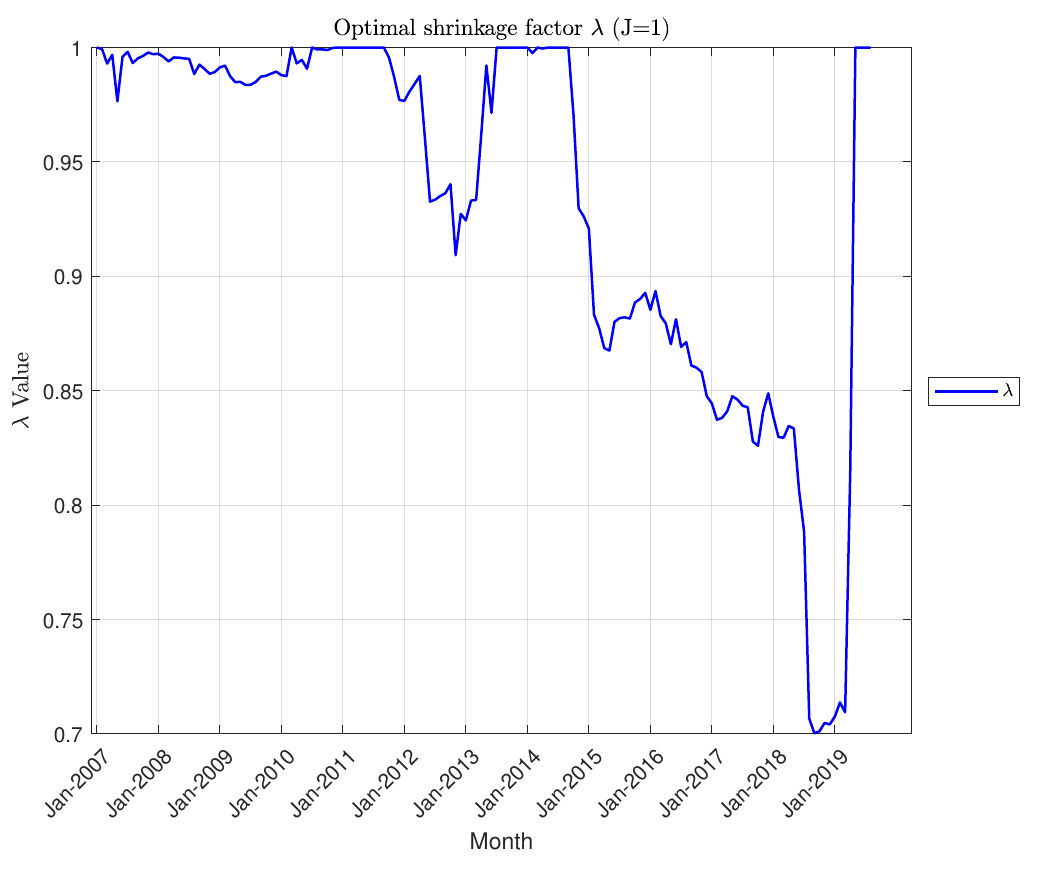}
        }
        \hspace{0.02\textwidth}
        \subfloat[$J=3$]{
            \includegraphics[width=0.4\textwidth]{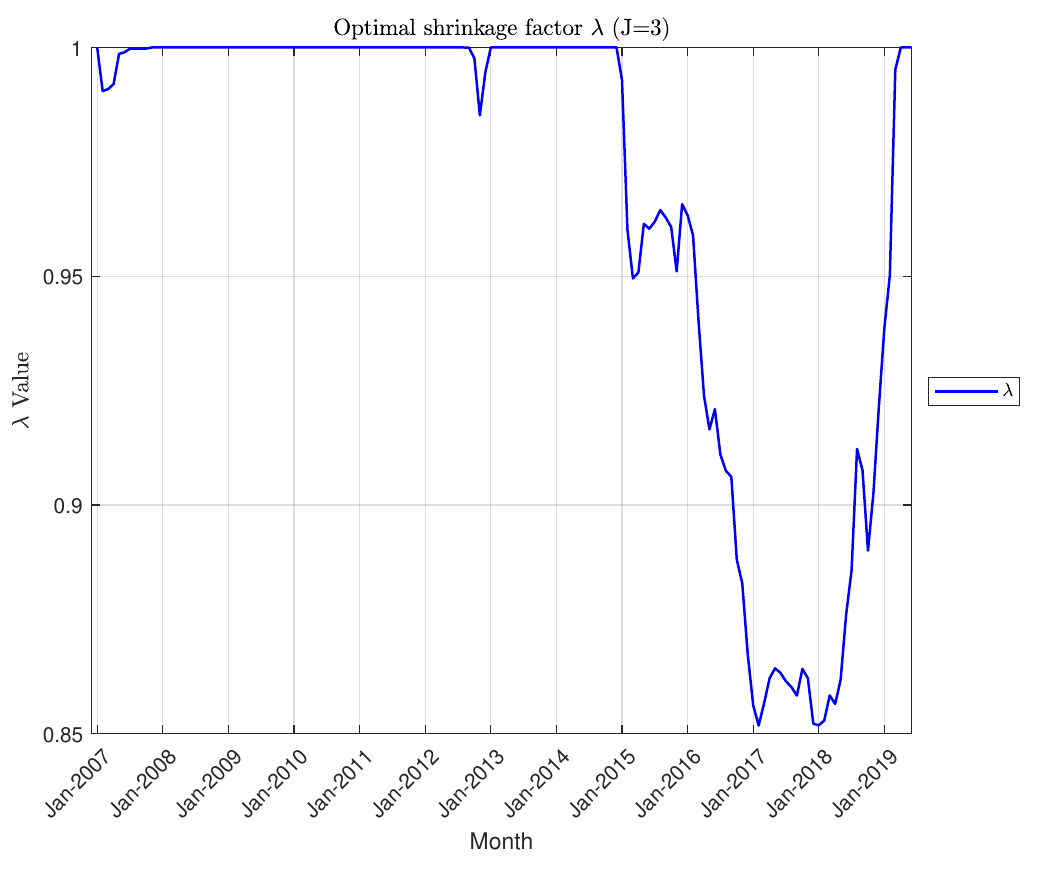}
        }
        \hspace{0.02\textwidth}
        \subfloat[$J=6$]{
            \includegraphics[width=0.4\textwidth]{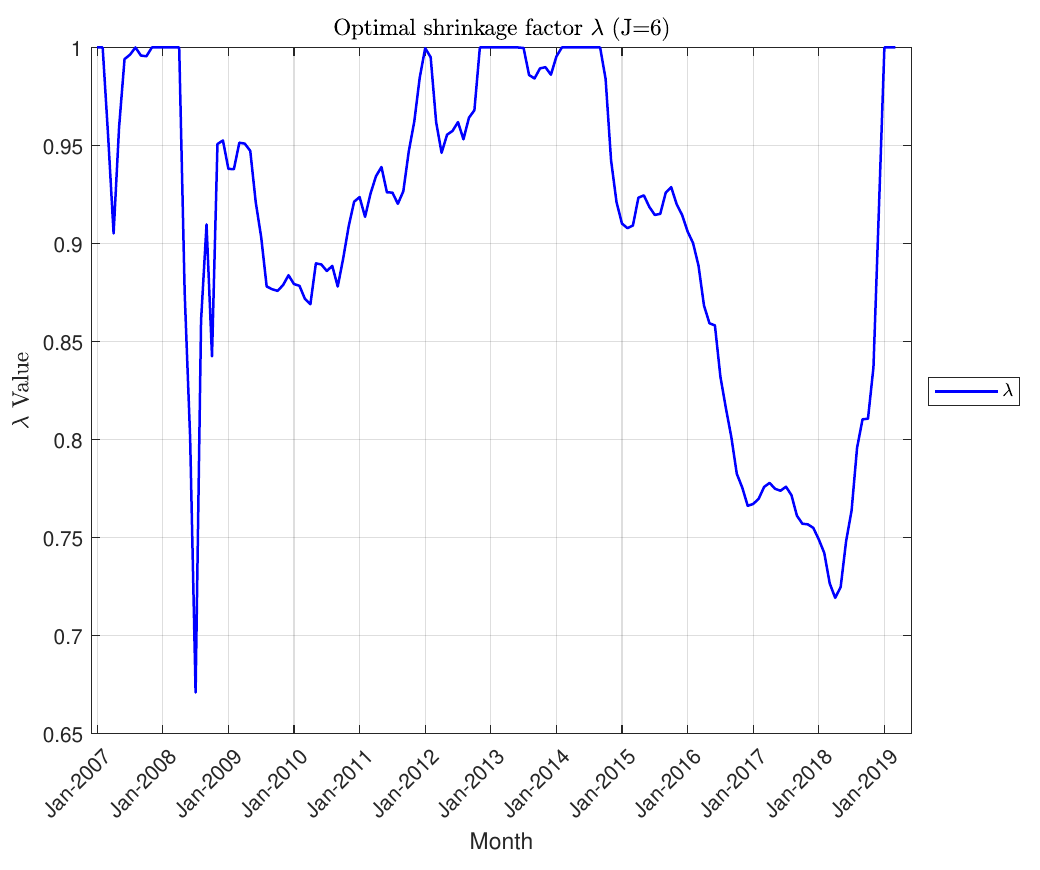}
        }

        \begin{minipage}[t]{0.48\textwidth}
            \centering
            \subfloat[$J=12$]{
                \includegraphics[width=0.9\textwidth]{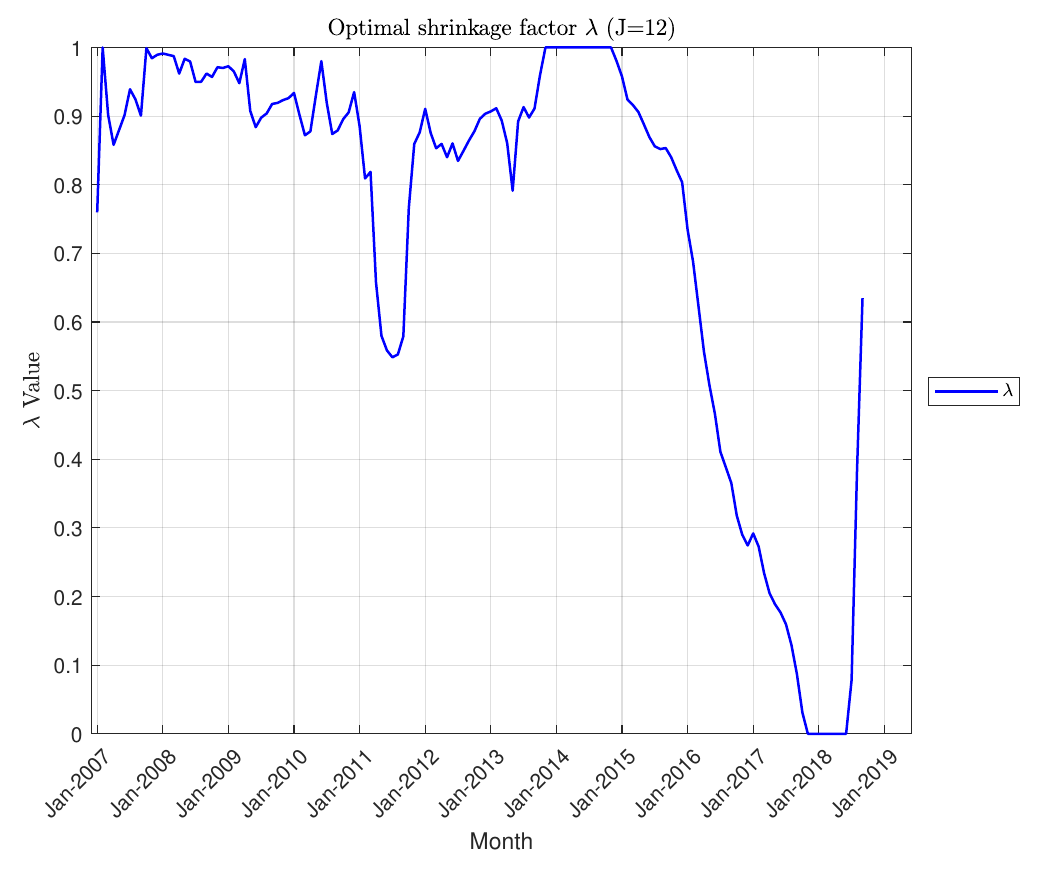}
            }
        \end{minipage}
        \hspace{0.04\textwidth}
        \begin{minipage}[t]{0.48\textwidth}
            \centering
            \subfloat[$J=24$]{
                \includegraphics[width=0.9\textwidth]{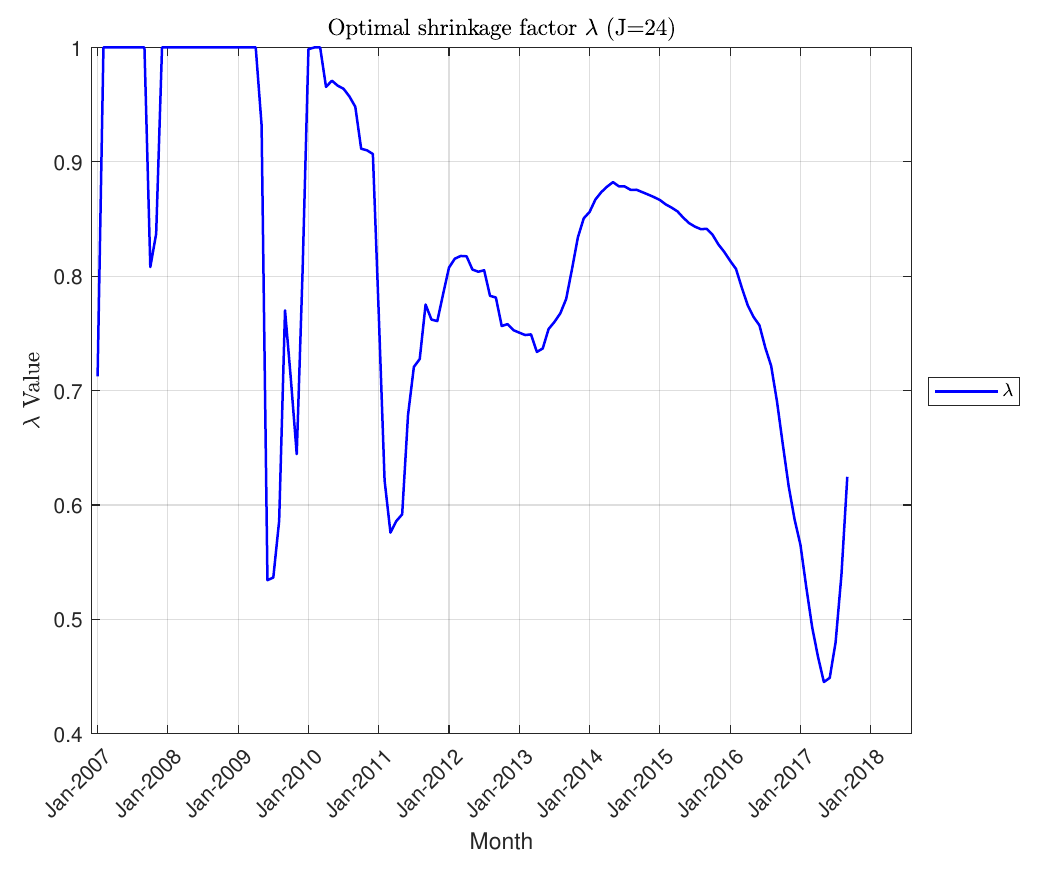}
            }
        \end{minipage}
        \caption{{\bf Optimal shrinkage factor $\lambda_{opt}$ in adaptive shrinkage combination ($T_{c1}=6$, $T_{c2}=126$ when $T_1=42$).} {These figures exhibit the optimal shrinkage factor $\lambda_{opt}$, for $J=1,3,6,12,24$, respectively, which is computed in the validation period $T_{c2}$. If $\lambda_{opt} > 0.5$, 2S--ASWs tend to lean more toward EWs; otherwise, 2S--ASWs tend to lean more toward OWs. The proportions of $\lambda_{opt}=1$ are 19.74\%, 57.33\%, 18.37\%, 9.93\%, 20.93\% and those for $\lambda_{opt}\geq 0.9$ are 66.45\%, 84.67\%, 63.95\%, 44.68\%, 29.46\%, respectively. }}
        \label{fig:lambda_opt}
    \end{figure}
\end{landscape}

    \clearpage
    \begin{landscape}
    	\begin{figure}
         \centerline{\includegraphics[width=8in]{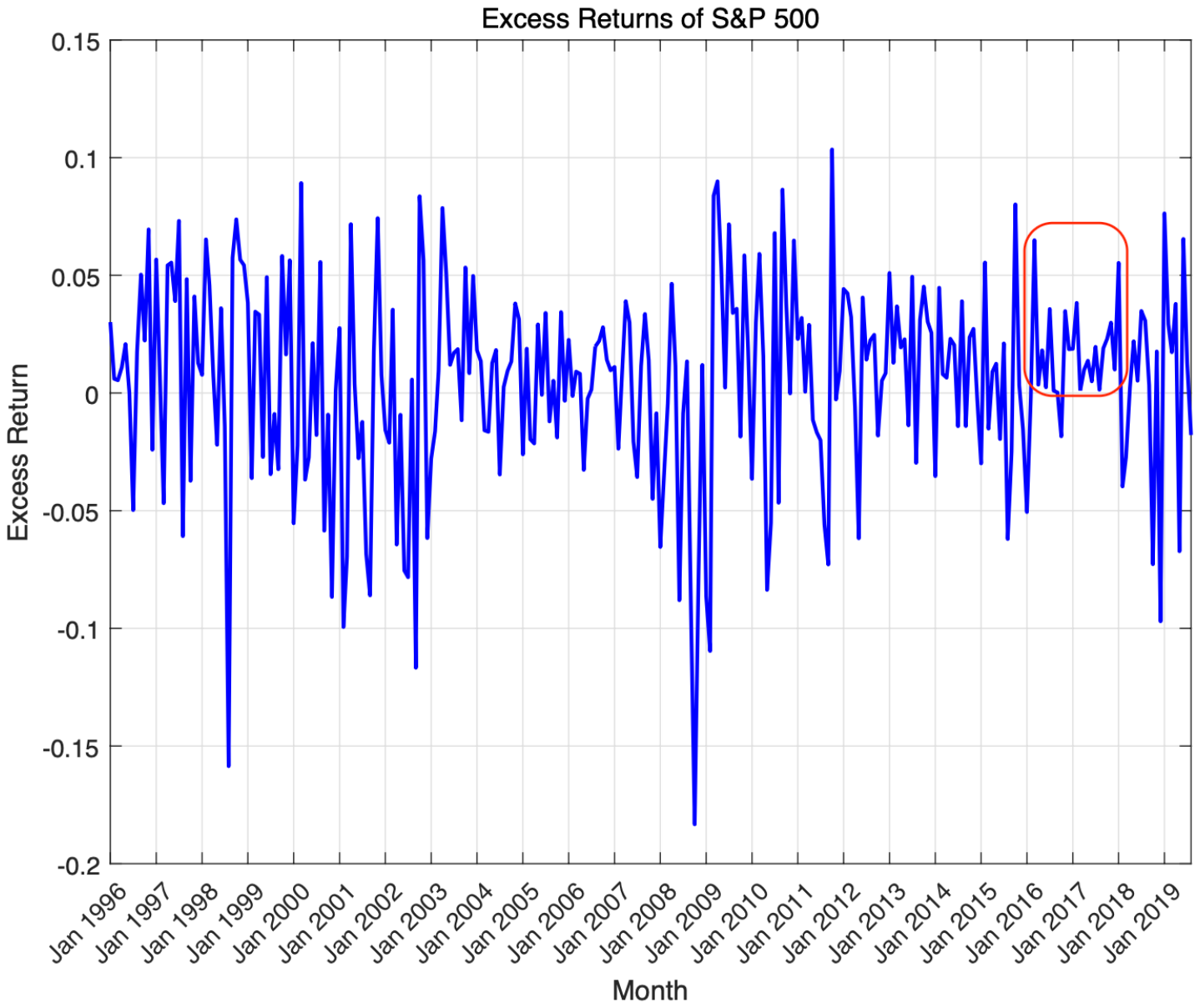}}
        \caption{{\bf Monthly excess returns of S\&P 500.} {This figure plots the time series of monthly excess aggregate market returns on the S\&P 500 index from January 1996 to August 2019.}}
        \label{fig:excess returns}
    \end{figure}
    \end{landscape}

    \clearpage
    \begin{landscape}
    	\begin{figure}
        \centering
        \subfloat[$J=1$]{
            \includegraphics[width=0.4\textwidth]{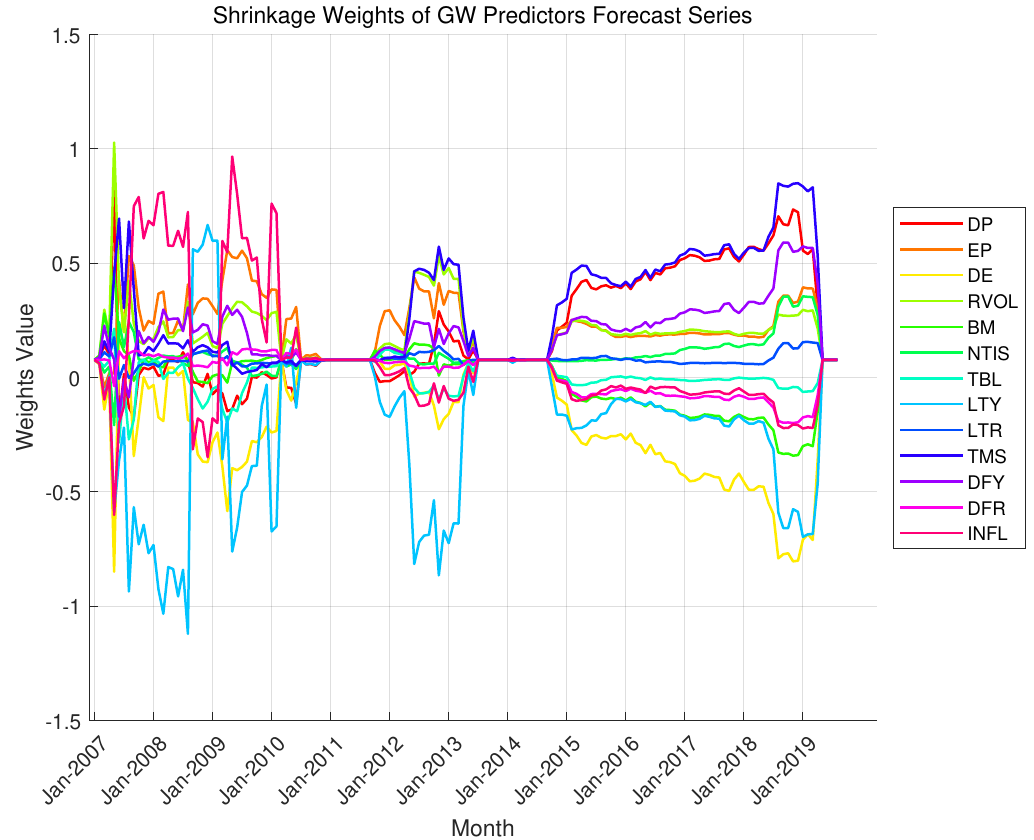}
        }
        \hspace{0.02\textwidth}
        \subfloat[$J=3$]{
            \includegraphics[width=0.4\textwidth]{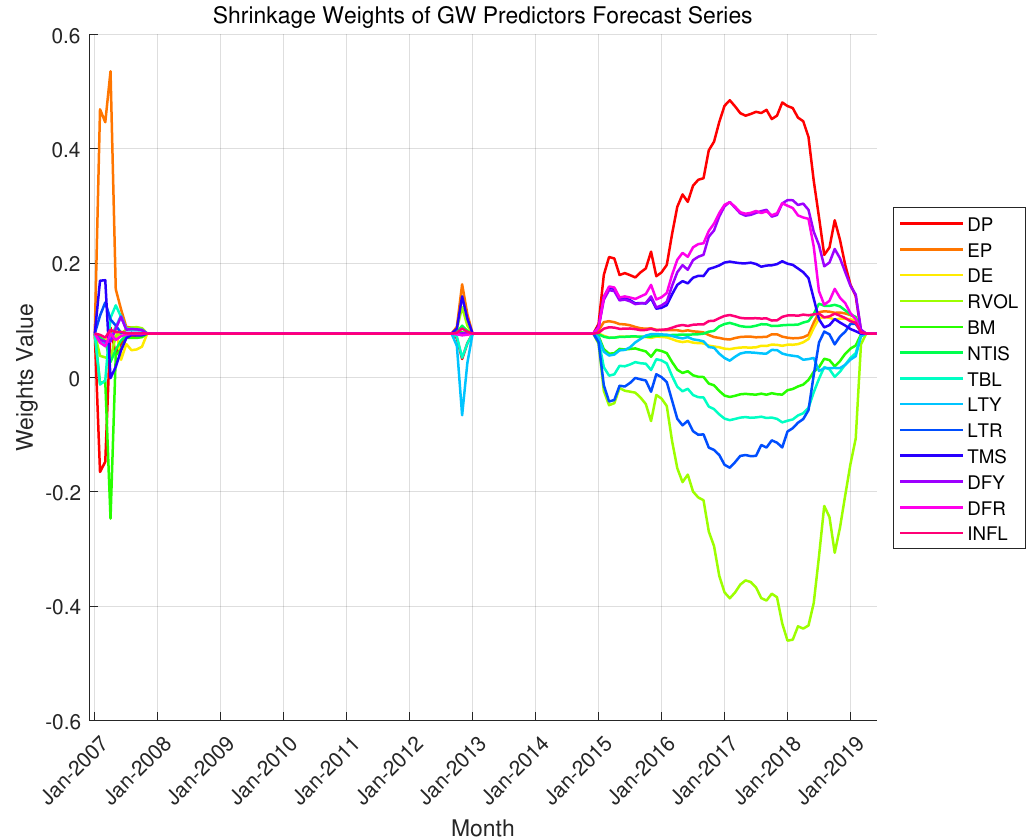}
        }
        \hspace{0.02\textwidth}
        \subfloat[$J=6$]{
            \includegraphics[width=0.4\textwidth]{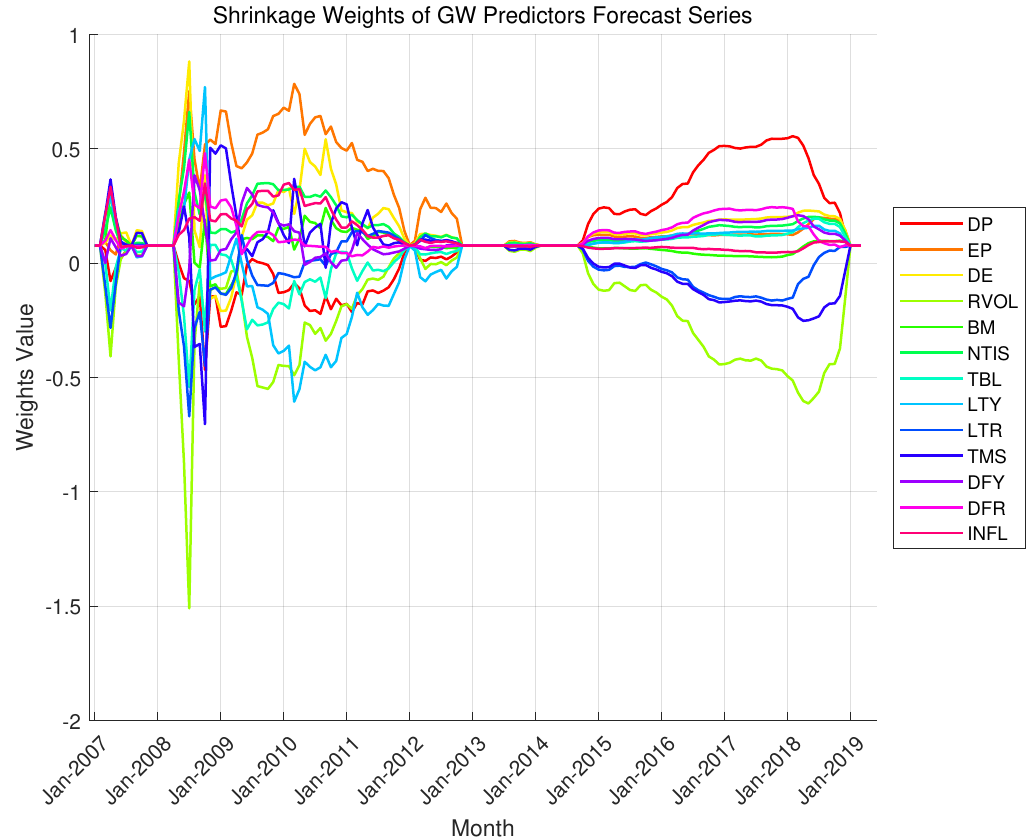}
        }

        \begin{minipage}[t]{0.48\textwidth}
            \centering
            \subfloat[$J=12$]{
                \includegraphics[width=0.9\textwidth]{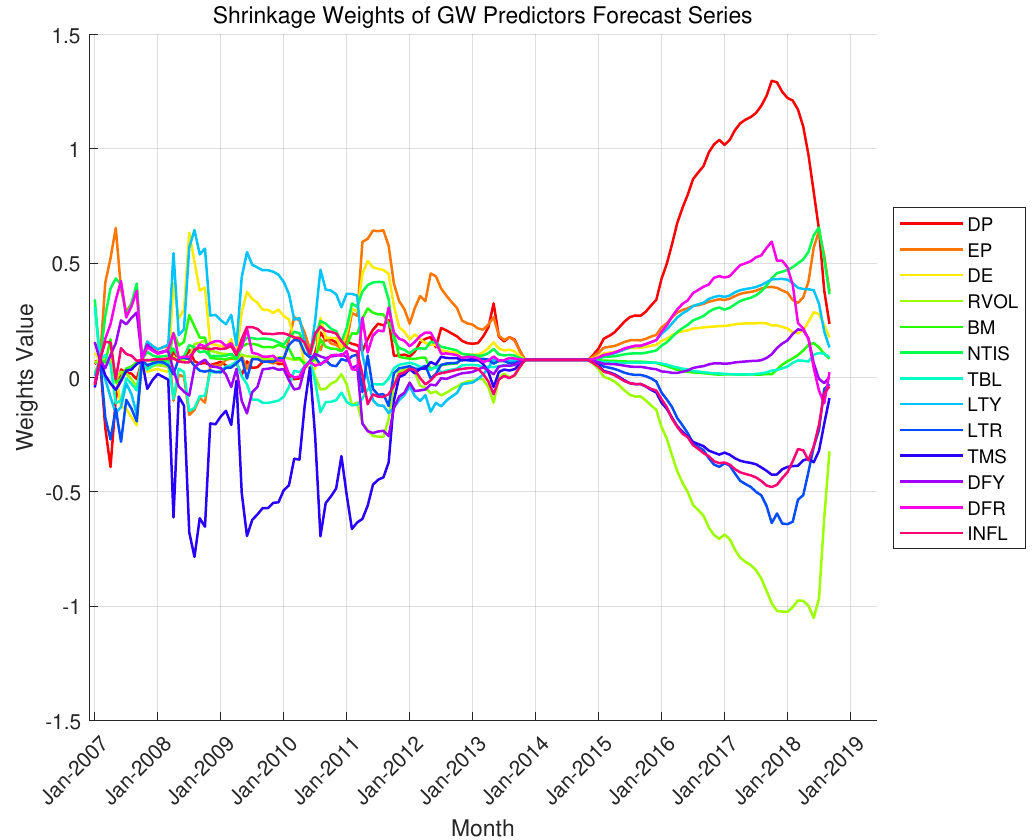}
            }
        \end{minipage}
        \hspace{0.04\textwidth}
        \begin{minipage}[t]{0.48\textwidth}
            \centering
            \subfloat[$J=24$]{
                \includegraphics[width=0.9\textwidth]{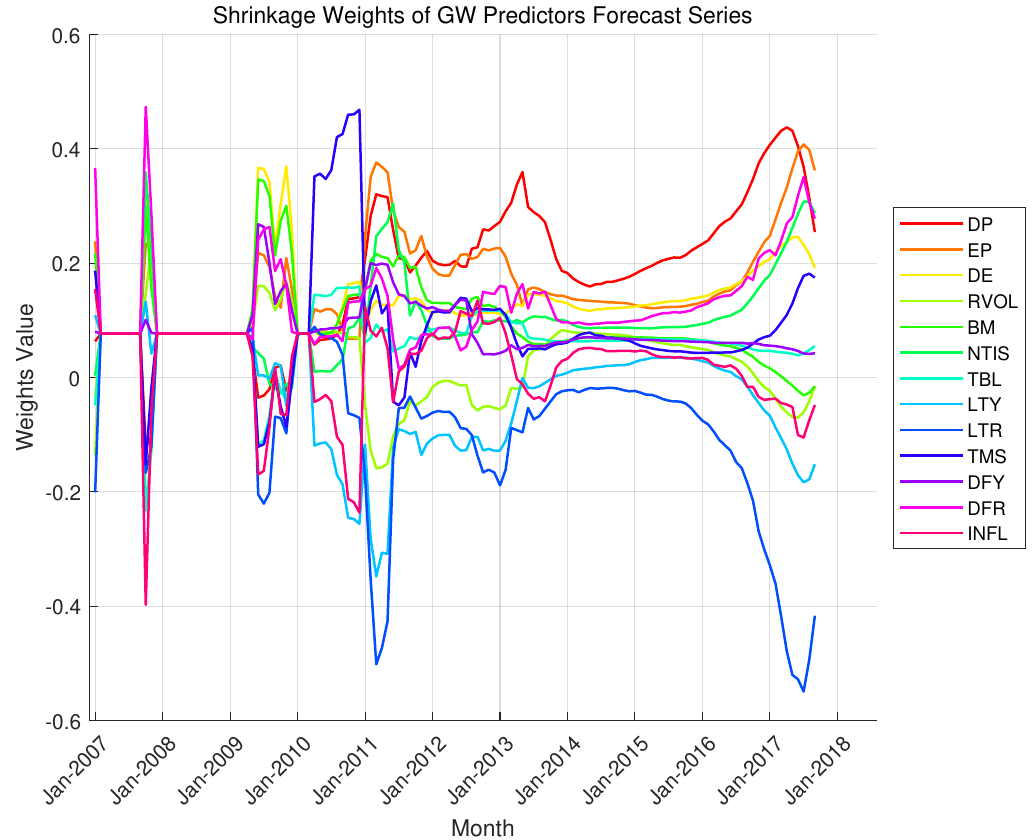}
            }
        \end{minipage}
    		\caption{{\bf Weights in 2S--ASWs assigned to each GW predictor's forecast.} {These figures demonstrate the detailed elements of 2S--ASWs assigned to each GW forecast series, for $J=1,3,6,12,24$, respectively, where 2S--ASWs are computed from the validation period $T_{c2}$. }}
        \label{fig:SWs}
    	\end{figure}
    \end{landscape}

    \clearpage
    \begin{landscape}
    	\begin{figure}
        \centering
        \subfloat[$J=1$]{
            \includegraphics[width=0.4\textwidth]{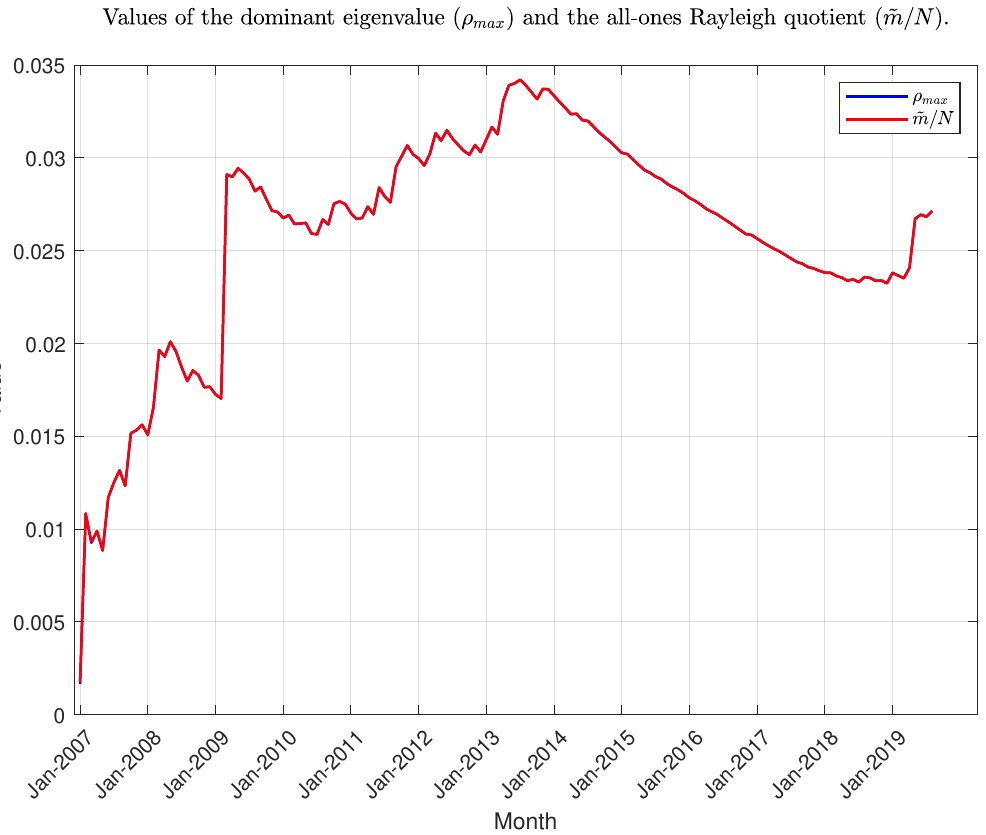}
        }
        \hspace{0.02\textwidth}
        \subfloat[$J=3$]{
            \includegraphics[width=0.4\textwidth]{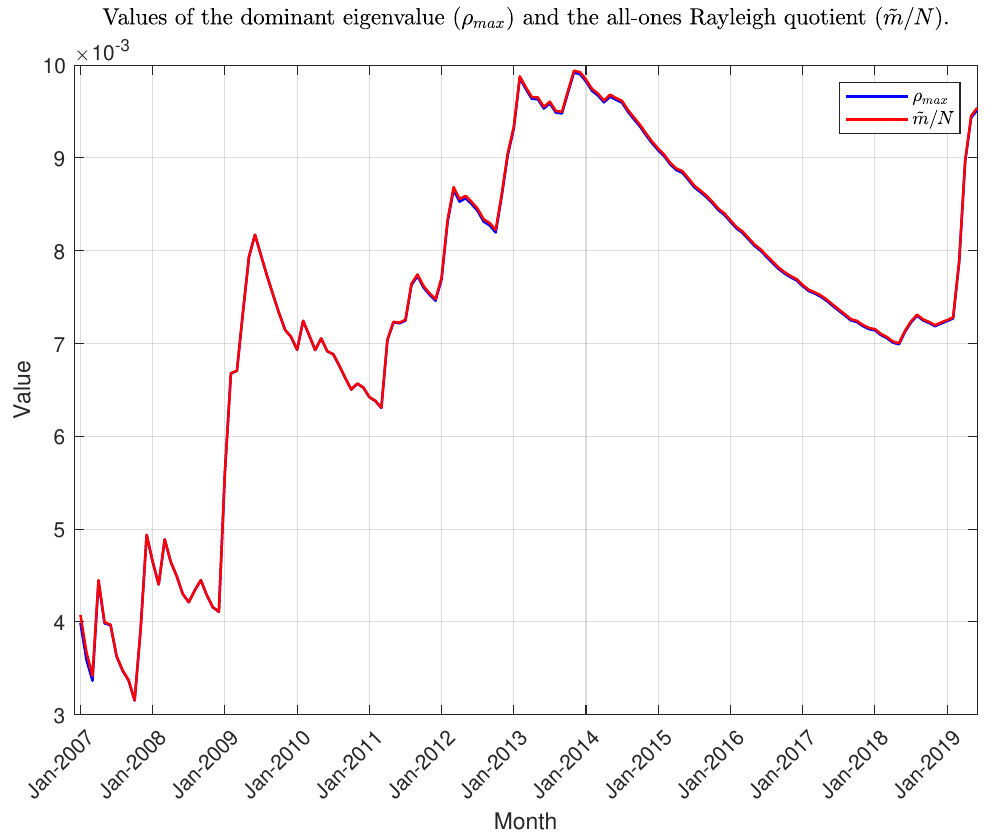}
        }
        \hspace{0.02\textwidth}
        \subfloat[$J=6$]{
            \includegraphics[width=0.4\textwidth]{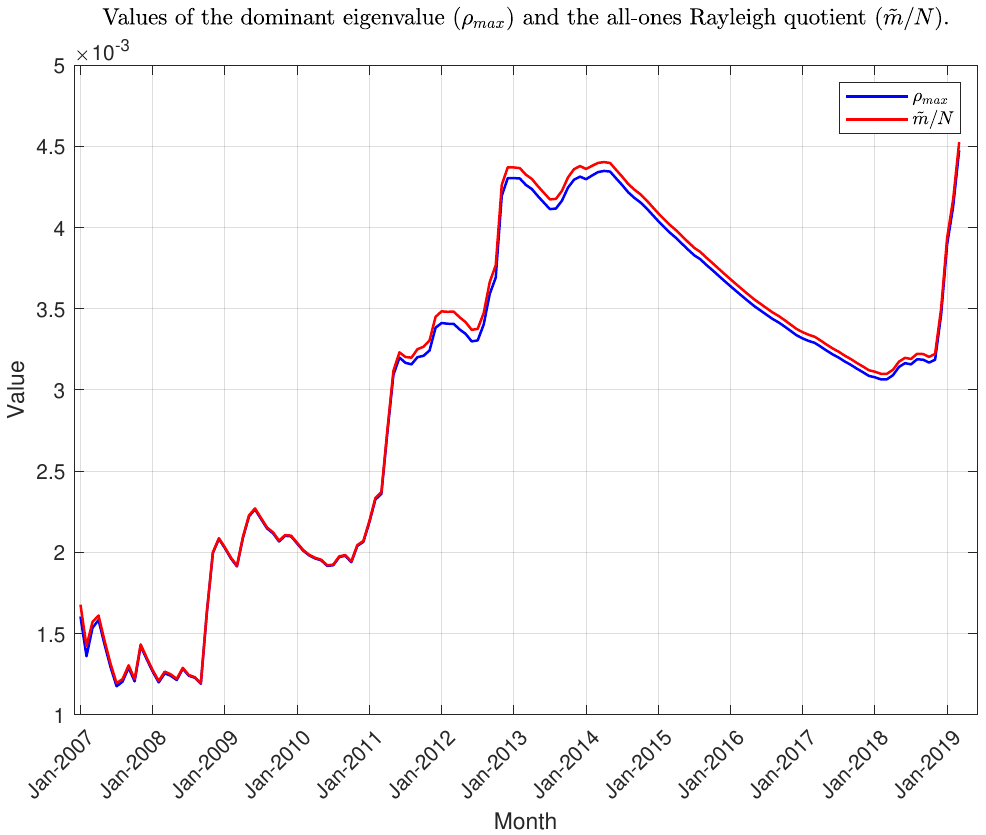}
        }

        \begin{minipage}[t]{0.48\textwidth}
            \centering
            \subfloat[$J=12$]{
                \includegraphics[width=0.9\textwidth]{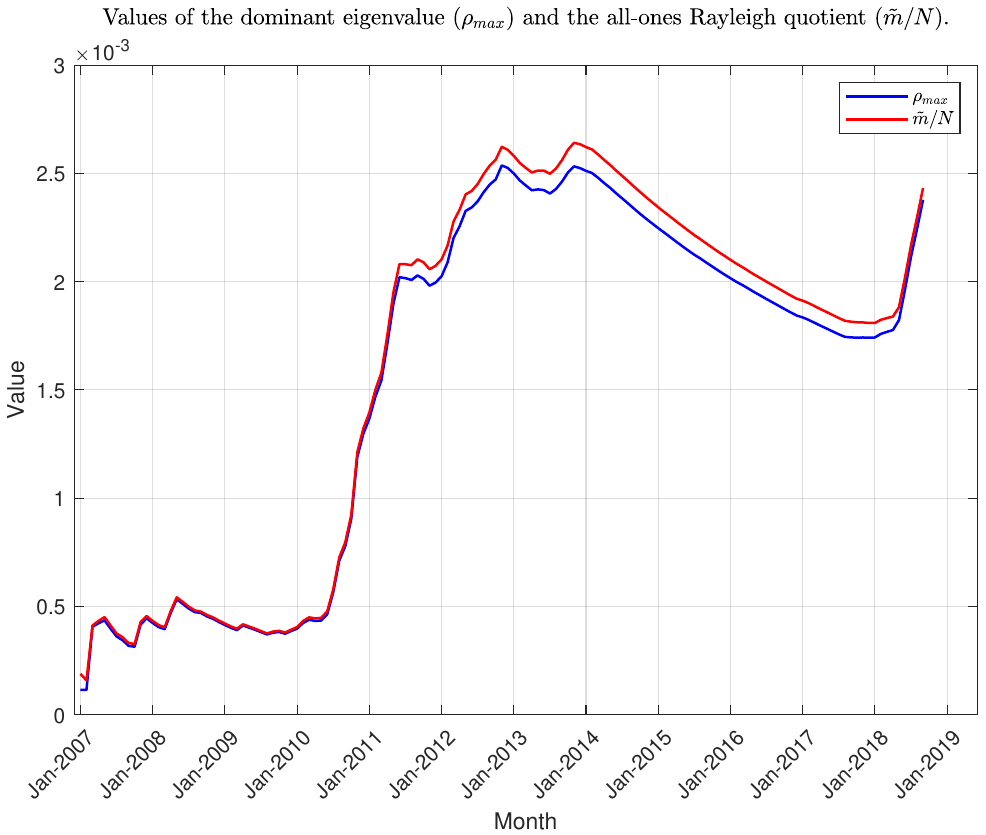}
            }
        \end{minipage}
        \hspace{0.04\textwidth}
        \begin{minipage}[t]{0.48\textwidth}
            \centering
            \subfloat[$J=24$]{
                \includegraphics[width=0.9\textwidth]{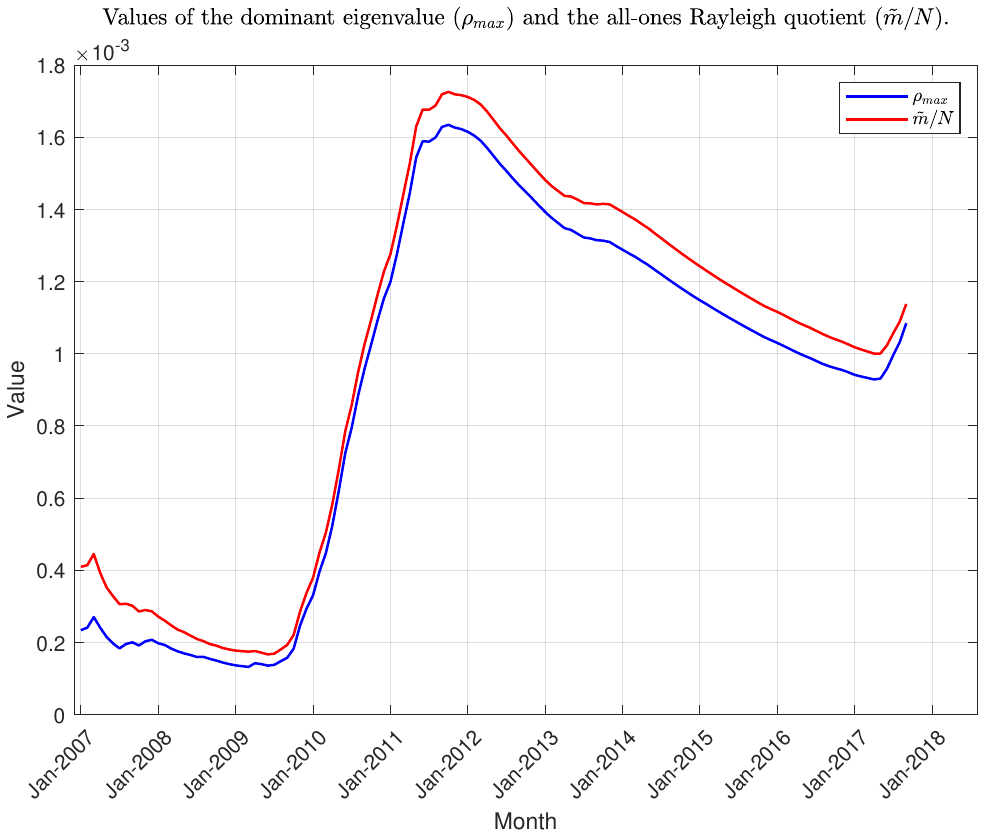}
            }
        \end{minipage}
    		\caption{{\bf The dominant eigenvalue $\rho_{\max}(\mathbf{\Sigma}_1)$ and the all-ones Rayleigh quotient $\frac{\tilde m}{N}$.} {These figures demonstrate the empirically dominant eigenvalue $\rho_{\max}(\mathbf{\Sigma}_1)$, and the all-ones Rayleigh quotient $\frac{\tilde m}{N}$ for $J=1,3,6,12,24$, respectively, where $\rho_{\max}(\mathbf{\Sigma}_1)$ refers to the error covariance matrix in the first-stage optimization problem, $\tilde m=\mathbf{1}^{T}\mathbf{\Sigma}_1\mathbf{1}$ refers to the sum of all the elements in $\mathbf{\Sigma}_1$, and $N=13$ is the number of GW predictors.}}
        \label{fig:delta_1 and m/N}
    	\end{figure}
    \end{landscape}

    \clearpage
    \begin{landscape}
        \begin{table}
        \centering
        \caption{{\bf Summary statistics and correlation matrix.}
        Panel A reports summary statistics (Mean, Median, 1st and 99th percentiles, Std.dev) for the GW predictors using monthly data from January 1996 to August 2019.
        Panel B presents the Pearson correlation coefficients among the predictors. Note: Values of variables marked with $*$ are shown as percentages for readability.}
        \label{tab:summary_and_correlation}

        \textbf{Panel A: Summary Statistics}\\[0.3em]
         \setlength{\tabcolsep}{6pt}
\begin{tabular}{lccccc}
            \toprule
            Variable & Mean & Median & 1st pct & 99th pct & Std. dev. \\
            \midrule
            DP & -4.01 & -3.97 & -4.48 & -3.42 & 0.20 \\
            EP & -3.15 & -3.09 & -4.79 & -2.66 & 0.37 \\
            DE & -0.86 & -0.94 & -1.24 & 1.26 & 0.42 \\
            RVOL & 0.14 & 0.13 & 0.06 & 0.31 & 0.06 \\
            BM & 0.27 & 0.27 & 0.13 & 0.41 & 0.07 \\
            NTIS & 0.00 & 0.01 & -0.05 & 0.03 & 0.02 \\
            TBL$^{*}$ & 2.19 & 1.64 & 0.01 & 6.06 & 2.05 \\
            LTY$^{*}$ & 4.38 & 4.57 & 1.81 & 7.14 & 1.45 \\
            LTR$^{*}$ & 0.61 & 0.69 & -6.42 & 8.40 & 3.02 \\
            TMS$^{*}$ & 2.19 & 2.13 & -0.33 & 4.46 & 1.30 \\
            DFY$^{*}$ & 0.99 & 0.90 & 0.55 & 3.06 & 0.42 \\
            DFR$^{*}$ & 0.03 & 0.05 & -6.17 & 5.77 & 1.77 \\
            INFL$^{*}$ & 0.18 & 0.19 & -0.94 & 0.95 & 0.35 \\
            \bottomrule
        \end{tabular}

        \vspace{1em}
        \textbf{Panel B: Correlation Matrix}\\[0.3em]
        \begin{tabular}{lccccccccccccc}
            \toprule
            Variables & DP & EP & DE & RVOL & BM & NTIS & TBL & LTY & LTR & TMS & DFY & DFR & INFL \\
            \midrule
            DP & 1 & -0.01 & 0.48 & -0.07 & 0.67 & -0.47 & -0.58 & -0.54 & -0.02 & 0.32 & 0.55 & 0.11 & -0.14 \\
            EP & -0.01 & 1 & -0.88 & -0.64 & 0.40 & 0.07 & -0.03 & -0.21 & 0.05 & -0.19 & -0.50 & -0.18 & 0.07 \\
            DE & 0.48 & -0.88 & 1 & 0.52 & -0.02 & -0.30 & -0.26 & -0.07 & -0.03 & 0.32 & 0.72 & 0.17 & -0.13 \\
            RVOL & -0.07 & -0.64 & 0.52 & 1 & -0.25 & 0.18 & -0.06 & 0.21 & 0.05 & 0.33 & 0.43 & 0.12 & -0.08 \\
            BM & 0.67 & 0.40 & -0.02 & -0.25 & 1 & -0.24 & -0.70 & -0.64 & 0.05 & 0.38 & 0.33 & -0.03 & -0.08 \\
            NTIS & -0.47 & 0.07 & -0.30 & 0.18 & -0.24 & 1 & 0.29 & 0.57 & 0.03 & 0.18 & -0.47 & 0.02 & 0.04 \\
            TBL & -0.58 & -0.03 & -0.26 & -0.06 & -0.70 & 0.29 & 1 & 0.77 & 0.02 & -0.71 & -0.39 & -0.09 & 0.12 \\
            LTY & -0.54 & -0.21 & -0.07 & 0.21 & -0.64 & 0.57 & 0.77 & 1 & -0.08 & -0.10 & -0.24 & -0.03 & 0.14 \\
            LTR & -0.02 & 0.05 & -0.03 & 0.05 & 0.05 & 0.03 & 0.02 & -0.08 & 1 & -0.12 & 0.03 & -0.48 & -0.10 \\
            TMS & 0.32 & -0.19 & 0.32 & 0.33 & 0.38 & 0.18 & -0.71 & -0.10 & -0.12 & 1 & 0.34 & 0.10 & -0.03 \\
            DFY & 0.55 & -0.50 & 0.72 & 0.43 & 0.33 & -0.47 & -0.39 & -0.24 & 0.03 & 0.34 & 1 & 0.11 & -0.26 \\
            DFR & 0.11 & -0.18 & 0.17 & 0.12 & -0.03 & 0.02 & -0.09 & -0.03 & -0.48 & 0.10 & 0.11 & 1 & -0.06 \\
            INFL & -0.14 & 0.07 & -0.13 & -0.08 & -0.08 & 0.04 & 0.12 & 0.14 & -0.10 & -0.03 & -0.26 & -0.06 & 1 \\
            \bottomrule
        \end{tabular}
    \end{table}
    \end{landscape}

	\clearpage
    \begin{landscape}
    \renewcommand{\arraystretch}{0.7}
		\begin{longtable}{ccccccccccc}
        \caption{{\bf Individual in-sample forecasting performance.} This table presents the in-sample $R^2$ statistics ($R^2_{IS}$) and $p$-values for the \citet{clark2007approximately} MSFE-adjusted statistic of the GW predictors when $J=1/3/6/12/24$. $R^2_{IS}$ measures the percent reduction in mean squared forecast error (MSFE) for the predictive regression model at a $J$-month horizon forecast relative to the historical mean benchmark forecast. The $p$-values for \citet{clark2007approximately} MSFE-adjusted statistic test the null hypothesis that the historical average MSFE is less than or equal to the predictive regression MSFE against the alternative that the historical average MSFE is greater than the predictive regression MSFE (corresponding to $H_0$: $R^2_{IS} \leq  0$ against $H_{A}$: $R^2_{IS} > 0$). The in-sample evaluation period is from January 1996 to June 2000 for $T_1=54$, and from January 1996 to June 1999 for $T_1=42$. CW $p$-value represents the $p$-value for \citet{clark2007approximately} MSFE-adjusted statistical significance throughput.}\label{tab:ins indvidual} \\
            \toprule
            & \multicolumn{2}{c}{J=1} & \multicolumn{2}{c}{J=3} & \multicolumn{2}{c}{J=6} & \multicolumn{2}{c}{J=12} & \multicolumn{2}{c}{J=24} \\
            \cmidrule(l){2-3} \cmidrule(l){4-5} \cmidrule(l){6-7} \cmidrule(l){8-9} \cmidrule(l){10-11}
            Variables & $R_{IS}^2(\%)$ & \multicolumn{1}{c}{\shortstack{CW\\$p$-value}} & $R_{IS}^2(\%)$ & \multicolumn{1}{c}{\shortstack{CW\\$p$-value}} & $R_{IS}^2(\%)$ & \multicolumn{1}{c}{\shortstack{CW\\$p$-value}} & $R_{IS}^2(\%)$ & \multicolumn{1}{c}{\shortstack{CW\\$p$-value}} & $R_{IS}^2(\%)$ & \multicolumn{1}{c}{\shortstack{CW\\$p$-value}}  \\
            \midrule
            \multicolumn{11}{c}{\textit{Panel A: $T_1=54$}} \\
            DP & 6.41 & 0.01 & 19.30 & 0.00 & 36.80 & 0.00 & 70.77 & 0.00 & 89.10 & 0.00 \\
\bottomrule
        \end{longtable}	
    \end{landscape}

    \clearpage
    \begin{landscape}
    \renewcommand{\arraystretch}{0.7}
		\begin{longtable}{ccccccccccc}
         \caption{{\bf In-sample combination performance.} This table presents the in-sample $R^2$ statistics ($R^2_{CIS}$) and $p$-values for the \citet{clark2007approximately} MSFE-adjusted statistic of in-sample combination methods and PCA methods when $J=1/3/6/12/24$. $R^2_{CIS}$ measures the percent reduction in mean squared forecast error (MSFE) for the predictive regression model at a $J$-month horizon forecast relative to the historical mean benchmark forecast. The $p$-values for \citet{clark2007approximately} MSFE-adjusted statistic test the null hypothesis that the historical average MSFE is less than or equal to the predictive regression MSFE against the alternative that the historical average MSFE is greater than the predictive regression MSFE (corresponding to $H_0$: $R^2_{CIS} \leq 0$ against $H_{A}$: $R^2_{CIS} > 0$). Note: The in-sample results for 2S--ASWs, Lasso, Ridge, and ENet are all evaluated on the same validation period $T_{c2}$ (from July 2000 to December 2006 when $T_1 = 54$, and from July 1999 to December 2006 when $T_1 = 42$); the in-sample results for OWs, \cite{blanc2020bias}, EWs, Median, Trimmed Mean, ASWs--\cite{li2025dominate} and PCA are evaluated on the combined training and validation periods $T_{c1} + T_{c2}$ (from January 1996 to December 2006). Given that DMSPE determines combination weights through rolling calculations of historical mean squared errors, we do not consider this method in the in-sample evaluation.}\label{tab:ins combination} \\
                \toprule
                 & \multicolumn{2}{c}{J=1} & \multicolumn{2}{c}{J=3} & \multicolumn{2}{c}{J=6} & \multicolumn{2}{c}{J=12} & \multicolumn{2}{c}{J=24} \\
\bottomrule
         \end{longtable}
	\end{landscape}

	\clearpage
    \begin{landscape}
    \renewcommand{\arraystretch}{0.7}
		\begin{longtable}{ccccccccccc}
        \caption{{\bf Individual out-of-sample forecasting performance.} This table presents the out-of-sample $R^2$ statistics ($R^2_{IOS}$) and $p$-values for the \citet{clark2007approximately} MSFE-adjusted statistic of individual GW predictors when $J=1/3/6/12/24$. $R^2_{IOS}$ measures the percent reduction in mean squared forecast error (MSFE) for the predictive regression model at a $J$-month horizon forecast relative to the historical mean benchmark forecast. The $p$-values for \citet{clark2007approximately} MSFE-adjusted statistic test the null hypothesis that the historical average MSFE is less than or equal to the predictive regression MSFE against the alternative that the historical average MSFE is greater than the predictive regression MSFE (corresponding to $H_0$: $R^2_{IOS}\leq 0$ against $H_{A}$: $R^2_{IOS} > 0$). The out-of-sample evaluation period is from July 2000 to August 2019 when $T_1=54$, and from July 1999 to August 2019 when $T_1=42$. }\label{tab:oos indvidual} \\
            \toprule
            & \multicolumn{2}{c}{J=1} & \multicolumn{2}{c}{J=3} & \multicolumn{2}{c}{J=6} & \multicolumn{2}{c}{J=12} & \multicolumn{2}{c}{J=24} \\
            \cmidrule(l){2-3} \cmidrule(l){4-5} \cmidrule(l){6-7} \cmidrule(l){8-9} \cmidrule(l){10-11}
            Variables & $R_{IOS}^2(\%)$ & \multicolumn{1}{c}{\shortstack{CW\\$p$-value}} & $R_{IOS}^2(\%)$ & \multicolumn{1}{c}{\shortstack{CW\\$p$-value}} & $R_{IOS}^2(\%)$ & \multicolumn{1}{c}{\shortstack{CW\\$p$-value}} & $R_{IOS}^2(\%)$ & \multicolumn{1}{c}{\shortstack{CW\\$p$-value}} & $R_{IOS}^2(\%)$ & \multicolumn{1}{c}{\shortstack{CW\\$p$-value}}  \\
            \midrule
            \multicolumn{11}{c}{\textit{Panel A: $T_1$=54}} \\
            DP & 0.77 & 0.14 & 0.05 & 0.25 & 0.57 & 0.22 & 2.44 & 0.10 & -2.07 & 0.05 \\
            EP & -2.53 & 0.25 & -13.20 & 0.46 & -25.47 & 0.58 & -34.39 & 0.54 & -62.47 & 0.90 \\
\bottomrule
        \end{longtable}	
    \end{landscape}

    \clearpage

    \begin{landscape}
    \small
    \setlength{\tabcolsep}{0.5pt}
    \renewcommand{\arraystretch}{0.85}
        \begin{longtable}{
        c 
        c 
        *{5}{
            S[table-format=-3.2] 
            @{\hspace{1.2em}}
            S[table-format=1.3]
        }
    } 
            \caption{{\bf Out-of-sample combination performance.} This table presents the out-of-sample $R^2$ statistics ($R^2_{COS}$) and $p$-values for the \citet{clark2007approximately} MSFE-adjusted statistic of combination methods and PCA when $J=1/3/6/12/24$. $R^2_{COS}$ measures the percent reduction in mean squared forecast error (MSFE) for the predictive regression model at a $J$-month horizon forecast relative to the historical mean benchmark forecast. The $p$-values for \citet{clark2007approximately} MSFE-adjusted statistic test the null hypothesis that the historical average MSFE is less than or equal to the predictive regression MSFE against the alternative that the historical average MSFE is greater than the predictive regression MSFE (corresponding to $H_0$: $R^2_{COS}\leq 0$ against $H_{A}$: $R^2_{COS} > 0$). Asterisks attached to the $R^2_{COS}$ entries report the significance levels of the \citet{harvey1998tests} forecast-encompassing test (HLN), which compares the information content of 2S--ASWs with that of the benchmark method in the same row. The out-of-sample evaluation period is from January 2007 to August 2019 when $T_1=54$ and $T_1=42$.}\label{tab:oos combination} \\
                \toprule
               \multicolumn{1}{c}{Estimation Sample} & \multicolumn{1}{c}{Methods} & \multicolumn{2}{c}{J=1} & \multicolumn{2}{c}{J=3} & \multicolumn{2}{c}{J=6} & \multicolumn{2}{c}{J=12} & \multicolumn{2}{c}{J=24} \\
\end{longtable}
        \end{landscape}

        \clearpage
        \begin{landscape}
        \renewcommand{\arraystretch}{0.8}
		\begin{longtable}{ccccccc}
            \caption{{\bf Utility gain.} This table reports the annualized utility gain (in percent) results when $J=1/3/6/12/24$ for an investor with mean-variance preferences and a relative risk-aversion coefficient of three who allocates between equities and risk-free bills using either a historical average or equity risk premium forecast based on the methods discussed. Buy and hold (B\&H) corresponds to the investor passively holding the market portfolio. The evaluation period is from January 2007 to August 2019 when $T_1=54$ and $T_1=42$.}\label{tab:oos combination utility} \\
                \toprule
                \multicolumn{1}{c}{Estimation Sample} & \multicolumn{1}{c}{Methods} & \multicolumn{1}{c}{J=1} & \multicolumn{1}{c}{J=3} & \multicolumn{1}{c}{J=6} & \multicolumn{1}{c}{J=12} & \multicolumn{1}{c}{J=24}\\
                \midrule
                \multicolumn{7}{c}{\textit{Panel A: $T_1=54$}} \\
                \multirow{4}{*}{$T_{c1}=12,T_{c2}=120$} & 2S--ASWs & 6.80 & 4.15 & 4.58 & 5.17 & 2.38 \\
                & Lasso &  2.83 & -0.38 & 0.26 & 4.27 & -8.89 \\
                & Ridge & -5.30 & -4.17 & -0.94 & -4.39 & -9.62 \\
                & ENet & 2.44 & 0.23 & 0.10 & 3.99 & -7.72 \\
                & OWs ($\mathbf{\Sigma}_2$) & -2.49 & 8.98 & 10.97 & 8.53 & 6.94 \\
\bottomrule
            \end{longtable}
        \end{landscape}

        \clearpage
        \begin{landscape}
        \renewcommand{\arraystretch}{0.8}
		\begin{longtable}{ccccccc}
            \caption{{\bf Sharpe ratios.} This table reports the annualized Sharpe ratio results when $J=1/3/6/12/24$ for an investor with mean-variance preferences and a relative risk-aversion coefficient of three who allocates between equities and risk-free bills using either a historical average or equity risk premium forecast based on the methods discussed. Buy-and-hold (B\&H) corresponds to the investor passively holding the market portfolio. The evaluation period is from January 2007 to August 2019 when $T_1=54$ and $T_1=42$.}\label{tab:oos combination Sharp} \\
                \toprule
                \multicolumn{1}{c}{Estimation Sample} & \multicolumn{1}{c}{Methods} & \multicolumn{1}{c}{J=1} & \multicolumn{1}{c}{J=3} & \multicolumn{1}{c}{J=6} & \multicolumn{1}{c}{J=12} & \multicolumn{1}{c}{J=24}\\
                \midrule
                \multicolumn{7}{c}{\textit{Panel A: $T_1=54$}} \\
                \multirow{4}{*}{$T_{c1}=12,T_{c2}=120$} & 2S--ASWs & 0.99 & 0.83 & 0.80 & 0.63 & 0.58 \\
                & Lasso & 0.79 & 0.45 & 0.37 & 0.66 & 0.19 \\
                & Ridge & 0.11 & 0.07 & 0.25 & 0.19 & 0.18 \\
                & ENet & 0.79 & 0.59 & 0.36 & 0.65 & 0.17 \\
                & OWs ($\mathbf{\Sigma}_2$) & 0.32 & 1.03 & 1.03 & 0.85 & 0.78 \\
\bottomrule
            \end{longtable}
        \end{landscape}

	\clearpage
	\appendix

	\begin{center}
		{\bf \Large Internet Appendix}
	\end{center}

	\pagenumbering{arabic}		
	\renewcommand*{\thepage}{Appendix-\arabic{page}}

	\counterwithin*{table}{section}
	\counterwithin*{figure}{section}
	\setcounter{table}{0}
	\setcounter{figure}{0}
	\setcounter{equation}{0}
	\renewcommand{\thetable}{\Alph{section}.\arabic{table}}
	\renewcommand{\thefigure}{\Alph{section}.\arabic{figure}}
	\renewcommand{\theequation}{\Alph{section}.\arabic{equation}}

\newcommand{\configureappendix}[1]{%
    \gdef\appendixletter{#1}
    \setcounter{equation}{0}%
    \setcounter{figure}{0}%
    \setcounter{table}{0}%
    \setcounter{theorem}{0}%
    \renewcommand{\thesubsection}{#1.\arabic{subsection}}%
    \renewcommand{\thesubsubsection}{#1.\arabic{subsection}.\arabic{subsubsection}}%
    \renewcommand{\theequation}{\appendixprefix.\arabic{equation}}%
    \renewcommand{\thefigure}{\appendixprefix.\arabic{figure}}%
    \renewcommand{\thetable}{\appendixprefix.\arabic{table}}%
    \renewcommand{\thetheorem}{\appendixprefix.\arabic{theorem}}%
    \renewcommand{\thelemma}{\appendixprefix.\arabic{theorem}}%
    \renewcommand{\theremark}{\appendixprefix.\arabic{theorem}}%
    \renewcommand{\thecorollary}{\appendixprefix.\arabic{theorem}}%
    \renewcommand{\theproposition}{\appendixprefix.\arabic{theorem}}%
}

\newcommand{\appendixprefix}{%
    \ifnum\value{subsubsection}>0
        \appendixletter.\arabic{subsection}.\arabic{subsubsection}%
    \else
        \appendixletter.\arabic{subsection}%
    \fi
}

\newcommand{\resettheoremcounter}{%
    \setcounter{theorem}{0}%
}

\newcommand{\resettheoremcountersubsub}{%
    \setcounter{theorem}{0}%
}

\section{Proofs}\label{sec: appendix proof}
\configureappendix{A}

\subsection{Derivation of the Optimal Shrinkage Parameter}\label{sec: appendix derivation}

We provide the detailed derivation of the closed-form solution for the optimal shrinkage parameter $\lambda_{opt}$ used in the adaptive shrinkage combination methodology.

\begin{lemma}[Accuracy-Diversity Decomposition]\label{lemma:decomposition}
The mean squared forecast error of the combined forecast admits the matrix representation:
\begin{equation}\label{eq:matrix_msfe}
\text{MSFE}_{comb} = \mathbf{w}^T\mathbf{S} - \frac{1}{2}\mathbf{w}^T\mathbf{D}\mathbf{w},
\end{equation}
where $\mathbf{w} = (w_1,w_2,\ldots,w_N)^T$ is the weight vector, $\mathbf{S} = (\text{MSFE}_1, \text{MSFE}_2, \ldots, \text{MSFE}_N)^T$ is the vector of individual forecast errors, and $\mathbf{D}$ is the $N \times N$ diversity matrix with elements $D_{ij} = \text{Div}_{ij}$ for $i \neq j$ and $D_{ii} = 0$.
\end{lemma}

\begin{proof}
Starting from a weighted combined forecast MSFE at observation $h$:
\begin{align}
    \sum_{i=1}^{N}w_i \left(f_{ih}-y_h\right)^2 &= \sum_{i=1}^{N}w_i \left[\left(f_{ih}-f_{ch}\right)+\left(f_{ch}-y_h\right)\right]^2 \notag \\
    &= \sum_{i=1}^{N}w_i\left[\left(f_{ih}-f_{ch}\right)^2 + 2\left(f_{ih}-f_{ch}\right)\left(f_{ch}-y_h\right) + \left(f_{ch}-y_h\right)^2\right] \notag \\
    &= \sum_{i=1}^{N}w_i\left(f_{ih}-f_{ch}\right)^2 + 2\sum_{i=1}^{N}w_i\left(f_{ih}-f_{ch}\right)\left(f_{ch}-y_h\right)+\sum_{i=1}^{N}w_i\left(f_{ch}-y_h\right)^2,
\label{eq:one step MSFEcomb}
\end{align}

\noindent where $f_{ch}=\sum_{i=1}^{N}w_if_{ih}$ represents the combined forecast. Note that
\begin{align*}
    \sum_{i=1}^{N}w_i\left(f_{ih}-f_{ch}\right)\left(f_{ch}-y_h\right) &= \left(f_{ch}-y_h\right)\sum_{i=1}^{N}w_i\left(f_{ih}-f_{ch}\right)  \\
    &= \left(f_{ch}-y_h\right) \left(\sum_{i=1}^{N}w_i f_{ih} - f_{ch}\sum_{i=1}^{N}w_i\right)\\
    &= \left(f_{ch}-y_h\right)\left(f_{ch}-f_{ch}\right) \\
    &= 0,
\end{align*}
and
\begin{align*}
\sum_{i=1}^{N}w_i\left(f_{ch}-y_h\right)^2 = \left(f_{ch}-y_h\right)^2\sum_{i=1}^{N}w_i = \left(f_{ch}-y_h\right)^2.
\end{align*}

\noindent Therefore, equation~(\ref{eq:one step MSFEcomb}) can be simplified to
\begin{align}
    \sum_{i=1}^{N}w_i \left(f_{ih}-y_h\right)^2 &= \sum_{i=1}^{N}w_i\left(f_{ih}-f_{ch}\right)^2 + \left(f_{ch}-y_h\right)^2.
\label{eq:simple one step MSFEcomb}
\end{align}

\noindent Rearranging the terms in equation~(\ref{eq:simple one step MSFEcomb}) yields:
\begin{align*}
    \left(f_{ch}-y_h\right)^2 &= \left(\sum_{i=1}^{N}w_i f_{ih}-y_h\right)^2\\
    &= \sum_{i=1}^{N}w_i \left(f_{ih}-y_h\right)^2-\sum_{i=1}^{N}w_i\left(f_{ih}-f_{ch}\right)^2,
\end{align*}
and naturally
\begin{align}
    \frac{1}{H}\sum_{h=1}^{H}\left(\sum_{i=1}^{N}w_i f_{ih}-y_h\right)^2=\frac{1}{H}
    \sum_{h=1}^{H}\left[\sum_{i=1}^{N}w_i \left(f_{ih}-y_h\right)^2-\sum_{i=1}^{N}w_i\left(f_{ih}-f_{ch}\right)^2\right].
\label{eq:1st decomposition}
\end{align}

Next, we illustrate $\sum_{i=1}^{N}w_i\left(f_{ih}-f_{ch}\right)^2=\sum_{i=1}^{N-1}\sum_{j>i}^{N}w_iw_j\left(f_{ih}-f_{jh}\right)^2$.

\begin{align}
    \sum_{i=1}^{N}w_i\left(f_{ih}-f_{ch}\right)^2 &= \sum_{i=1}^{N} w_i(f_{ih}-\sum_{k=1}^{N}w_kf_{kh})^2 \notag \\
    &= \sum_{i=1}^{N}w_i[\sum_{k=1}^{N}w_k(f_{ih}-f_{kh})]^{2} \notag \\
    &= \sum_{i=1}^{N}w_i[\sum_{k=1}^{N}w_k(f_{ih}-f_{kh})][\sum_{m=1}^{N}w_m(f_{ih}-f_{mh})] \notag \\
    &= \sum_{i=1}^{N}\sum_{k=1}^{N}\sum_{m=1}^{N}w_iw_kw_m(f_{ih}-f_{kh})(f_{ih}-f_{mh})
\label{eq:decomposition second term}
\end{align}

According to the identity $(a-b)(a-c)=\frac{1}{2}[(a-b)^2+(a-c)^2-(b-c)^2]$, and letting $a=f_{ih}$, $b=f_{kh}$, $c=f_{mh}$, equation~(\ref{eq:decomposition second term}) becomes:

{\allowdisplaybreaks
\begin{equation}
\label{eq:three term}
\begin{split}
&\sum_{i=1}^{N}\sum_{k=1}^{N}\sum_{m=1}^{N}w_iw_kw_m(f_{ih}-f_{kh})(f_{ih}-f_{mh}) \\
    &\quad = \frac{1}{2}\sum_{i=1}^{N}\sum_{k=1}^{N}\sum_{m=1}^{N}w_iw_kw_m \left[ (f_{ih}-f_{kh})^2 + (f_{ih}-f_{mh})^2 - (f_{kh}-f_{mh})^2 \right].
\end{split}
\end{equation}

\noindent We compute each term on the right side of equation~(\ref{eq:three term}) separately. Given that $i$, $k$, and $m$ are mutually independent,
\begin{align*}
    \sum_{i=1}^{N}\sum_{k=1}^{N}\sum_{m=1}^{N}w_iw_kw_m(f_{ih}-f_{kh})^2 &= (\sum_{m=1}^{N}w_m)\sum_{i=1}^{N}\sum_{k=1}^{N}w_iw_k(f_{ih}-f_{kh})^2 \\
    &= \sum_{i=1}^{N}\sum_{k=1}^{N}w_iw_k(f_{ih}-f_{kh})^2
    \equiv \sum_{i=1}^{N}\sum_{j=1}^{N}w_iw_j(f_{ih}-f_{jh})^2 \\
    \sum_{i=1}^{N}\sum_{k=1}^{N}\sum_{m=1}^{N}w_iw_kw_m(f_{ih}-f_{mh})^2 &= (\sum_{k=1}^{N}w_k)\sum_{i=1}^{N}\sum_{m=1}^{N}w_iw_m(f_{ih}-f_{mh})^2 \\
    &= \sum_{i=1}^{N}\sum_{m=1}^{N}w_iw_m(f_{ih}-f_{mh})^2
    \equiv  \sum_{i=1}^{N}\sum_{j=1}^{N}w_iw_j(f_{ih}-f_{jh})^2 \\
    \sum_{i=1}^{N}\sum_{k=1}^{N}\sum_{m=1}^{N}w_iw_kw_m(f_{kh}-f_{mh})^2 &= (\sum_{i=1}^{N}w_i)\sum_{k=1}^{N}\sum_{m=1}^{N}w_kw_m(f_{kh}-f_{mh})^2 \\
    &= \sum_{k=1}^{N}\sum_{m=1}^{N}w_kw_m(f_{kh}-f_{mh})^2
    \equiv  \sum_{i=1}^{N}\sum_{j=1}^{N}w_iw_j(f_{ih}-f_{jh})^2
\end{align*}
}

\noindent Therefore,
\begin{align}
    \sum_{i=1}^{N}w_i(f_{ih}-f_{ch})^2 = &\sum_{i=1}^{N}\sum_{k=1}^{N}\sum_{m=1}^{N}w_iw_kw_m(f_{ih}-f_{kh})(f_{ih}-f_{mh}) \notag \\
    &= \frac{1}{2}[\sum_{i=1}^{N}\sum_{j=1}^{N}w_iw_j(f_{ih}-f_{jh})^2+\sum_{i=1}^{N}\sum_{j=1}^{N}w_iw_j(f_{ih}-f_{jh})^2-\sum_{i=1}^{N}\sum_{j=1}^{N}w_iw_j(f_{ih}-f_{jh})^2] \notag \\
    &= \frac{1}{2}\sum_{i=1}^{N}\sum_{j=1}^{N}w_iw_j(f_{ih}-f_{jh})^2 \equiv \sum_{i=1}^{N-1}\sum_{j>i}^{N}w_iw_j(f_{ih}-f_{jh})^2.
\label{eq:f_ch term}
\end{align}

Finally, equations~(\ref{eq:1st decomposition}) and (\ref{eq:f_ch term}) jointly yield:
\begin{align*}
    \frac{1}{H}\sum_{h=1}^{H}\left(\sum_{i=1}^{N}w_i f_{ih}-y_h\right)^2 &= \frac{1}{H}
    \sum_{h=1}^{H}\left[\sum_{i=1}^{N}w_i \left(f_{ih}-y_h\right)^2-\sum_{i=1}^{N}w_i\left(f_{ih}-f_{ch}\right)^2\right] \\
    &= \frac{1}{H} \sum_{h=1}^{H}\left[\sum_{i=1}^{N}w_i \left(f_{ih}-y_h\right)^2-\sum_{i=1}^{N-1}\sum_{j>i}^{N}w_i w_j\left(f_{ih}-f_{jh}\right)^2\right] \\
    &= \mathbf{w}^T\mathbf{S} - \frac{1}{2}\mathbf{w}^T\mathbf{D}\mathbf{w},
\end{align*}

where
\begin{align*}
\mathbf{S}= \begin{bmatrix}
        \frac{1}{H}\sum_{h=1}^{H}(f_{1h}-y_h)^2\\
        \frac{1}{H}\sum_{h=1}^{H}(f_{2h}-y_h)^2\\
        \vdots \\
        \frac{1}{H}\sum_{h=1}^{H}(f_{Nh}-y_h)^2
    \end{bmatrix} = \begin{bmatrix}
        \text{MSFE}_1 \\
        \text{MSFE}_2 \\
        \vdots \\
        \text{MSFE}_N
    \end{bmatrix},
\end{align*}
and
\begin{align*}
\mathbf{D} &=
    \begin{bmatrix}
        0 & \frac{1}{H}\sum_{h=1}^{H}(f_{1h}-f_{2h})^2 & \frac{1}{H}\sum_{h=1}^{H}(f_{1h}-f_{3h})^2 & \cdots & \frac{1}{H}\sum_{h=1}^{H}(f_{1h}-f_{Nh})^2 \\
        \frac{1}{H}\sum_{h=1}^{H}(f_{2h}-f_{1h})^2 & 0 & \frac{1}{H}\sum_{h=1}^{H}(f_{2h}-f_{3h})^2 & \cdots & \frac{1}{H}\sum_{h=1}^{H}(f_{2h}-f_{Nh})^2 \\
        \vdots & \vdots & \vdots & \ddots & \vdots \\
        \frac{1}{H}\sum_{h=1}^{H}(f_{Nh}-f_{1h})^2 & \frac{1}{H}\sum_{h=1}^{H}(f_{Nh}-f_{2h})^2 & \frac{1}{H}\sum_{h=1}^{H}(f_{Nh}-f_{3h})^2 & \cdots & 0
    \end{bmatrix} \\
    &= \begin{bmatrix}
        0 & \text{Div}_{12} & \text{Div}_{13} & \cdots & \text{Div}_{1N} \\
        \text{Div}_{21} & 0 & \text{Div}_{23} & \cdots & \text{Div}_{2N} \\
        \vdots & \vdots & \vdots & \ddots & \vdots \\
        \text{Div}_{N1} & \text{Div}_{N2} & \text{Div}_{N3} & \cdots & 0
    \end{bmatrix}.
\end{align*}
\end{proof}

\begin{theorem}[Optimal Shrinkage Parameter]\label{theorem:optimal_lambda}
The optimal shrinkage parameter that minimizes the combined forecast MSFE is given by:
\begin{equation*}\label{eq:optimal_lambda}
\lambda_{opt} = \max\left\{0, \min\left\{1, \frac{\left(\frac{1}{N}\mathbf{1}^T - \left(\hat{\mathbf{w}}^{ow}\right)^T\right)\left(\mathbf{S} - \mathbf{D}\hat{\mathbf{w}}^{ow}\right)}{\left(\frac{1}{N}\mathbf{1}^T - \left(\hat{\mathbf{w}}^{ow}\right)^T\right)\mathbf{D}\left(\frac{1}{N}\mathbf{1} - \hat{\mathbf{w}}^{ow}\right)}\right\}\right\}.
\end{equation*}
\end{theorem}

\begin{proof}
Substituting the adaptive shrinkage weight specification $\hat{\mathbf{w}}^\lambda = \frac{\lambda}{N}\mathbf{1} + (1-\lambda)\hat{\mathbf{w}}^{ow}$ into equation (\ref{eq:matrix_msfe}), we obtain the objective function:
\begin{align*}
f(\lambda) &= \left(\frac{\lambda}{N}\mathbf{1} + (1-\lambda)\hat{\mathbf{w}}^{ow}\right)^T\mathbf{S} \nonumber \\
&\quad - \frac{1}{2}\left(\frac{\lambda}{N}\mathbf{1} + (1-\lambda)\hat{\mathbf{w}}^{ow}\right)^T\mathbf{D}\left(\frac{\lambda}{N}\mathbf{1} + (1-\lambda)\hat{\mathbf{w}}^{ow}\right).
\end{align*}

\noindent Expanding the quadratic term:
\begin{align*}
&\left(\frac{\lambda}{N}\mathbf{1} + (1-\lambda)\hat{\mathbf{w}}^{ow}\right)^T\mathbf{D}\left(\frac{\lambda}{N}\mathbf{1} + (1-\lambda)\hat{\mathbf{w}}^{ow}\right) \nonumber \\
&= \frac{\lambda^2}{N^2}\mathbf{1}^T\mathbf{D}\mathbf{1} + 2\frac{\lambda(1-\lambda)}{N}\mathbf{1}^T\mathbf{D}\hat{\mathbf{w}}^{ow} + (1-\lambda)^2\left(\hat{\mathbf{w}}^{ow}\right)^T\mathbf{D}\hat{\mathbf{w}}^{ow}.
\end{align*}

\noindent Taking the first-order derivative with respect to $\lambda$:
\begin{align*}
\frac{\partial f(\lambda)}{\partial \lambda} &= \left(\frac{1}{N}\mathbf{1}^T - \left(\hat{\mathbf{w}}^{ow}\right)^T\right)\mathbf{S} \nonumber \\
&\quad - \left[\frac{\lambda}{N^2}\mathbf{1}^T\mathbf{D}\mathbf{1} + \frac{(1-2\lambda)}{N}\mathbf{1}^T\mathbf{D}\hat{\mathbf{w}}^{ow} - (1-\lambda)\left(\hat{\mathbf{w}}^{ow}\right)^T\mathbf{D}\hat{\mathbf{w}}^{ow}\right] \nonumber \\
&= \left(\frac{1}{N}\mathbf{1}^T - \left(\hat{\mathbf{w}}^{ow}\right)^T\right)\mathbf{S} \nonumber \\
&\quad - \lambda\left(\frac{1}{N}\mathbf{1}^T - \left(\hat{\mathbf{w}}^{ow}\right)^T\right)\mathbf{D}\left(\frac{1}{N}\mathbf{1} - \hat{\mathbf{w}}^{ow}\right)- \left(\frac{1}{N}\mathbf{1}^T -\left(\hat{\mathbf{w}}^{ow}\right)^T\right)\mathbf{D}\hat{\mathbf{w}}^{ow}.
\end{align*}

\noindent Setting the derivative equal to zero:
\begin{align*}
&\left(\frac{1}{N}\mathbf{1}^T - \left(\hat{\mathbf{w}}^{ow}\right)^T\right)\mathbf{S} - \left(\frac{1}{N}\mathbf{1}^T - \left(\hat{\mathbf{w}}^{ow}\right)^T\right)\mathbf{D}\hat{\mathbf{w}}^{ow} \nonumber \\
&\quad = \lambda\left(\frac{1}{N}\mathbf{1}^T - \left(\hat{\mathbf{w}}^{ow}\right)^T\right)\mathbf{D}\left(\frac{1}{N}\mathbf{1} - \hat{\mathbf{w}}^{ow}\right).
\end{align*}

\noindent Solving for $\lambda$ yields the unconstrained optimum:
\begin{equation*}
\lambda^* = \frac{\left(\frac{1}{N}\mathbf{1}^T - \left(\hat{\mathbf{w}}^{ow}\right)^T\right)\left(\mathbf{S} - \mathbf{D}\hat{\mathbf{w}}^{ow}\right)}{\left(\frac{1}{N}\mathbf{1}^T - \left(\hat{\mathbf{w}}^{ow}\right)^T\right)\mathbf{D}\left(\frac{1}{N}\mathbf{1} - \hat{\mathbf{w}}^{ow}\right)}.
\end{equation*}

The constraint $\lambda \in [0,1]$ ensures economic interpretability and numerical stability.
\end{proof}

\subsection{Theoretical Analysis of Adaptive Shrinkage Weights}\label{sec:appendix A.2}

\renewcommand{\theequation}{\thesubsection.\arabic{equation}}
\setcounter{equation}{0}

We establish the theoretical properties of adaptive shrinkage weights through a two-stage optimization framework that illuminates the conditions under which shrinkage preserves optimality.

\subsubsection{Framework Setup}
\resettheoremcounter

Consider $N$ forecasts with covariance matrices $\mathbf{\Sigma}_1$ and $\mathbf{\Sigma}_2$ representing consecutive periods, both symmetric and positive definite. Define the constraint set:
\begin{equation*}
\mathcal{S} = \left\{\mathbf{w} \in \mathbb{R}^N : \mathbf{w}^T \mathbf{1} = 1\right\}.
\end{equation*}

The first-stage problem determines optimal weights under the initial covariance structure $\mathbf{\Sigma}_1$:
\begin{align}\label{eq:first_stage}
\min_{\mathbf{w} \in \mathcal{S}} \mathbf{w}^T \mathbf{\Sigma}_1 \mathbf{w},
\end{align}
with solution:
\begin{equation*}
\hat{\mathbf{w}}^{ow} = \frac{\mathbf{\Sigma}_1^{-1}\mathbf{1}}{\mathbf{1}^{T} \mathbf{\Sigma}_1^{-1} \mathbf{1}}.
\end{equation*}

The adaptive shrinkage weights are defined as:
\begin{equation*}
\hat{\mathbf{w}}^{\lambda} = \frac{\lambda}{N}\mathbf{1} + (1-\lambda)\hat{\mathbf{w}}^{ow},
\end{equation*}
where $\frac{1}{N}\mathbf{1}$ represents equal weights $\mathbf{w}^{ew} = \left(\frac{1}{N}, \frac{1}{N}, \ldots, \frac{1}{N}\right)^T$.

The second-stage problem evaluates adaptive shrinkage weights under the updated covariance structure $\mathbf{\Sigma}_2$:
\begin{equation}\label{eq:second_stage}
\min_{\lambda \in [0,1]} \left(\hat{\mathbf{w}}^{\lambda}\right)^T \mathbf{\Sigma}_2 \hat{\mathbf{w}}^{\lambda}.
\end{equation}

\subsubsection{Main Results} We present a detailed explanation of the three propositions mentioned in Section \ref{sec:shrinkage method}.
\resettheoremcounter

\noindent \textbf{Proposition~\ref{proposition 1}} [A Special Case] \label{prop:ow=ew} $\hat{\mathbf{w}}^{\lambda}=\hat{\mathbf{w}}^{ow} = \mathbf{w}^{ew}$ if and only if $\mathbf{1}$ is an eigenvector of $\mathbf{\Sigma}_1$.

\begin{remark}
    If neither $\mathbf{\Sigma}_1$ nor $\mathbf{\Sigma}_2$ has $\mathbf{1}$ as an eigenvector, then the equal weights $\mathbf{w}^{ew}$ are not guaranteed to be optimal. However, as long as either $\mathbf{\Sigma}_1$ or $\mathbf{\Sigma}_2$ has $\mathbf{1}$ as an eigenvector, the equal weights $\mathbf{w}^{ew}$ remain optimal throughout.
\end{remark}

\begin{proof}
\textbf{Observation:}
If a vector $\mathbf{L}$ is an eigenvector of a matrix $\mathbf{A}$, it is also an eigenvector of the inverse matrix $\mathbf{A^{-1}}$, as $\mathbf{A}\cdot \mathbf{L} = \delta\cdot \mathbf{L}$, and $\mathbf{A}^{-1}\cdot \mathbf{L} = \frac{1}{\delta}\cdot \mathbf{L}$, where $\delta$ is the eigenvalue corresponding to the eigenvector $\mathbf{L}$ of matrix $\mathbf{A}$.

If $\mathbf{1}$ is an eigenvector of $\mathbf{\Sigma}_1$, it suggests $\mathbf{\Sigma}_1^{-1}\cdot \mathbf{1} = \frac{1}{\delta_1}\cdot \mathbf{1}$, where $\delta_1$ is the eigenvalue corresponding to the eigenvector $\mathbf{1}$ of matrix $\mathbf{\Sigma}_1$. Denote $m=\mathbf{1}^{T} \mathbf{\Sigma}_1^{-1}\mathbf{1}$, we have:
\begin{align*}
    m = \mathbf{1}^{T} \mathbf{\Sigma}_1^{-1}\mathbf{1} = \mathbf{1}^T \frac{1}{\delta_1}\mathbf{1} = \frac{1}{\delta_1}\mathbf{1}^T\mathbf{1} = \frac{N}{\delta_1}.
\end{align*}

\noindent Therefore,
\begin{align*}
    \hat{\mathbf{w}}^{ow} = \frac{\mathbf{\Sigma}_1^{-1}\mathbf{1}}{\mathbf{1}^{T} \mathbf{\Sigma}_1^{-1} \mathbf{1}}=\frac{1}{\delta_1 m}\mathbf{1}=\frac{1}{\delta_1}\cdot\frac{\delta_1}{N}\mathbf{1}=\frac{1}{N}\mathbf{1}=\mathbf{w}^{ew}.
\end{align*}

In this case, the adaptive shrinkage weights reduce to $\hat{\mathbf{w}}^{\lambda}=\lambda \mathbf{w}^{ew} + (1-\lambda)\hat{\mathbf{w}}^{ow}=\hat{\mathbf{w}}^{ow}=\mathbf{w}^{ew}$ for any $\lambda \in [0,1]$. This indicates that when the optimal weights equal the equal weights, 2S--ASWs will invariably yield the same result regardless of the shrinkage parameter choice.

Similarly, if $\mathbf{1}$ is an eigenvector of $\mathbf{\Sigma_2}$, it suggests $\mathbf{\Sigma_2^{-1}}\mathbf{1} = \frac{1}{\delta_2}\mathbf{1}$, where $\delta_2$ is the eigenvalue corresponding to the eigenvector $\mathbf{1}$ of matrix $\mathbf{\Sigma}_2$. Therefore, the solution of the following optimization problem,
\begin{equation}\label{eq:general two-stage problem}
\min_{\mathbf{w} \in \mathcal{S}} \mathbf{w}^T \mathbf{\Sigma}_2 \mathbf{w},
\end{equation}

\noindent will be $\mathbf{w}^{ew}$. Note that, the feasible set of second-stage problem~(\ref{eq:second_stage}) is a subset of that of problem~(\ref{eq:general two-stage problem}), and $\mathbf{w}^{ew}$ (i.e., $\lambda=1$) is the unique global minimizer of problem~(\ref{eq:general two-stage problem}). As a result, it follows that $\mathbf{w}^{ew}$ ($\lambda=1$) is also the optimal solution to the constrained second-stage problem~(\ref{eq:second_stage}).

In contrast, if both $\mathbf{\Sigma}_1$ and $\mathbf{\Sigma}_2$ do not have $\mathbf{1}$ as an eigenvector, then $\hat{\mathbf{w}}^{ow} \neq \mathbf{w}^{ew}$ and $\mathbf{w}^{ew}$ is not the unconstrained minimum-variance solution under $\mathbf{\Sigma}_2$. Therefore, equal weights are not guaranteed to be optimal in the second-stage optimization problem, although they can still be obtained mechanically by setting $\lambda=1$.

Conversely, suppose that \(\hat{\mathbf w}^{ow}=\mathbf w^{ew}\). Then
\begin{equation*}
\frac{\mathbf\Sigma_1^{-1}\mathbf 1}
{\mathbf 1^T\mathbf\Sigma_1^{-1}\mathbf 1}
=\frac{1}{N}\mathbf 1.
\end{equation*}
Let $m=\mathbf 1^T\mathbf\Sigma_1^{-1}\mathbf 1$. Then
\begin{equation*}
\mathbf\Sigma_1^{-1}\mathbf 1=\frac{m}{N}\mathbf 1,
\end{equation*}
so $\mathbf 1$ is an eigenvector of $\mathbf\Sigma_1^{-1}$, and therefore
also of $\mathbf\Sigma_1$.

\end{proof}

In the subsequent analysis, we focus on the general scenario where $\hat{\mathbf{w}}^{ow} \neq \mathbf{w}^{ew}$.

\noindent \textbf{Proposition~\ref{proposition 3}} [Stationarity Result]\label{prop:stationarity}
Under covariance stationarity ($\mathbf{\Sigma}_1 = \mathbf{\Sigma}_2 = \mathbf{\Sigma}$), the second-stage optimum is attained at $\lambda_{opt}=0$, and therefore adaptive shrinkage weights coincide with optimal weights: $\hat{\mathbf{w}}^{\lambda} = \hat{\mathbf{w}}^{ow}$.

\begin{proof}
Under stationarity, the second-stage objective function is:
\begin{align*}
f(\lambda) = \left[\lambda \mathbf{w}^{ew} + (1-\lambda) \hat{\mathbf{w}}^{ow}\right]^T \mathbf{\Sigma} \left[\lambda \mathbf{w}^{ew} + (1-\lambda) \hat{\mathbf{w}}^{ow}\right].
\end{align*}

The first-order derivative with respect to $\lambda$ is:
\begin{align*}
\frac{\partial f}{\partial \lambda} &= 2 \left[\lambda \mathbf{w}^{ew} + (1-\lambda) \hat{\mathbf{w}}^{ow}\right]^T \mathbf{\Sigma} \left(\mathbf{w}^{ew} - \hat{\mathbf{w}}^{ow}\right) \nonumber \\
&= 2\lambda \left(\mathbf{w}^{ew} - \hat{\mathbf{w}}^{ow}\right)^T \mathbf{\Sigma} \left(\mathbf{w}^{ew} - \hat{\mathbf{w}}^{ow}\right) \nonumber \\
&\quad + 2 \left(\hat{\mathbf{w}}^{ow}\right)^T \mathbf{\Sigma} \left(\mathbf{w}^{ew} - \hat{\mathbf{w}}^{ow}\right).
\end{align*}

\noindent Let $\frac{\partial f}{\partial \lambda}=0$, we obtain that
\begin{equation}
    \lambda = -\frac{\left(\hat{\mathbf{w}}^{ow}\right)^T \mathbf{\Sigma} \left(\mathbf{w}^{ew} - \hat{\mathbf{w}}^{ow}\right)}{\left(\mathbf{w}^{ew} - \hat{\mathbf{w}}^{ow}\right)^T \mathbf{\Sigma} \left(\mathbf{w}^{ew} - \hat{\mathbf{w}}^{ow}\right)}.
\label{eq:lambda appendix}
\end{equation}

\noindent Since $\hat{\mathbf{w}}^{ow} \neq \mathbf{w}^{ew}$, the denominator of equation~(\ref{eq:lambda appendix}) will not be zero. We compute the numerator as:

\begin{align*}
\left(\hat{\mathbf{w}}^{ow}\right)^T \mathbf{\Sigma} \left(\mathbf{w}^{ew} - \hat{\mathbf{w}}^{ow}\right) \nonumber
&= \frac{\mathbf{1}^T \mathbf{\Sigma}^{-1}}{\mathbf{1}^T \mathbf{\Sigma}^{-1} \mathbf{1}} \mathbf{\Sigma} \left(\frac{\mathbf{1}}{N} - \frac{\mathbf{\Sigma}^{-1}\mathbf{1}}{\mathbf{1}^T \mathbf{\Sigma}^{-1} \mathbf{1}}\right) \nonumber \\
&= \frac{1}{\mathbf{1}^T \mathbf{\Sigma}^{-1} \mathbf{1}}\cdot \frac{\mathbf{1}^T \mathbf{1}}{N} -  \frac{\mathbf{1}^T \mathbf{\Sigma}^{-1} \mathbf{1}}{\left(\mathbf{1}^T \mathbf{\Sigma}^{-1} \mathbf{1}\right)^2} \nonumber \\
&= \frac{1}{\mathbf{1}^T \mathbf{\Sigma}^{-1} \mathbf{1}} - \frac{1}{\mathbf{1}^T \mathbf{\Sigma}^{-1} \mathbf{1}} = 0.
\end{align*}

\noindent Therefore, $\lambda=0$ in equation~(\ref{eq:lambda appendix}). Besides, due to the fact that the second derivative $\frac{d^2 f}{d\lambda^2}=2\left(\mathbf{w}^{ew} - \hat{\mathbf{w}}^{ow}\right)^T \mathbf{\Sigma} \left(\mathbf{w}^{ew} - \hat{\mathbf{w}}^{ow}\right) > 0$, $\lambda = 0$ is indeed the global minimum. This result suggests that adaptive shrinkage weights $\hat{\mathbf{w}}^{\lambda}$ of the second-stage problem are equal to the optimal weights $\hat{\mathbf{w}}^{ow}$ in the first stage.
\end{proof}

\noindent \textbf{Proposition~\ref{proposition 4}} [Non-Stationarity Condition]\label{prop:nonstationarity}
Under covariance non-stationarity ($\mathbf{\Sigma}_1 \neq \mathbf{\Sigma}_2$), the unconstrained optimal shrinkage parameter satisfies $\lambda^*=0$ if and only if
\begin{equation}\label{eq:nonstationarity_condition}
\mathbf{1}^T\mathbf{\Sigma}_1^{-1}\mathbf{\Sigma}_2
\left(\mathbf{I}-\frac{N}{m}\mathbf{\Sigma}_1^{-1}\right)\mathbf{1}=0,
\end{equation}
where $m=\mathbf{1}^T\mathbf{\Sigma}_1^{-1}\mathbf{1}$.

\begin{proof}
Still, we focus on the numerator $\left(\hat{\mathbf{w}}^{ow}\right)^T \mathbf{\Sigma}_2 \left(\mathbf{w}^{ew} - \hat{\mathbf{w}}^{ow}\right)$ of equation~(\ref{eq:lambda appendix}). Substituting the expressions for $\hat{\mathbf{w}}^{ow}$ and $\mathbf{w}^{ew}$:
\begin{align*}
&\left(\hat{\mathbf{w}}^{ow}\right)^T \mathbf{\Sigma}_2 \left(\mathbf{w}^{ew} - \hat{\mathbf{w}}^{ow}\right) \nonumber \\
&= \frac{\mathbf{1}^T\mathbf{\Sigma}_1^{-1}}{\mathbf{1}^T\mathbf{\Sigma}_1^{-1}\mathbf{1}}\mathbf{\Sigma}_2\left(\frac{\mathbf{1}}{N}-\frac{\mathbf{\Sigma}_1^{-1}\mathbf{1}}{\mathbf{1}^T\mathbf{\Sigma}_1^{-1}\mathbf{1}}\right) \nonumber \\
&= \frac{1}{N}\frac{\mathbf{1}^T\mathbf{\Sigma}_1^{-1}\mathbf{\Sigma}_2\mathbf{1}}{\mathbf{1}^T\mathbf{\Sigma}_1^{-1}\mathbf{1}}-\frac{\mathbf{1}^T\mathbf{\Sigma}_1^{-1}\mathbf{\Sigma}_2\mathbf{\Sigma}_1^{-1}\mathbf{1}}{\left(\mathbf{1}^T\mathbf{\Sigma}_1^{-1}\mathbf{1}\right)^2}.
\end{align*}

\noindent Denote $m = \mathbf{1}^T\mathbf{\Sigma}_1^{-1}\mathbf{1}$. Setting the expression equal to zero:
\begin{align*}
&\frac{\mathbf{1}^T\mathbf{\Sigma}_1^{-1}\mathbf{\Sigma}_2\mathbf{1}}{Nm} - \frac{\mathbf{1}^T\mathbf{\Sigma}_1^{-1}\mathbf{\Sigma}_2\mathbf{\Sigma}_1^{-1}\mathbf{1}}{m^2}=0 \nonumber \\
&\Leftrightarrow
\frac{\mathbf{1}^T\mathbf{\Sigma}_1^{-1}\mathbf{\Sigma}_2\mathbf{1}}{N} - \frac{\mathbf{1}^T\mathbf{\Sigma}_1^{-1}\mathbf{\Sigma}_2\mathbf{\Sigma}_1^{-1}\mathbf{1}}{m}=0 \nonumber \\
&\Leftrightarrow
\mathbf{1}^T\mathbf{\Sigma}_1^{-1}\mathbf{\Sigma}_2\mathbf{1}-\frac{N}{m}\mathbf{1}^T\mathbf{\Sigma}_1^{-1}\mathbf{\Sigma}_2\mathbf{\Sigma}_1^{-1}\mathbf{1} = 0 \nonumber \\
&\Leftrightarrow
\mathbf{1}^T\mathbf{\Sigma}_1^{-1}\mathbf{\Sigma}_2\left(\mathbf{I}-\frac{N}{m}\mathbf{\Sigma}_1^{-1}\right)\mathbf{1}=0.
\end{align*}
\end{proof}

\noindent \textbf{Illustrative Examples:}
We provide concrete examples demonstrating scenarios where the equation~(\ref{eq:nonstationarity_condition}) in Proposition \ref{proposition 4} is satisfied despite non-stationarity.

\begin{example}[With
an Eigenvector $\mathbf{1}$]\label{example:strategical matrix} Suppose $\mathbf{\Sigma}_1^{-1}\mathbf{1}=\frac{1}{\delta_1}\mathbf{1}$, where $\delta_1$ is the eigenvalue corresponding to the eigenvector $\mathbf{1}$ of $\mathbf{\Sigma}_1$. Therefore, the equation~(\ref{eq:nonstationarity_condition})
\begin{align*} \mathbf{1}^T\mathbf{\Sigma}_1^{-1}\mathbf{\Sigma}_2\left(\mathbf{I}-\frac{N}{m}\mathbf{\Sigma}_1^{-1}\right)\mathbf{1} &= \frac{1}{\delta_1} \mathbf{1}^T\mathbf{\Sigma}_2\mathbf{1}-(\frac{1}{\delta_1})^2\mathbf{1}^T\mathbf{\Sigma}_2\frac{N}{m}\mathbf{1} \\
&= \frac{1}{\delta_1}\left(1-\frac{1}{\delta_1}\frac{N}{m}\right)\mathbf{1}^T\mathbf{\Sigma}_2\mathbf{1} \\
&= \frac{1}{\delta_1}\left(1-\frac{m}{N}\cdot\frac{N}{m}\right)\mathbf{1}^T\mathbf{\Sigma}_2\mathbf{1} \\
&= 0 
\end{align*} naturally holds, regardless of $\Sigma_2$.

This identity implies that, irrespective of the value taken by $\mathbf{\Sigma}_2$, as long as $\mathbf{1}$ is an eigenvector of the matrix $\mathbf{\Sigma}_1$ or $\mathbf{\Sigma}_1^{-1}$, $\lambda^{*}=0$ and $\hat{\mathbf{w}}^{ow}=\hat{\mathbf{w}}^{\lambda}=\mathbf{w}^{ew}$ according to Proposition~\ref{proposition 1}, i.e., the optimal weights, which remain optimal in the second stage.
\end{example}

\begin{example}[Proportional Covariance Matrices]\label{example:proportional}
Consider the case where $\mathbf{\Sigma}_1 = k\mathbf{\Sigma}_2$ for some scalar $k > 0$. This proportional relationship implies $\mathbf{\Sigma}_1^{-1} = \frac{1}{k}\mathbf{\Sigma_2^{-1}}$, and $\mathbf{1}^T\mathbf{\Sigma}_1^{-1}\mathbf{1} = \frac{1}{k}\mathbf{1}^T\mathbf{\Sigma_2^{-1}}\mathbf{1}$.

Substituting into the condition in equation (\ref{eq:nonstationarity_condition}):
\begin{align*}
&\mathbf{1}^T\mathbf{\Sigma}_1^{-1}\mathbf{\Sigma}_2\left(\mathbf{I}-\frac{N}{m}\mathbf{\Sigma}_1^{-1}\right)\mathbf{1} \nonumber \\
&= \mathbf{1}^T\frac{1}{k}\mathbf{\Sigma_2^{-1}}\mathbf{\Sigma}_2\left(\mathbf{I}-\frac{N}{m}\frac{1}{k}\mathbf{\Sigma_2^{-1}}\right)\mathbf{1} \nonumber \\
&= \frac{1}{k}\mathbf{1}^T\left(\mathbf{I}-\frac{N}{mk}\mathbf{\Sigma_2^{-1}}\right)\mathbf{1} \nonumber \\
&= \frac{1}{k}\left(\mathbf{1}^T\mathbf{1} - \frac{N}{mk}\mathbf{1}^T\mathbf{\Sigma_2^{-1}}\mathbf{1}\right) \\
&= \frac{1}{k}\left(N - \frac{N}{mk}\cdot mk\right) \nonumber \\
&= \frac{1}{k}\left(N - N\right) = 0.
\end{align*}

\noindent This equality holds trivially, confirming that $\lambda^{*}=0$ and $\hat{\mathbf{w}}^{ow} = \hat{\mathbf{w}}^{\lambda}$ when covariance matrices are proportional.
\end{example}

\begin{example}[Non-Proportional Matrices]\label{example:nonproportional} We present two kinds of detailed examples in which $\mathbf{\Sigma}_1 \neq \mathbf{\Sigma}_2$, and $\mathbf{\Sigma}_1$ and $\mathbf{\Sigma}_2$ are not proportionally related.

\textbf{{\textnormal{\bfseries A sufficient condition.}}}
Consider a specific scenario where $\mathbf{\Sigma}_1 \neq \mathbf{\Sigma}_2$ and the matrices are not proportionally related. We construct the following $5 \times 5$ precision matrices:

{\small
\allowdisplaybreaks
\addtolength{\arraycolsep}{-2pt}

\begin{align*}
\mathbf{\Sigma}_1^{-1} = \begin{bmatrix}
6.0 & 0.5 & 1.0 & 0.5 & 0.0 \\
0.5 & 5.0 & 1.0 & 0.0 & 1.5 \\
1.0 & 1.0 & 4.0 & 1.0 & 1.0 \\
0.5 & 0.0 & 1.0 & 6.0 & 0.5 \\
0.0 & 1.5 & 1.0 & 0.5 & 5.0 \\
\end{bmatrix}, \quad \mathbf{\Sigma_2^{-1}} = \begin{bmatrix}
2.0 & 1.0 & 0.0 & 0.5 & 0.0 \\
1.0 & 3.0 & 0.5 & 0.0 & 1.0 \\
0.0 & 0.5 & 4.0 & 1.0 & 0.0 \\
0.5 & 0.0 & 1.0 & 3.0 & 0.5 \\
0.0 & 1.0 & 0.0 & 0.5 & 2.0 \\
\end{bmatrix}.
\end{align*}

Direct computation yields:
\begin{align*}
N &= 5,\\
m &= \mathbf{1}^{T}\mathbf{\Sigma}_1^{-1}\mathbf{1} = \sum_{i=1}^{N} \sum_{j=1}^{N}\left(\mathbf{\Sigma}_1^{-1}\right)_{ij} = 40,
\\
\left ( \mathbf{I}-\frac{N}{m}\mathbf{\Sigma}_1^{-1} \right )\mathbf{1} &= \left(\begin{bmatrix}
    1.0 & 0.0 & 0.0 & 0.0 & 0.0 \\
    0.0 & 1.0 & 0.0 & 0.0 & 0.0 \\
    0.0 & 0.0 & 1.0 & 0.0 & 0.0 \\
    0.0 & 0.0 & 0.0 & 1.0 & 0.0 \\
    0.0 & 0.0 & 0.0 & 0.0 & 1.0
\end{bmatrix} -\frac{5}{40}\begin{bmatrix}
    6.0 & 0.5 & 1.0 & 0.5 & 0.0 \\
    0.5 & 5.0 & 1.0 & 0.0 & 1.5 \\
    1.0 & 1.0 & 4.0 & 1.0 & 1.0 \\
    0.5 & 0.0 & 1.0 & 6.0 & 0.5 \\
    0.0 & 1.5 & 1.0 & 0.5 & 5.0
\end{bmatrix}\right)\mathbf{1} \\
&= \begin{bmatrix}
    \frac{1}{4} & -\frac{1}{16} & -\frac{1}{8} & -\frac{1}{16} & 0 \\
    -\frac{1}{16} & \frac{3}{8} & -\frac{1}{8} & 0 & -\frac{3}{16} \\
    -\frac{1}{8} & -\frac{1}{8} & \frac{1}{2} & -\frac{1}{8} & -\frac{1}{8} \\
    -\frac{1}{16} & 0 & -\frac{1}{8} & \frac{1}{4} & -\frac{1}{16} \\
    0 & -\frac{3}{16} & -\frac{1}{8} & -\frac{1}{16} & \frac{3}{8} \\
\end{bmatrix}\mathbf{1}\\
&=\begin{bmatrix}
0.0 \\ 0.0 \\ 0.0 \\ 0.0 \\ 0.0
\end{bmatrix}.
\end{align*}
}

\noindent Therefore, the condition in equation(\ref{eq:nonstationarity_condition}) $\mathbf{1}^T\mathbf{\Sigma}_1^{-1}\mathbf{\Sigma}_2\left(\mathbf{I}-\frac{N}{m}\mathbf{\Sigma}_1^{-1}\right)\mathbf{1}=0$ definitely holds.

For the condition to hold, a sufficient condition is easily observed that $\left(\mathbf{I}-\frac{N}{m}\mathbf{\Sigma}_1^{-1}\right)\mathbf{1}=\mathbf{0}$. This suggests that we can focus on the construction of $\mathbf{\Sigma}_1$ to satisfy the sufficient condition, whereas $\mathbf{\Sigma}_2$ can be any symmetric positive definite matrix.

\textbf{{\textnormal{\bfseries A case where the sufficient condition is not necessary.}}}
We present the following $5\times5$ matrices that do not satisfy the sufficient condition $\left ( \mathbf{I}-\frac{N}{m}\mathbf{\Sigma}_1^{-1} \right )\mathbf{1}=\mathbf{0}$, but still satisfy the condition in equation~(\ref{eq:nonstationarity_condition}).

{\small
\allowdisplaybreaks
\addtolength{\arraycolsep}{-2pt}

\begin{align*}
\mathbf{\Sigma}_1^{-1} = \begin{bmatrix}
4.0 & 2.0 & 1.0 & 3.0 & 2.0 \\
2.0 & 5.0 & 3.0 & 1.0 & 2.0 \\
1.0 & 3.0 & 6.0 & 2.0 & 1.0 \\
3.0 & 1.0 & 2.0 & 5.0 & 3.0 \\
2.0 & 2.0 & 1.0 & 3.0 & 4.0 \\
\end{bmatrix},  \quad \mathbf{\Sigma}_2 =
\begin{bmatrix}
3.0 & -2.0 & 1.0 & -1.0 & 1.0 \\
-2.0 & \frac{69}{13} & -3.0 & 2.0 & -1.0 \\
1.0 & -3.0 & 5.0 & -2.0 & 1.0 \\
-1.0 & 2.0 & -2.0 & 4.0 & -2.0 \\
1.0 & -1.0 & 1.0 & -2.0 & 3.0 \\
\end{bmatrix}.
\end{align*}
\noindent We compute $\mathbf{1}^T\mathbf{\Sigma}_1^{-1}\mathbf{\Sigma}_2\left(\mathbf{I}-\frac{N}{m}\mathbf{\Sigma}_1^{-1}\right)\mathbf{1}$ directly.
\begin{align*}
N &= 5,\\
m &= \mathbf{1}^{T}\mathbf{\Sigma}_1^{-1}\mathbf{1} = \sum_{i=1}^{N} \sum_{j=1}^{N}\left(\mathbf{\Sigma}_1^{-1}\right)_{ij} = 64, \\
\mathbf{v} = \left ( \mathbf{I}-\frac{N}{m}\mathbf{\Sigma}_1^{-1} \right )\mathbf{1} &= \mathbf{1}-\frac{N}{m}\mathbf{\Sigma}_1^{-1}\mathbf{1} =\begin{bmatrix}
1 \\ 1 \\ 1 \\ 1 \\ 1
\end{bmatrix}-\frac{5}{64}\begin{bmatrix}
    12.0 \\
    13.0 \\
    13.0 \\
    14.0 \\
    12.0 \\
\end{bmatrix} = \frac{1}{64}\begin{bmatrix}
4 \\ -1 \\ -1 \\ -6 \\ 4
\end{bmatrix} \\
\mathbf{1}^T\mathbf{\Sigma}_1^{-1}\mathbf{\Sigma}_2\mathbf{v} &= \frac{1}{64}\begin{bmatrix}
1 & 1 & 1 & 1 & 1
\end{bmatrix}\begin{bmatrix}
4.0 & 2.0 & 1.0 & 3.0 & 2.0 \\
2.0 & 5.0 & 3.0 & 1.0 & 2.0 \\
1.0 & 3.0 & 6.0 & 2.0 & 1.0 \\
3.0 & 1.0 & 2.0 & 5.0 & 3.0 \\
2.0 & 2.0 & 1.0 & 3.0 & 4.0
\end{bmatrix} \\
&\quad \times \begin{bmatrix}
3.0 & -2.0 & 1.0 & -1.0 & 1.0 \\
-2.0 & \frac{69}{13} & -3.0 & 2.0 & -1.0 \\
1.0 & -3.0 & 5.0 & -2.0 & 1.0 \\
-1.0 & 2.0 & -2.0 & 4.0 & -2.0 \\
1.0 & -1.0 & 1.0 & -2.0 & 3.0 \\
\end{bmatrix}
\begin{bmatrix}
4 \\ -1 \\ -1 \\ -6 \\ 4
\end{bmatrix} \\
&= \frac{1}{64}\begin{bmatrix}
12.0 & 13.0 & 13.0 & 14.0 & 12.0
\end{bmatrix} \\
&\quad \times \begin{bmatrix}
3.0 & -2.0 & 1.0 & -1.0 & 1.0 \\
-2.0 & \frac{69}{13} & -3.0 & 2.0 & -1.0 \\
1.0 & -3.0 & 5.0 & -2.0 & 1.0 \\
-1.0 & 2.0 & -2.0 & 4.0 & -2.0 \\
1.0 & -1.0 & 1.0 & -2.0 & 3.0 \\
\end{bmatrix}\begin{bmatrix}
4 \\ -1 \\ -1 \\ -6 \\ 4
\end{bmatrix} = 0.
\end{align*}
}

\noindent Given that such a pair of solutions ($\mathbf{\Sigma}_1^{-1},\mathbf{\Sigma}_2$) exists, any scalar $k$ multiples $\mathbf{\Sigma}_2$ will also satisfy the specified condition (\ref{eq:nonstationarity_condition}), thereby generating an infinite family of solutions.
\end{example}

In the preceding analysis, we examine the problem in the context of the adaptive shrinkage combination framework. We now turn to a general scenario that is independent of adaptive shrinkage.

\begin{remark}\label{general remark}
Under covariance non-stationarity ($\mathbf{\Sigma}_1 \neq \mathbf{\Sigma}_2$), the unconstrained optimal weights under $\mathbf{\Sigma}_1$ equal those under $\mathbf{\Sigma}_2$ if and only if there exists a scalar $\eta>0$ such that:
\begin{equation}\label{eq:condition2}
(\mathbf{\Sigma}_1^{-1}-\eta\mathbf{\Sigma}_2^{-1})\mathbf{1}=\mathbf{0}.
\end{equation}
\end{remark}

\begin{proof}
The minimum-variance portfolio under $\mathbf{\Sigma}_2$ is

\begin{align*}
\mathbf{w_2^{*}}=\frac{\mathbf{\Sigma_2^{-1}}\mathbf{1}}{\mathbf{1}^T\mathbf{\Sigma_2^{-1}}\mathbf{1}}.
\end{align*}

Let $\mathbf{w_1^{*}} = \mathbf{w_2^{*}}$, i.e.,
\begin{align}
    \frac{\mathbf{\Sigma}_1^{-1}\mathbf{1}}{\mathbf{1}^T\mathbf{\Sigma}_1^{-1}\mathbf{1}}=\frac{\mathbf{\Sigma_2^{-1}}\mathbf{1}}{\mathbf{1}^T\mathbf{\Sigma_2^{-1}}\mathbf{1}}.
\label{eq:w1=w2}
\end{align}
Denote $k_1=\mathbf{1}^T\mathbf{\Sigma}_1^{-1}\mathbf{1}$, $k_2=\mathbf{1}^T\mathbf{\Sigma_2^{-1}}\mathbf{1}$, and equation~(\ref{eq:w1=w2}) becomes
\begin{align*}
    \mathbf{\Sigma}_1^{-1}\mathbf{1} = \frac{k_1}{k_2}\mathbf{\Sigma_2^{-1}}\mathbf{1},
\end{align*}
i.e.,
\begin{align*}
    (\mathbf{\Sigma_1^{-1}-\eta\Sigma_2^{-1}})\mathbf{1} = \mathbf{0},\quad \text{where } \eta=\frac{k_1}{k_2}.
\end{align*}

Conversely, If $\mathbf{\Sigma_1}^{-1}\mathbf 1=\eta\mathbf{\Sigma_2}^{-1}\mathbf{1}$, then
\begin{equation*}
\frac{\mathbf{\Sigma}_1^{-1}\mathbf{1}}{\mathbf 1^T\mathbf{\Sigma}_1^{-1}\mathbf 1}=\frac{\eta\mathbf{\Sigma}_2^{-1}\mathbf 1}{\eta\mathbf 1^T\mathbf{\Sigma}_2^{-1}\mathbf 1}=\frac{\mathbf{\Sigma}_2^{-1}\mathbf 1}{\mathbf 1^T\mathbf{\Sigma}_2^{-1}\mathbf 1}.
\end{equation*}
Thus $\mathbf{w}_1^{*}=\mathbf{w}_2^{*}$.
\end{proof}

\noindent \textbf{Illustrative Examples:} We provide concrete examples demonstrating scenarios where the condition~(\ref{eq:condition2}) in Remark~\ref{general remark} is satisfied despite non-stationarity. 

\begin{example}[Proportional Covariance Matrices]\label{example:proportional_general}
Consider the case where $\mathbf{\Sigma}_1 = k\mathbf{\Sigma}_2$ for some scalar $k > 0$. This proportional relationship implies $\mathbf{\Sigma}_1^{-1} = \frac{1}{k}\mathbf{\Sigma_2^{-1}}$, and $\mathbf{1}^{T}\mathbf{\Sigma}_1^{-1}\mathbf{1}=\frac{1}{k}\mathbf{1}^{T}\mathbf{\Sigma_2^{-1}}\mathbf{1}$.

Substituting into the condition in equation (\ref{eq:condition2}):
\begin{align*}
 (\mathbf{\Sigma_1^{-1}-\eta\Sigma_2^{-1}})\mathbf{1} =  (\mathbf{\Sigma}_1^{-1}-\frac{1}{k}\mathbf{\Sigma_2^{-1}})\mathbf{1}=\mathbf{0}
\end{align*}

This equality holds trivially, confirming that $\mathbf{w}_1^{*} = mathbf{w}\_2^{*}$ when the covariance matrices are proportional.
\end{example}

\begin{example}[Non-Proportional Matrices]\label{example:nonproportional_general}
Consider a specific scenario where $\mathbf{\Sigma}_1 \neq \mathbf{\Sigma}_2$ and the matrices are not proportionally related. We construct the following $5 \times 5$ matrices:

{\small
\allowdisplaybreaks
\addtolength{\arraycolsep}{-2pt}

\begin{align*}
\mathbf{\Sigma}_1^{-1} = \begin{bmatrix}
1.5 & 0.5 & 0.0 & 0.0 & 0.0 \\
0.5 & 1.5 & 0.5 & 0.0 & 0.0 \\
0.0 & 0.5 & 1.5 & 0.5 & 0.0 \\
0.0 & 0.0 & 0.5 & 1.5 & 0.5 \\
0.0 & 0.0 & 0.0 & 0.5 & 1.5
\end{bmatrix}, \quad \mathbf{\Sigma_2^{-1}} = \begin{bmatrix}
1.4 & 0.3 & 0.1 & 0.1 & 0.1 \\
0.3 & 1.7 & 0.3 & 0.1 & 0.1 \\
0.1 & 0.3 & 1.7 & 0.3 & 0.1 \\
0.1 & 0.1 & 0.3 & 1.7 & 0.3 \\
0.1 & 0.1 & 0.1 & 0.3 & 1.4
\end{bmatrix}.
\end{align*}

Direct computation yields:
\begin{align*}
\mathbf{\Sigma}_1^{-1}\mathbf{1}
&=\begin{bmatrix}
1.5 & 0.5 & 0.0 & 0.0 & 0.0 \\
0.5 & 1.5 & 0.5 & 0.0 & 0.0 \\
0.0 & 0.5 & 1.5 & 0.5 & 0.0 \\
0.0 & 0.0 & 0.5 & 1.5 & 0.5 \\
0.0 & 0.0 & 0.0 & 0.5 & 1.5
\end{bmatrix}
\begin{bmatrix}
1.0 \\ 1.0 \\ 1.0 \\ 1.0 \\ 1.0\\
\end{bmatrix}=\begin{bmatrix}
2.0 \\ 2.5 \\ 2.5 \\ 2.5 \\ 2.0
\end{bmatrix} \\
\mathbf{\Sigma_2^{-1}}\mathbf{1}
&=\begin{bmatrix}
1.4 & 0.3 & 0.1 & 0.1 & 0.1 \\
0.3 & 1.7 & 0.3 & 0.1 & 0.1 \\
0.1 & 0.3 & 1.7 & 0.3 & 0.1 \\
0.1 & 0.1 & 0.3 & 1.7 & 0.3 \\
0.1 & 0.1 & 0.1 & 0.3 & 1.4
\end{bmatrix}
\begin{bmatrix}
1.0 \\ 1.0 \\ 1.0 \\ 1.0 \\ 1.0\\
\end{bmatrix}=\begin{bmatrix}
2.0 \\ 2.5 \\ 2.5 \\ 2.5 \\ 2.0
\end{bmatrix} \\
\mathbf{1}^{T}\mathbf{\Sigma}_1^{-1}\mathbf{1} &= \mathbf{1}^{T}\mathbf{\Sigma_2^{-1}}\mathbf{1} = 2.0 + 2.5 + 2.5 + 2.5 + 2.0 = 11.5
\end{align*}
}

Therefore, the condition of equation~(\ref{eq:condition2}) $(\mathbf{\Sigma}_1^{-1}-\eta\mathbf{\Sigma}_2^{-1})\mathbf{1}=\mathbf{0}, \text{where} \; \eta=1$ holds.

\end{example}

\begin{remark}[Economic Interpretation]
The condition in Proposition \ref{proposition 4} admits infinitely many solutions for pairs $(\mathbf{\Sigma}_1, \mathbf{\Sigma}_2)$, suggesting that optimal combination weights can remain robust to specific types of structural changes in the forecasting environment. This finding provides a theoretical foundation for understanding when forecast combination strategies exhibit stability across regime changes, which has important implications for the practical implementation of adaptive shrinkage methods in financial forecasting.
\end{remark}

	\clearpage

\section{Additional Methodology Details}\label{sec: appendix detail}

\renewcommand{\thesubsection}{B.\arabic{subsection}}
\setcounter{subsection}{0}

\renewcommand{\thetable}{\Alph{section}.\arabic{table}}

Table~\ref{tab:methods_glossary} summarizes all the (combination) forecasting methods used in this paper, and Table~\ref{tab:sample division} presents the comprehensive sample division framework and parameter specifications employed in the empirical analysis. The methodological approach requires careful consideration of sample partitioning to ensure robust out-of-sample evaluation while maintaining sufficient data for parameter estimation and model selection.

For individual regression analysis of each GW predictor, the full sample of $T$ months is partitioned into two non-overlapping periods following standard forecasting practice. The first $T_1$ months constitute the in-sample estimation period for fitting predictive regression models, while the remaining $T_2 = T-T_1$ months are reserved for genuine out-of-sample forecast evaluation. This division ensures that no future information is used in model estimation, thereby providing an unbiased assessment of forecasting performance. To assess the robustness of our adaptive shrinkage combination approach to different estimation window lengths, we consider two alternative specifications for the in-sample period: $T_1=54$ months (January 1996 to June 2000) and $T_1=42$ months (January 1996 to June 1999). Both configurations maintain identical out-of-sample periods extending from the end of their respective in-sample windows to August 2019, allowing for direct comparison of forecasting accuracy across different training sample sizes.

The combination methods employ a more sophisticated three-phase sample division to accommodate the two-step estimation process inherent in shrinkage techniques, including 2S--ASWs, Lasso, Ridge, and ENet. This approach recognizes that these methods require both weight estimation and hyperparameter tuning, necessitating separate data partitions to avoid overfitting. The individual forecast series generated from each GW predictor in the previous step are sequentially partitioned into training ($T_{c1}$), validation ($T_{c2}$), and test ($T_{c3}$) periods, where $T_{c1} + T_{c2} + T_{c3} = T$.

For shrinkage methods, the training period $T_{c1}$ is used to estimate optimal combination weights, while the validation period $T_{c2}$ serves a dual purpose: selecting the optimal shrinkage parameter $\lambda_{opt}$ through validation and subsequently deriving the final adaptive shrinkage weights using this optimal parameter. The test period $T_{c3}$ is reserved exclusively for out-of-sample evaluation, ensuring that the forecasting performance assessment remains uncontaminated by the parameter selection process. In the in-sample implementation, $T_{c1}$ equals $T_1$ from the individual analysis to maintain consistency, while $T_{c2}$ uniformly terminates in December 2006 across all specifications. This design creates validation sets spanning July 2000 to December 2006 ($T_{c2} = 78$ months) for $T_1=54$ and July 1999 to December 2006 ($T_{c2} = 90$ months) for $T_1=42$, with the remaining periods allocated to out-of-sample testing.

We distinguish between two methodological categories based on their computational requirements for out-of-sample evaluation. Shrinkage methods (2S--ASWs, Lasso, Ridge, and ENet) require separate validation sets for hyperparameter determination and thus follow the three-phase division described above. In contrast, one-step methods (OWs, \cite{blanc2020bias}, EWs, Median, Trimmed Mean, DMSPE, ASWs--\cite{li2025dominate}, and PCA) do not require hyperparameter tuning and therefore utilize combined training and validation samples spanning January 1996 to December 2006, with the remainder allocated to out-of-sample evaluation. This approach maximizes the available data for these simpler methods while maintaining computational efficiency. 

\cite{blanc2020bias} proposes a shrinkage combination that shares the same functional form with ours, i.e., $\hat{\mathbf{w}}=\lambda \mathbf{w}^{ew}+(1-\lambda)\hat{\mathbf{w}^{ow}}$, but the determination of $\lambda$ differs: OWs, $\lambda_{opt}$, and shrinkage weights in \cite{blanc2020bias} are all derived directly from the same training samples. Based on the bias-variance trade-off, the optimal shrinkage parameter $\lambda_{opt}$ is calculated as: $\lambda_{opt}=\frac{\mathbf{V}}{\mathbf{B}+\mathbf{V}}$, where $\mathbf{B}=\frac{\mathbf{1}^T\Sigma 1}{N^2} - \frac{1}{\mathbf{1}^T\Sigma^{-1} 1}$, and $\mathbf{V}=\frac{N-1}{H-N}\frac{1}{\mathbf{1}^T\Sigma^{-1} 1}$. \cite{li2025dominate} proposes a conservative constraint on the predictive slope coefficient of individual predictor regression, ensuring it outperforms the historical average, which is not a forecast combination method. Motivated by its core idea, we apply a conservative constraint to OWs, therefore constructing combination weights that have locally equal elements. The weights are computed as $\frac{\mathbf{\delta}}{A} + \frac{1}{N}$, where $\mathbf{\delta}=\mathrm{sign}({\hat{\mathbf{w}^{ow}}-\mathbf{w}^{ew}})$, and $A$ is set to 10/50/100, consistent with that in \cite{li2025dominate}. The sum of ASWs--\cite{li2025dominate}'s weights is not guaranteed to be 1. We point out that the results reported for ASWs--\cite{li2025dominate} with different parameters $A$ in all the tables represent our modified version that adapts its methodology to the forecast combination context, rather than its original approach designed for individual GW predictor forecasts. DMSPE, proposed by \cite{stock2004combination}, employs the weights $\hat{w}_{i,t} = \phi_{i,t}^{-1} / \sum_{j=1}^{N} \phi_{j,t}^{-1}$, where $\phi_{i,t} = \sum_{s=m}^{t-1} \theta^{t-1-s} (r_{s+1} - \hat{r}_{i,s+1})^2$. We set $m=1$ to utilize all the historical data. Following \cite{rapach2010out}, the parameter $\theta$ of DMSPE is set to 0.9 and 1.0, respectively.

Table~\ref{tab:sample division} provides detailed specifications for the key parameters governing our regularization methods. The regularization parameter $\alpha$ controls the strength of the penalty term in the objective function and is selected from a logarithmically spaced grid over the interval $[10^{-6},10^{6}]$ with 500 equally spaced points in log-space. This wide range ensures comprehensive coverage of the parameter space, while the logarithmic spacing provides finer resolution in regions where the parameter is expected to be most effective. For the Elastic Net (ENet) method, the mixing parameter $\beta \in [0,1]$ governs the relative weighting between the $\ell_1$ (Lasso) and $\ell_2$ (Ridge) norm components in the penalty term, where $\beta = 1$ reduces to pure Lasso and $\beta = 0$ to pure Ridge regression. This parameter is selected from a uniform grid over the range $[0.0001, 1.0]$ with an appropriate spacing to ensure adequate coverage of the mixing space while avoiding numerical instability near the boundaries. 

The parameter selection process employs time-series cross-validation to account for the temporal structure of financial data. Specifically, we use expanding window validation where the training sample grows incrementally while maintaining the validation sample size, thereby respecting the chronological ordering of observations and avoiding look-ahead bias. The optimal parameters are selected based on minimizing the mean squared prediction error over the validation period, with ties broken by selecting the most parsimonious model (largest $\alpha$ for given performance).

    \begin{table}[htbp!]
    \centering
    \renewcommand{\arraystretch}{0.8}
    \caption{\textbf{Summary of forecasting and combination methods.} This table summarizes the forecasting and combination methods used in this paper. $T_{c1}$ and $T_{c2}$ denote the training (first-stage) and validation (second-stage) periods, associated with the error covariance matrices $\bm{\Sigma}_1$ and $\bm{\Sigma}_2$, respectively. }
    \label{tab:methods_glossary}
    \small
    \begin{tabularx}{\textwidth}{p{2cm} p{3.2cm} X}
    \toprule
    \textbf{Category} & \textbf{Method} & \textbf{Brief illustration} \\
    \midrule
    Individual regression & GW predictors & Individual predictive regressions based on Goyal-Welch variables. \\
    \hline
    \multirow{9}{*}{\parbox{2cm}{One-step \\ combination}} 
    & OWs & Optimal weights estimated using the combined samples of the training and validation sets ($T_{c1} + T_{c2}$). \\
    & OWs ($\bm{\Sigma}_2$) & Optimal weights estimated solely from the second-stage sample ($T_{c2}$). \\
    & \citet{blanc2020bias} & A shrinkage method following \citet{blanc2020bias} using the combined samples of $T_{c1} + T_{c2}$; see Appendix \ref{sec: appendix detail} for details. \\
    & ASWs--\citet{li2025dominate} & A modified shrinkage combination method based on \citet{li2025dominate} using the combined samples of $T_{c1} + T_{c2}$; see Appendix \ref{sec: appendix detail} for details. \\
    \cmidrule{2-3}
    & EWs & The simple equal-weight heuristic ($1/N$). \\
    & Median & The median values across all individual forecasts. \\
    & Trimmed Mean & A robust average assigning zero weight to forecasts with the minimum and maximum values, while equally weighting the remaining ones. \\
    & DMSPE & Discounted Mean Square Prediction Error: a dynamic weighting scheme using the combined samples ($T_{c1} + T_{c2}$), emphasizing recent performance via a discount factor. \\
    & PCA & Principal Component Analysis extracting dominant factors from GW predictors over the combined samples ($T_{c1} + T_{c2}$). \\
    \midrule
    \multirow{8}{*}{\parbox{2cm}{Two-stage \\ combination}} 
    & 2S--ASWs  \textbf{(Our Method)} &  Two-stage adaptive shrinkage weights: OWs estimated on $T_{c1}$; $\lambda_{opt}$ calibrated on $T_{c2}$. \\
    & Lasso & Stacking combination using Lasso; parameters trained on $T_{c1}$ and tuned on $T_{c2}$. \\
    & Ridge & Stacking combination using Ridge; parameters trained on $T_{c1}$ and tuned on $T_{c2}$. \\
    & ENet & Stacking combination using Elastic Net; parameters trained on $T_{c1}$ and tuned on $T_{c2}$. \\
    \cmidrule{2-3}
    & R2S--ASWs & Reversed 2S--ASWs: OWs estimated on $T_{c2}$; $\lambda_{opt}$ computed on $T_{c1}$. \\
    & R2S--Lasso & Reversed Lasso: parameters trained on $T_{c2}$; tuned on $T_{c1}$. \\
    & R2S--Ridge & Reversed Ridge: parameters trained on $T_{c2}$; tuned on $T_{c1}$. \\
    & R2S--ENet & Reversed Elastic Net: parameters trained on $T_{c2}$; tuned on $T_{c1}$. \\
    \bottomrule
    \end{tabularx}
    \end{table}

    \begin{landscape}
    \renewcommand{\arraystretch}{0.85}
    \begin{longtable}{ccccc}
    \caption{{\bf Sample division and parameters.} This table summarizes the sample division and parameter settings for both individual regressions and combination methods. As for individual regression, the full sample of each GW predictor is divided into $T_1$ and $T_2$, where $T_1$ is used for in-sample estimation and $T_2$ is used to generate out-of-sample forecast series. When $T_1 = 54$ months, it spans January 1996 to June 2000, and $T_2$ covers July 2000 to August 2019; when $T_1 = 42$ months, it spans January 1996 to June 1999, and $T_2$ covers July 1999 to August 2019. As for the combination, the full sample is sequentially split into $T_{c1}$, $T_{c2}$, and $T_{c3}$. $T_{c1}$ and $T_{c2}$ are varied in empirical analysis, while out-of-sample evaluation of combination methods is performed on $T_{c3}$. In in-sample analysis, $T_{c1}$ equals to $T_1$, and the end of $T_{c2}$ is always December 2006. The regularization parameter $\alpha$ applies to the penalty term in Lasso, Ridge, and Elastic Net, while $\beta$ controls the relative weight between the $\ell_1$ and $\ell_2$ norms in Elastic Net.} \label{tab:sample division} \\
    \toprule
    \multicolumn{5}{c}{\textit{Panel A: In-sample Analysis}} \\
    \multicolumn{1}{c}{Method} & \multicolumn{1}{c}{Training} & \multicolumn{1}{c}{Validation} & \multicolumn{1}{c}{Evaluation} & \multicolumn{1}{c}{Candidate Parameters} \\
    \midrule
    GW predictors & $T_1$ & - & $T_1$ & - \\
    2S--ASWs& $T_{c1}$ & $T_{c2}$ & $T_{c2}$ & $\lambda_{opt}$ \\
    Lasso & $T_{c1}$ & $T_{c2}$ & $T_{c2}$ & $\alpha \in [10^{-6}, 10^{6}]$ \\
    Ridge & $T_{c1}$ & $T_{c2}$ &$T_{c2}$ & $\alpha \in [10^{-6}, 10^{6}]$ \\
    ENet & $T_{c1}$ & $T_{c2}$ & $T_{c2}$ & $\alpha \in [10^{-6}, 10^{6}]$, $\beta \in [1e\!-\!4, 0.1, 0.2, ..., 1.0]$ \\
    OWs & $T_{c1}+T_{c2}$ & - & $T_{c1}+T_{c2}$ &  - \\
    \cite{blanc2020bias} & $T_{c1}+T_{c2}$ & - & $T_{c1}+T_{c2}$ & $\lambda_{opt}$ \\
    ASWs--\cite{li2025dominate} & $T_{c1}+T_{c2}$ & - & $T_{c1}+T_{c2}$ & $A=10/50/100$ \\
    EWs & $T_{c1}+T_{c2}$ & - & $T_{c1}+T_{c2}$ & - \\
    Median & $T_{c1}+T_{c2}$ & - & $T_{c1}+T_{c2}$ & - \\
    Trimmed Mean & $T_{c1}+T_{c2}$ & - & $T_{c1}+T_{c2}$ & - \\
    DMSPE & $T_{c1}+T_{c2}$ & - & $T_{c1}+T_{c2}$ & $\theta=0.9/1.0$ \\
    PCA & $T_{c1}+T_{c2}$ & - & $T_{c1}+T_{c2}$ &  - \\
    \midrule
    \multicolumn{5}{c}{\textit{Panel B: Out-of-sample Analysis}} \\
    \multicolumn{1}{c}{Method} & \multicolumn{1}{c}{Training} & \multicolumn{1}{c}{Validation} & \multicolumn{1}{c}{Evaluation} & \multicolumn{1}{c}{Candidate Parameters} \\
    \midrule
    GW predictors & $T_1$ & - & $T_{c2}+T_{c3}$ & - \\
    2S--ASWs& $T_{c1}$ & $T_{c2}$ & $T_{c3}$ & $\lambda_{opt}$ \\
    Lasso & $T_{c1}$ & $T_{c2}$ & $T_{c3}$ & $\alpha \in [10^{-6}, 10^{6}]$ \\
    Ridge & $T_{c1}$ & $T_{c2}$ & $T_{c3}$ & $\alpha \in [10^{-6}, 10^{6}]$ \\
    ENet & $T_{c1}$ & $T_{c2}$ & $T_{c3}$ & $\alpha \in [10^{-6}, 10^{6}]$, $\beta \in [1e\!-\!4, 0.1, 0.2, ..., 1.0]$ \\
    OWs & $T_{c1}+T_{c2}$ & - & $T_{c3}$ & -  \\
    OWs ($\mathbf{\Sigma}_2$) &  $T_{c2}$ & - & $T_{c3}$ & - \\
    \cite{blanc2020bias} & $T_{c1}+T_{c2}$ & - & $T_{c3}$ & $\lambda_{opt}$ \\
    ASWs--\cite{li2025dominate} & $T_{c1}+T_{c2}$ & - & $T_{c3}$ & $A=10/50/100$ \\
    EWs & - & - & $T_{c3}$ & - \\
    Median &  - & - & $T_{c3}$ & - \\
    Trimmed Mean &  - & - & $T_{c3}$ & - \\
    DMSPE & $T_{c1}+T_{c2}$ & - & $T_{c3}$ & $\theta=0.9/1.0$ \\
    PCA & $T_{c1}+T_{c2}$ & - & $T_{c3}$ & - \\
    \bottomrule
    \end{longtable}
    \end{landscape}

    \clearpage

    \section{Additional Results}\label{sec: appendix results}

    \renewcommand{\thesubsection}{C.\arabic{subsection}}
    \renewcommand{\thetable}{C.\arabic{table}}
    \renewcommand{\thefigure}{C.\arabic{figure}}
    \setcounter{figure}{0}  
    \setcounter{subsection}{0}
    \setcounter{table}{0}

This appendix provides additional robustness exercises for the adaptive shrinkage combination method to ensure the reliability and generalizability of our main findings. We examine four critical dimensions: sensitivity to alternative sample partitioning schemes, performance under non-negative weight constraints that reflect practical portfolio management considerations, a reversed specification with interchanging $\mathbf{\Sigma}_2$  and $\mathbf{\Sigma}_1$, and an investigation when the samples are extended to December~2024 using the full 29 predictors updated from \citep{goyal2024comprehensive}.

\subsection{Sensitivity to Sample Partitioning}

Table~\ref{tab:other oos combination} presents out-of-sample combination results under varying training period lengths $T_1$ for individual regression, combined with different specifications for the combination training set $T_{c1}$ and validation set $T_{c2}$. To maintain comparability with our baseline analysis, the overall training-validation samples consistently span January 1996 to December 2006 across all configurations, ensuring a uniform evaluation framework while allowing for systematic variation in the internal partitioning structure.

\bigskip
\centerline{\bf [Place Table~\ref{tab:other oos combination} about here]}
\bigskip

We examine four alternative specifications for the individual regression training period: $T_1 \in \{24, 30, 48, 60\}$ months. These choices represent a comprehensive range from relatively short estimation windows that capture recent market dynamics to longer periods that provide more stable parameter estimates. For each $T_1$, the combination training period $T_{c1}$ is set equal to $T_1$ to maintain consistency, while the validation period $T_{c2}$ adjusts accordingly to ensure the combined sample terminates in December 2006.

The results demonstrate that 2S--ASWs achieve most positive $R^2_{COS}$ values, especially at $J=1$ and many long horizons, and significantly outperform alternative benchmarks. This robustness to different training window lengths underscores the method's ability to adapt to varying amounts of historical information while maintaining forecasting accuracy. The other shrinkage methods (Lasso, Ridge, ENet) generally yield negative $R^2_{COS}$ values in the short term but demonstrate improved performance at longer horizons, suggesting that these methods may be more suitable for medium- to long-term forecasting applications.

OWs exhibit the weakest performance at short horizons ($J=1$ and $J=3$) but display the strongest predictive power at the longest horizon ($J=24$), indicating a potential trade-off between short-term adaptability and long-term stability. EWs and PCA methods demonstrate relatively stable but modest predictive ability across all forecast horizons, providing consistent if unremarkable performance. Notably, the $R^2_{COS}$ statistics for 2S--ASWs exhibit the characteristic ``check-mark'' pattern across forecast horizons $J$ observed in our main analysis, with initial decline usually at $J=6$ followed by recovery at longer horizons, confirming the robustness of this empirical regularity.

\subsection{Non-Negative Weight Constraints}

Table~\ref{tab:nonnegative oos combination} investigates the performance of combination methods under non-negative weight constraints, addressing a practical consideration in portfolio management where negative weights correspond to short positions that may be either unavailable to retail investors or undesirable for institutional practitioners due to regulatory constraints, borrowing costs, or risk management policies. In this specification, non-negativity constraints are imposed exclusively during the estimation of optimal weights, thus ensuring the non-negativity of adaptive shrinkage weights while the application of individual forecasting methods (Lasso, Ridge, ENet, PCA) remains identical to the unconstrained setting. This approach isolates the effect of weight constraints on combination performance while maintaining comparability with our baseline results.

\bigskip
\centerline{\bf [Place Table~\ref{tab:nonnegative oos combination} about here]}
\bigskip

The analysis covers various combinations of sample partitioning parameters ($T_1, T_{c1}, T_{c2}$) to ensure robustness across different estimation window configurations. 2S--ASWs remain competitive, particularly at short forecast horizons, providing additional empirical evidence even under the restricted weighting scenario. The method's ability to maintain strong performance under non-negative constraints suggests that the shrinkage approach naturally tends toward economically sensible weight allocations without requiring explicit restrictions.

OWs, EWs, PCA, and other benchmark methods exhibit predictive performance patterns similar to those observed in our unconstrained analysis, indicating that the non-negativity constraint does not fundamentally alter the relative ranking of methods. However, a notable distinction emerges in the horizon-dependent behavior: under non-negative constraints, there is a pronounced upward trend in $R^2_{COS}$ as the forecast horizon $J$ increases across most methods. This pattern differs markedly from the unconditional case and suggests that weight constraints may have differential effects across forecast horizons, potentially due to the varying stability of individual predictor performance over different time scales.

This finding has important implications for practical implementation, as it suggests that the benefits of imposing economically motivated constraints may be particularly pronounced for longer-horizon forecasting applications. The upward trend in performance with horizon length under constraints may reflect the fact that longer-horizon forecasts are less sensitive to short-term noise and therefore benefit more from the stabilizing effect of non-negative weight restrictions.

\subsection{Reversed Error Covariance Matrices}

The theoretical analysis presented in Appendix~\ref{sec:appendix A.2} reveals that the reversed specification is mathematically equivalent in form to the main two-stage optimization problem: after interchanging the stage roles of $\mathbf{\Sigma}_1$ and $\mathbf{\Sigma}_2$, the same derivations and propositions remain valid. This form equivalence motivates the empirical examination presented in Table~\ref{tab:inverse matrix oos combination}, where we reverse the allocation of error covariance matrices: $\mathbf{\Sigma}_2$ for the first stage, $\mathbf{\Sigma}_1$ for the second. Empirically, $T_{c2}$ is first used to estimate OWs, followed by $T_{c1}$, which is used to compute $\lambda_{opt}$ and the corresponding 2S--ASWs. To maintain a rigorously consistent basis for comparison, we apply this reversal to Lasso, Ridge, and ENet. We train the parameters on the sample associated with $\mathbf{\Sigma}_2$ and calibrate the optimal $\alpha$, $\beta$ on the sample associated with $\mathbf{\Sigma}_1$. To clearly distinguish these reverse-order specifications from standard methods, we denote them as R2S--ASWs, R2S--Lasso, R2S--Ridge, and R2S--ENet, respectively.

\bigskip
\centerline{\bf [Place Table~\ref{tab:inverse matrix oos combination} about here]}
\bigskip

Under the $T_1=54$ specification, R2S--ASWs continue to demonstrate superior performance relative to other methods, particularly excelling at short horizons $J=1$ and $J=6$. Notably, R2S--ASWs exhibit significantly enhanced out-of-sample $R^2_{COS}$ values at longer forecast horizons ($J=12$ and $J=24$) when $\mathbf{\Sigma}_1$ and $\mathbf{\Sigma}_2$ are reversely used. The characteristic ``check-mark'' pattern in $R^2_{COS}$ statistics across forecast horizons remains evident, though it typically occurs at $J=3$ rather than $J=6$, as observed in Table~\ref{tab:oos combination}. However, other shrinkage methods can outperform R2S--ASWs in some cases under this reversed specification, in which Ridge, in particular, together with Lasso and ENet, exhibit a pronounced superiority at the longest forecast horizon $J=24$. EWs, Median, Trimmed Mean, and ASWs--\cite{li2025dominate} methods show robust but relatively limited performance in this reversed configuration. DMSPE with $\theta=0.9$ performs relatively better, particularly from $J=3$ onward. The results under $T_1=42$ reveal a deterioration in general. None of these methods can yield positive $R^2_{COS}$ statistics at $J=1$. Unlike a ``check-mark'' pattern observed with increasing horizon $J$, the $R^2_{COS}$ statistics for R2S--ASWs exhibit a consistently upward trend across $J$, starting from negative values at short horizons and gradually improving to positive values from $J=6$.

Overall, while the forecasting performance under this reversed specification ($\mathbf{\Sigma}_2 - \mathbf{\Sigma}_1$, $T_{c2}-T_{c1}$) is inferior to the results of sequential $T_{c1}-T_{c2}$ partitioning in Table~\ref{tab:oos combination}, it does exhibit notable improvements at long-term horizons for R2S--ASWs. The deterioration may be attributed to the sample allocation. In the main specification, $T_{c1}$ is used to compute OWs while $T_{c2}$ is employed to estimate $\lambda_{opt}$ and derive R2S--ASWs. This arrangement ensures that the optimal shrinkage parameter and adaptive shrinkage weights are determined using more recent data $T_{c2}$, which enhances their relevance and accuracy for out-of-sample forecasting. The temporal proximity of $T_{c2}$ to the forecast evaluation period allows the adaptive shrinkage combination to better capture current market dynamics. In contrast, the reversed specification uses $T_{c2}$ to estimate OWs, which may result in more accurate estimation of OWs, but the critical shrinkage parameter and combination weights are calibrated on the past, potentially less relevant sample $T_{c1}$. This temporal misalignment reduces the timeliness and effectiveness of the adaptive shrinkage combination, as the key shrinkage parameter governing the combination process is based on historical rather than recent information. The surprising gains of R2S--Ridge, R2S--Lasso, and R2S--ENet at $J=24$ suggest that their regularization terms interact non-trivially with the composition of the estimation sample, such that the reversed allocation creates a particularly favorable environment for their shrinkage properties at long horizons. However, the findings of R2S--ASWs' improvement at long-term forecasting highlight the flexibility and stability across different configurations. The ability to enhance performance at longer horizons simply by reversing the roles of $\mathbf{\Sigma}_1$ and $\mathbf{\Sigma}_2$ demonstrates that the two-stage approach provides practitioners with strategic freedom in tailoring the method to specific objectives. By selecting which sample is used for optimal weight estimation versus shrinkage parameter computation, the adaptive shrinkage combination can be configured to deliver strong performance regardless of whether the focus is on short-term or long-term investment horizons.

Figure~\ref{fig:lambda_opt inverse} plots $\lambda_{opt}$ in R2S--ASWs under the reversed error covariance matrices, still with the configuration $T_{c1}=6$ and $T_{c2}=126$ when $T_1=42$. The reversed allocation produces a temporal shift in the distribution of $\lambda_{opt}$ values. Whereas the sequential specification exhibits $\lambda_{opt} < 0.5$ and approaches zero mainly in the latter part of the out-of-sample set, the reversed specification concentrates these lower $\lambda_{opt}$ in the early stages of the evaluation period. The general pattern is consistent with the dominance of $\lambda_{opt} > 0.5$ and approaching one, particularly at horizons $J=1/3/6$, indicating a persistent preference for EWs.

Compared to Figure~\ref{fig:lambda_opt}, the reversed configuration exhibits a higher frequency of $\lambda_{opt} < 0.5$ and instances where $\lambda_{opt} = 0$ (OWs). Across horizons from $J=3$ to $J=24$, there are several temporary periods where R2S--ASWs equal OWs, with this occurrence becoming more prevalent as the forecast horizon $J$ increases. This tendency partially explains the deterioration in R2S--ASWs' overall predictive performance under the reversed specification. Nevertheless, the absolute preponderance of $\lambda_{opt}$ selecting or leaning towards EWs again underscores the robustness and effectiveness of the simple equal-weighted combination strategy across different forecast scenarios.

\bigskip
\centerline{\bf [Place Figure~\ref{fig:lambda_opt inverse} about here]}
\bigskip

\subsection{Extended Sample Analysis}\label{sec:appendix 2024}

Recent evidence from \cite{denk2024predicting} raises questions about the predictive ability of forecast combination methods over time. Examining the U.S. equity premium from 1965 to 2022, \cite{denk2024predicting} documents a temporal difference: combination methods significantly outperform the historical mean benchmark during 1965--1993, but do not provide relative benefits from 1994 to 2022, the more recent decades. 

In light of these findings, we extend our evaluation sample from January 2007 to December 2024, allowing us to assess whether the adaptive shrinkage combination maintains its predictive ability in the recent period. We also update our predictors to include a comprehensive set of 29 monthly variables, as listed in \cite{goyal2024comprehensive}, to ensure the examination incorporates current and the most recent available information. The training and validation samples still range from January 1996 to December 2006. Table~\ref{tab:extend 2024 oos combination} reports the out-of-sample forecast performance of combination methods and other benchmarks.

\bigskip
\centerline{\bf [Place Table~\ref{tab:extend 2024 oos combination} about here]}
\bigskip

Consistent with the observations in \cite{denk2024predicting}, the results reveal a general decline in $R^{2}_{COS}$ statistics of combination methods, compared to Table~\ref{tab:oos combination} when the evaluation sample ends in August 2019. The number of advantageous parameter groups also significantly decreases. The reduction in $R^{2}_{COS}$ values is particularly pronounced for OWs and the method of \cite{blanc2020bias}. For instance, the worst $R^{2}_{COS}$ value for OWs is as low as -70.11\% when $T_1=72$. In contrast, the performance degradation for 2S--ASWs and ASWs--\cite{li2025dominate} is modest and primarily confined to the short-term forecast at $J=1$. Meanwhile, robust methods including EWs, median, and trimmed mean maintain stable performance in the extended evaluation at $J=1$. DMSPE with $\theta=0.9$ underperforms at $J=1$ but emerges as a strong competitor from $J=3$ onward ($R^2_{COS}=4.36\%$ at $J=3$, $11.59\%$ at $J=6$, $27.05\%$ at $J=12$, and $47.80\%$ at $J=24$), consistent with its dominance at intermediate horizons in the baseline results. 

Despite the overall decline in predictability, the adaptive shrinkage combination method continues to demonstrate competitive performance relative to the historical mean and other benchmarks. The advantage of 2S--ASWs is evident at $J=1$, where the method achieves positive out-of-sample $R^{2}_{COS}$ statistics in all settings considered. We emphasize that, as evidenced by evaluations in recent years, different combination methods exhibit distinct comparative advantages across forecast horizons. ASWs--\cite{li2025dominate}'s, together with EWs, median, and trimmed mean combinations, yield positive out-of-sample $R^{2}_{COS}$ statistics across al horizons except at $J=1$. Among these, ASWs--\cite{li2025dominate}'s method with parameter $A=10$ shows strong forecast ability at $J=3$, frequently attaining the highest $R^{2}_{COS}$ values among all competing methods. In contrast, OWs and \cite{blanc2020bias}, while exhibiting almost the weakest performance at short horizons, display remarkable strength at long forecast horizons. At $J=12$ and $J=24$, these methods achieve maximum out-of-sample $R^{2}_{COS}$ exceeding 50\% and 70\%, respectively, and substantially outperform other approaches at these horizons. 

Overall, although the extended sample period verifies that the U.S. equity premium prediction has become more challenging for combination methods in recent years, our empirical results suggest that combination forecasts can still retain meaningful predictive value when appropriately deployed. The divergent predictive patterns across forecast horizons also indicate an important practical consideration: investors can flexibly choose the optimal combination method according to their forecasting objective and investment horizon.

\subsection{Summary of Best-Performing Methods}

Table~\ref{tab:dominance_count} aggregates the out-of-sample dominating frequencies of forecasting methods used in this paper across all the evaluated scenarios (including out-of-sample combination in the main text (Table~\ref{tab:oos combination}), non-negative weight constraints (Table~\ref{tab:nonnegative oos combination}, reversed covariance matrices (Table~\ref{tab:inverse matrix oos combination}), and the extended sample to December 2024 (Table~\ref{tab:extend 2024 oos combination})). A cross-sectional analysis of these results reveals an interesting and horizon-dependent bifurcation in predictive superiority.

\bigskip
\centerline{\bf [Place Table~\ref{tab:dominance_count} about here]}
\bigskip

Panel A demonstrates the unequivocal dominance of the 2S--ASWs method at the short horizon $J=1$. By adaptively shrinking toward the heavily diversified equal-weight (na\"{\i}ve) portfolio, 2S--ASWs protect short-term forecasts from estimation error more effectively than traditional combination methods and other benchmarks. As the forecast horizon extends to the medium term ($J=3$ and $J=6$), the predictive landscape transitions into a shared dominance between 2S--ASWs and several alternative benchmarks. Among them, ASWs--\citet{li2025dominate} emerges as a relatively prominent benchmark, particularly when the sample is extended to December 2024 at $J=3$; at $J=6$, however, the dominance structure becomes dispersed. No single method commands a decisive advantage; instead, dominance is distributed across a diverse set of approaches, including DMSE with $\theta=0.9$, \citet{li2025dominate}, and 2S--ASWs. Within the non-negative weight combination framework, at both $J=3$ and $J=6$, DMSE with $\theta=0.9$ asserts a clear and consistent dominance. At long-term horizons ($J=12$ and $J=24$), as Panel B shows, the dominance is concentrated, where \citet{blanc2020bias} and OWs($\bm{\Sigma}_2$) stand out. In particular, R2S--Ridge performs best when the covariance matrices are reversed. 

From a practical asset allocation perspective, these findings yield clear, horizon-specific directives for investors. At the critical $J=1$ horizon, 2S--ASWs provides a highly investable, theoretically sound mechanism that captures genuine predictability while decisively guarding against the severe out-of-sample estimation errors. Because short-term market timing constitutes the most critical yet challenging component of active management, the robust empirical superiority of 2S--ASWs renders it a useful short-horizon forecasting tool. 

For long-term investment, the strategic emphasis evolves but remains rooted in the power of shrinkage. While traditional optimization or structural bias corrections are often employed to capture slow-moving macroeconomic trends, our findings reveal a potent alternative: the integration of shrinkage with machine learning predictors. Specifically, over extended horizons, the two-stage shrinkage methods, including R2S--ASWs, and especially the R2S--Ridge regression, can enhance the forecasting performance when the covariance matrices are reversed. Therefore, institutional investors with extended holding periods should not merely transition to pure optimization, but rather leverage reversed-shrinkage strategies to extract long-range signals. This dual-horizon efficacy ensures that the methodology remains robust across the entire investment lifecycle, from tactical monthly adjustments to strategic multi-year positioning.

\clearpage
    \begin{landscape}
    \begin{figure}
        \centering
        \subfloat[$J=1$]{
            \includegraphics[width=0.4\textwidth]{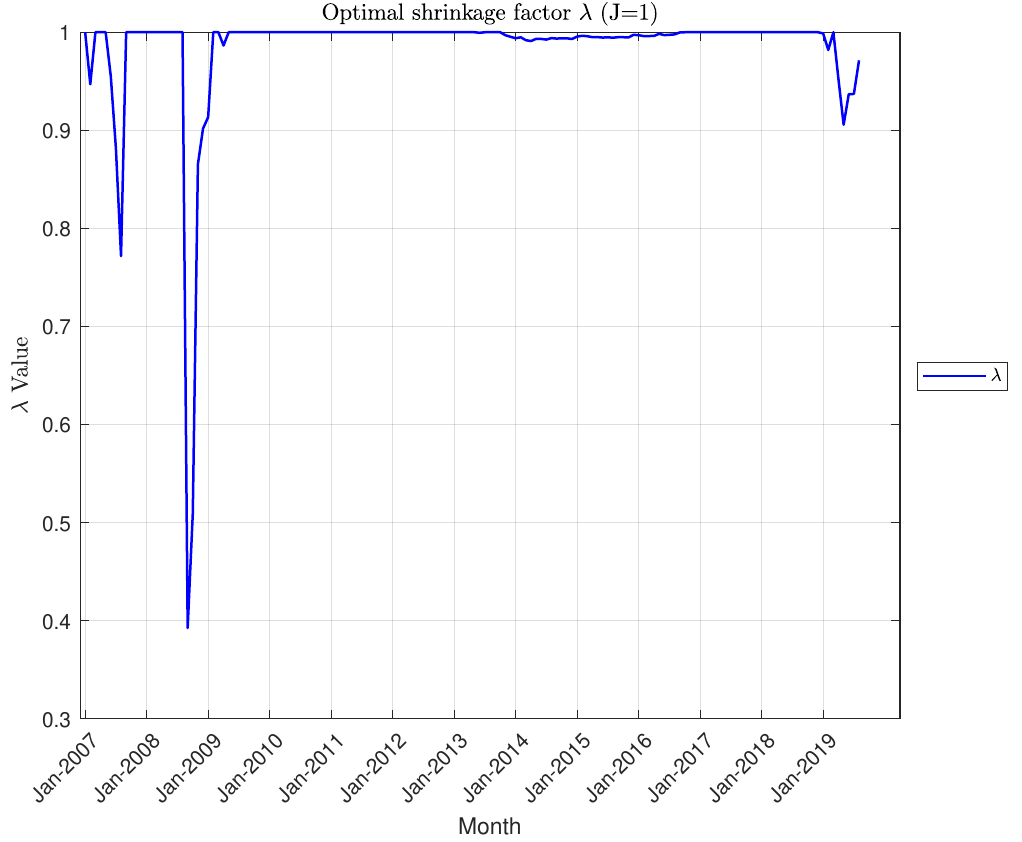}
        }
        \hspace{0.02\textwidth}
        \subfloat[$J=3$]{
            \includegraphics[width=0.4\textwidth]{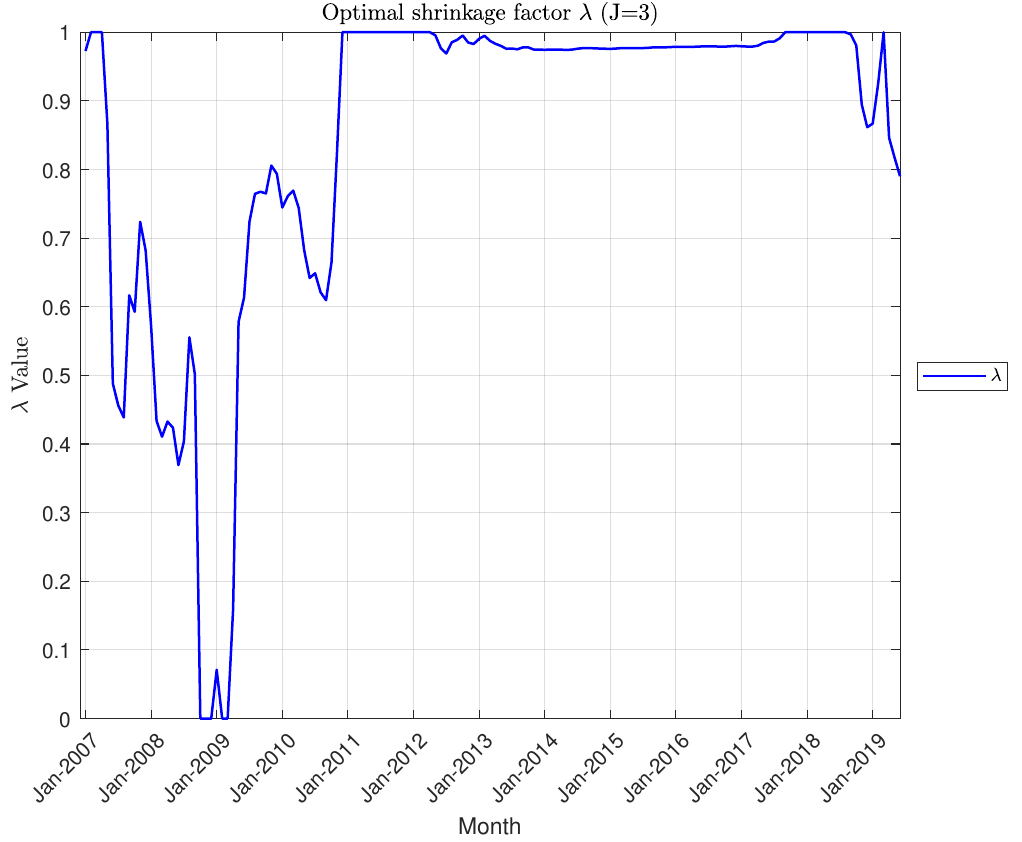}
        }
        \hspace{0.02\textwidth}
        \subfloat[$J=6$]{
            \includegraphics[width=0.4\textwidth]{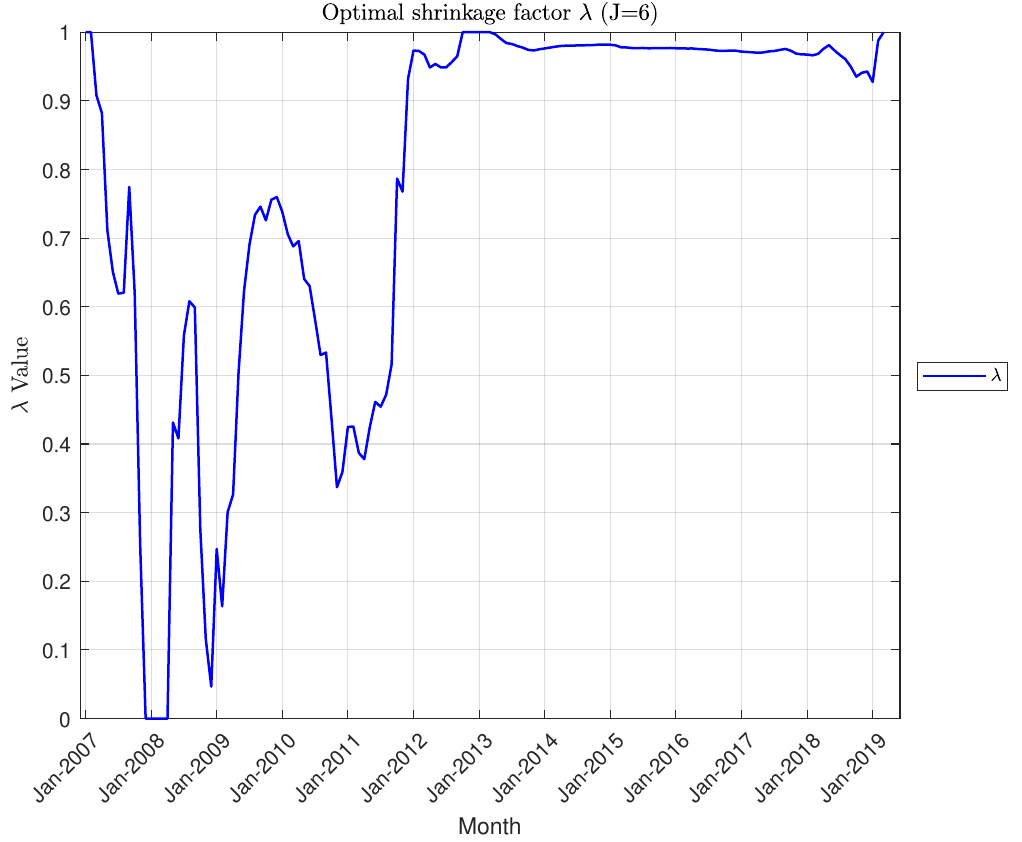}
        }

        \begin{minipage}[t]{0.48\textwidth}
            \centering
            \subfloat[$J=12$]{
                \includegraphics[width=0.88\textwidth]{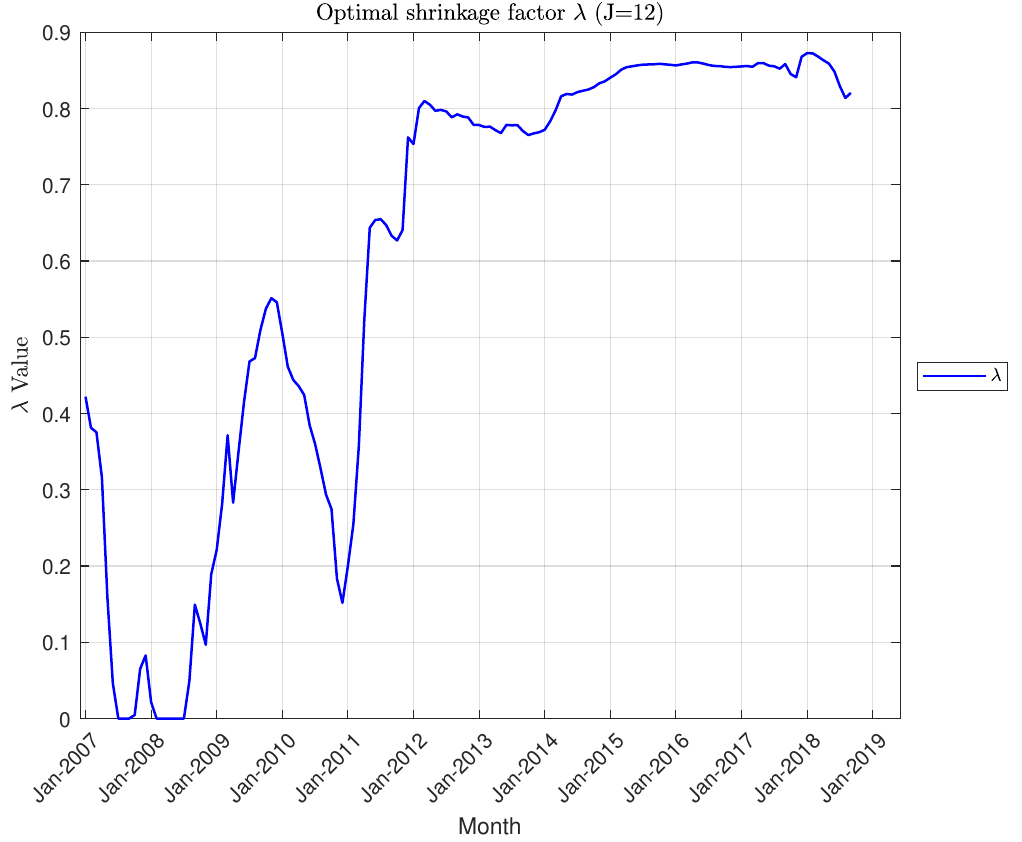}
            }
        \end{minipage}
        \hspace{0.04\textwidth}
        \begin{minipage}[t]{0.48\textwidth}
            \centering
            \subfloat[$J=24$]{
                \includegraphics[width=0.88\textwidth]{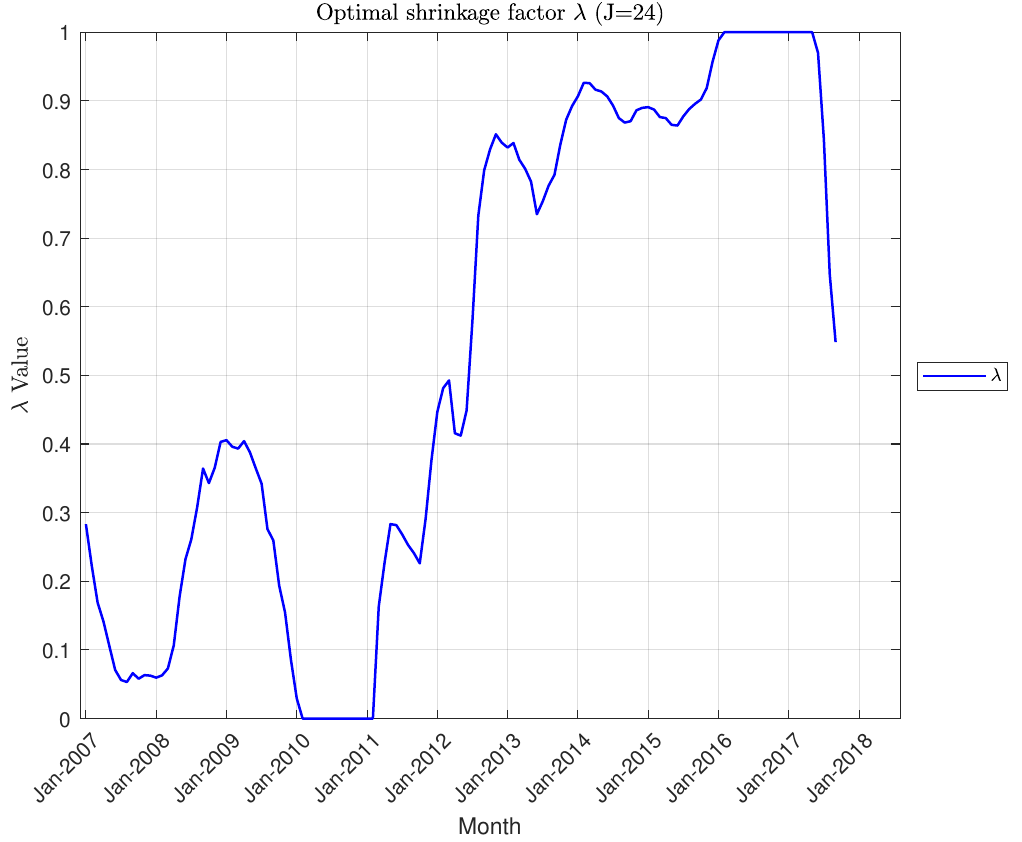}
            }
        \end{minipage}
        \caption{{\bf Optimal shrinkage factor $\lambda_{opt}$ in adaptive shrinkage combination with reversed error covariance matrices ($\mathbf{\Sigma}_2, \mathbf{\Sigma}_1)$ ($T_{c1}=6$, $T_{c2}=126$ when $T_1=42$).} {These figures report the optimal shrinkage factor $\lambda_{opt}$ for $J=1/3/6/12/24$, which is computed in validation set $T_{c1}$ with reversed error covariance matrices: the first-stage error covariance matrix $\mathbf{\Sigma}_2$ and the second-stage error covariance matrix $\mathbf{\Sigma}_1$. If $\lambda_{opt} > 0.5$, 2S--ASWs tend to lean more toward EWs; otherwise, 2S--ASWs tend to lean more toward OWs. The proportions of $\lambda_{opt}=1$ are 44.74\%, 26.00\%, 7.48\%, 1.42\%, 10.08\% and those for $\lambda_{opt}\geq 0.9$ are 85.53\%, 75.33\%, 60.54\%, 2.13\%, 15.50\%, respectively.}}
        \label{fig:lambda_opt inverse}
    \end{figure}
\end{landscape}

\clearpage
    \begin{landscape}
    \small
    \setlength{\tabcolsep}{0.5pt}
    \renewcommand{\arraystretch}{0.85}
\begin{longtable}{
    c                               
    c                                
    *{5}{                            
        S[table-format=-3.2]
        @{\hspace{1.2em}}
        S[table-format=1.3]             
    }
} 
            \caption{{\bf Other out-of-sample combination results.}  To further investigate the consistency of the adaptive shrinkage combination method across various settings, this table presents some other out-of-sample $R^2$ statistics ($R^2_{COS}$) and $p$-values for the \citet{clark2007approximately} MSFE-adjusted statistic of combination methods and PCA, which are not reported in the main text, when $J=1/3/6/12/24$. The training and validation samples together ($T_{c1}+T_{c2}$) range from January 1996 to December 2006, but $T_1$ is set to 24, 30, 48, and 60, respectively. The out-of-sample evaluation period is still from January 2007 to August 2019. $R^2_{COS}$ measures the percent reduction in mean squared forecast error (MSFE) for the predictive regression model at a $J$-month horizon forecast relative to the historical mean benchmark forecast. The $p$-values for \citet{clark2007approximately} MSFE-adjusted statistic test the null hypothesis that the historical average MSFE is less than or equal to the predictive regression MSFE against the alternative that the historical average MSFE is greater than the predictive regression MSFE (corresponding to $H_0$: $R^2_{COS}\leq 0$ against $H_{A}$: $R^2_{COS} > 0$). Asterisks attached to the $R^2_{COS}$ entries report the significance levels of the \citet{harvey1998tests} forecast-encompassing test (HLN), which compares the information content of 2S--ASWs with that of the benchmark method in the same row.}\label{tab:other oos combination} \\
                \toprule
                \multicolumn{1}{c}{Estimation Sample} & \multicolumn{1}{c}{Methods} & \multicolumn{2}{c}{J=1} & \multicolumn{2}{c}{J=3} & \multicolumn{2}{c}{J=6} & \multicolumn{2}{c}{J=12} & \multicolumn{2}{c}{J=24} \\
\end{longtable}
        \end{landscape}

    \clearpage
    \begin{landscape}
    \small
    \setlength{\tabcolsep}{0.5pt}
    \renewcommand{\arraystretch}{0.85}
        \begin{longtable}{
    c                               
    c                                
    *{5}{                            
        S[table-format=-3.2]      
        @{\hspace{1.2em}} 
        S[table-format=1.3]           
    }
} 
            \caption{{\bf Non-negative weight out-of-sample combination results.} To explore the performance of combination methods with all non-negative weights, in which the scenario corresponds to a long position, this table presents the out-of-sample $R^2$ statistics ($R^2_{COS}$) and $p$-values for the \citet{clark2007approximately} MSFE-adjusted statistic of combination methods and PCA, when $J=1/3/6/12/24$. The training and validation samples ($T_{c1}+T_{c2}$) range from January 1996 to December 2006, but $T_1$ is set to 24, 30, 48, and 60, respectively. The out-of-sample evaluation period is still from January 2007 to August 2019. $R^2_{COS}$ measures the percent reduction in mean squared forecast error (MSFE) for the predictive regression model at a $J$-month horizon forecast relative to the historical mean benchmark forecast. The $p$-values for the \citet{clark2007approximately} MSFE-adjusted statistic test the null hypothesis that the historical average MSFE is less than or equal to the predictive regression MSFE against the alternative hypothesis that the historical average MSFE is greater than the predictive regression MSFE (corresponding to $H_0$: $R^2_{COS}\leq 0$ against $H_{A}$: $R^2_{COS} > 0$). Asterisks attached to the $R^2_{COS}$ entries report the significance levels of the \citet{harvey1998tests} forecast-encompassing test (HLN), which compares the information content of 2S--ASWs with that of the benchmark method in the same row.}\label{tab:nonnegative oos combination} \\
                \toprule
                \multicolumn{1}{c}{Estimation Sample} & \multicolumn{1}{c}{Methods} & \multicolumn{2}{c}{J=1} & \multicolumn{2}{c}{J=3} & \multicolumn{2}{c}{J=6} & \multicolumn{2}{c}{J=12} & \multicolumn{2}{c}{J=24} \\
\end{longtable}
        \end{landscape}

    \clearpage
    \begin{landscape}
    \small
    \setlength{\tabcolsep}{0.5pt}
    \renewcommand{\arraystretch}{0.85}
         \begin{longtable}{
    c                               
    c                                
    *{5}{                            
        S[table-format=-3.2]      
        @{\hspace{1.2em}} 
        S[table-format=1.3]           
    }
} 
            \caption{{\bf Out-of-sample combination results with reversed error covariance matrices ($\mathbf{\Sigma}_2$, $\mathbf{\Sigma}_1$).} This table presents the out-of-sample $R^2$ statistics ($R^2_{COS}$) and $p$-values for the \citet{clark2007approximately} MSFE-adjusted statistic of combination methods with the first-stage error covariance matrix $\mathbf{\Sigma}_2$ and the second-stage error covariance matrix $\mathbf{\Sigma}_1$, and PCA when $J=1/3/6/12/24$. The error covariance matrices are reversed from the specification of Table~\ref{tab:oos combination}: $\mathbf{\Sigma}_2$ for the first stage and $\mathbf{\Sigma}_1$ for the second stage, where $T_{c2}$ estimate OWs and $T_{c1}$ computes $\lambda_{opt}$ and R2S--ASWs. $R^2_{COS}$ measures the percent reduction in mean squared forecast error (MSFE) for the predictive regression model at a $J$-month horizon forecast relative to the historical mean benchmark forecast. The $p$-values for \citet{clark2007approximately} MSFE-adjusted statistic test the null hypothesis that the historical average MSFE is less than or equal to the predictive regression MSFE against the alternative that the historical average MSFE is greater than the predictive regression MSFE (corresponding to $H_0$: $R^2_{COS}\leq 0$ against $H_{A}$: $R^2_{COS} > 0$). The out-of-sample evaluation period is from January 2007 to August 2019 when $T_1=54$ and $T_1=42$. Asterisks attached to the $R^2_{COS}$ entries report the significance levels of the \citet{harvey1998tests} forecast-encompassing test (HLN), which compares the information content of 2S--ASWs with that of the benchmark method in the same row.}\label{tab:inverse matrix oos combination} \\
\end{longtable}
        \end{landscape}

    \clearpage
    \begin{landscape}
    \small
    \setlength{\tabcolsep}{0.5pt}
    \renewcommand{\arraystretch}{0.85}
         \begin{longtable}{
    c                               
    c                                
    *{5}{                            
        S[table-format=-3.2]      
        @{\hspace{1.2em}} 
        S[table-format=1.3]           
    }
} 
            \caption{{\bf Out-of-sample combination results with extended sample period (to December 2024).} In light of the question raised by \cite{denk2024predicting}, we extend the evaluation period from January 2007 to December 2024, to assess the predictive ability of combination methods in more recent years. We also update the predictors to a comprehensive set of 29 variables, as listed in \cite{goyal2024comprehensive}. This table presents the out-of-sample $R^2$ statistics ($R^2_{COS}$) and $p$-values for the \citet{clark2007approximately} MSFE-adjusted statistic of combination methods and PCA, when $J=1/3/6/12/24$. The training and validation samples ($T_{c1}+T_{c2}$) still range from January 1996 to December 2006. $R^2_{COS}$ measures the percent reduction in mean squared forecast error (MSFE) for the predictive regression model at a $J$-month horizon forecast relative to the historical mean benchmark forecast. The $p$-values for the \citet{clark2007approximately} MSFE-adjusted statistic test the null hypothesis that the historical average MSFE is less than or equal to the predictive regression MSFE against the alternative hypothesis that the historical average MSFE is greater than the predictive regression MSFE (corresponding to $H_0$: $R^2_{COS}\leq 0$ against $H_{A}$: $R^2_{COS} > 0$). Asterisks attached to the $R^2_{COS}$ entries report the significance levels of the \citet{harvey1998tests} forecast-encompassing test (HLN), which compares the information content of 2S--ASWs with that of the benchmark method in the same row.}\label{tab:extend 2024 oos combination} \\
\end{longtable}
        \end{landscape}

\clearpage
\begin{landscape}
\begin{table} 
  \centering
  \caption{\textbf{Dominating methods and frequencies across scenarios.} This table presents the dominating forecasting methods and their out-of-sample dominating frequencies (in parentheses) for different forecast horizons ($J=1, 3, 6, 12, 24$) across various scenarios. Panel A reports the counting results for short forecast horizons $J=1, 3, 6$.}
  \label{tab:dominance_count}
  \renewcommand{\arraystretch}{1.5} 
  \small 

  \begin{tabular}{l ccc}
    \multicolumn{4}{c}{\textbf{Panel A: Short/Medium Forecast Horizons}} \\
    \toprule
    Scenarios & $J=1$ & $J=3$ & $J=6$ \\
    \midrule
    \makecell[l]{Out-of-sample combinations \\ (Table~\ref{tab:oos combination})} & 2S--ASWs (8) & DMSPE, $\theta=0.9$ (8) & \makecell{DMSPE, $\theta=0.9$ (7) \\ 2S--ASWs (1)} \\
    \cmidrule(l){2-4}
    \makecell[l]{Other out-of-sample combinations \\ (Table~\ref{tab:other oos combination})} & \makecell{2S--ASWs (4) \\ ENet (1)} & \makecell{2S--ASWs (2) \\ Lasso (2) \\ DMSPE, $\theta=0.9$ (1)} & DMSPE, $\theta=0.9$ (5) \\
    \cmidrule(l){2-4}
    \makecell[l]{Non-negative combinations \\ (Table~\ref{tab:nonnegative oos combination})} & 2S--ASWs (18) & DMSPE, $\theta=0.9$ (18)  & \makecell{DMSPE, $\theta=0.9$ (16) \\  2S--ASWs (1) \\ \citet{li2025dominate}, A=10 (1)}  \\
    \cmidrule(l){2-4}
    \makecell[l]{Reversed combinations \\ (Table~\ref{tab:inverse matrix oos combination})} & \makecell{R2S--ASWs (4) \\ Trimmed Mean (4)} & \makecell{DMSPE, $\theta=0.9$ (4) \\ R2S--ASWs (3) \\ R2S--Ridge (1) } & \makecell{R2S--ASWs (5) \\ \citet{li2025dominate}, A=10 (3)} \\
    \cmidrule(l){2-4}
    \makecell[l]{Combinations extended to Dec. 2024 \\ (Table~\ref{tab:extend 2024 oos combination})}  & \makecell{2S--ASWs (8) \\ Ridge (1)} & \makecell{\citet{li2025dominate}, A=10 (8) \\ 2S--ASWs (1)} & \makecell{\citet{blanc2020bias} (3) \\ \citet{li2025dominate}, A=10 (2) \\ 2S--ASWs (1) \\ Lasso (1) \\ OWs ($\bm{\Sigma}_2$) (1) \\ ENet (1)} \\
    \bottomrule
  \end{tabular}
\end{table}
\end{landscape}

\clearpage

\addtocounter{table}{-1}
\begin{landscape}
\begin{table}
  \centering
  \caption{\textbf{Dominating methods and frequencies across scenarios (continued).} Panel B reports the counting results for long horizons $J=12, 24$.}
  \renewcommand{\arraystretch}{1.5} 
  \small 

  \begin{tabular}{l cc}
    \multicolumn{3}{c}{\textbf{Panel B: Long Forecast Horizons}} \\
    \toprule
    Scenarios & $J=12$ & $J=24$ \\
    \midrule
    \makecell[l]{Out-of-sample combinations \\ (Table~\ref{tab:oos combination})} & \citet{blanc2020bias} (8) & \makecell{OWs($\bm{\Sigma}_2$) (7) \\ \citet{blanc2020bias} (1)} \\
    \cmidrule(l){2-3}
    \makecell[l]{Other out-of-sample combinations \\ (Table~\ref{tab:other oos combination})} & \makecell{\citet{blanc2020bias} (4) \\ OWs ($\bm{\Sigma}_2$) (1)} &  OWs ($\bm{\Sigma}_2$) (5) \\
    \cmidrule(l){2-3}
    \makecell[l]{Non-negative combinations \\ (Table~\ref{tab:nonnegative oos combination})} & \citet{blanc2020bias} (18) & \citet{blanc2020bias} (18) \\
    \cmidrule(l){2-3}
    \makecell[l]{Reversed combinations \\ (Table~\ref{tab:inverse matrix oos combination})} & \makecell{\citet{blanc2020bias} (5) \\ R2S--ASWs (3)} & \makecell{R2S--Ridge (5) \\ OWs ($\bm{\Sigma}_2$) (2) \\ \citet{blanc2020bias} (1)} \\
    \cmidrule(l){2-3}
    \makecell[l]{Combinations extended to December~2024 \\ (Table~\ref{tab:extend 2024 oos combination})} & \makecell{OWs ($\bm{\Sigma}_2$) (8) \\ \citet{blanc2020bias} (1)} & \makecell{OWs ($\bm{\Sigma}_2$) (7) \\ \citet{blanc2020bias} (2)} \\
    \bottomrule
  \end{tabular}
\end{table}
\end{landscape}

    \end{document}